\documentclass[twocolumn,prx,longbibliography,nofootinbib, superscriptaddress]{revtex4-2}
\usepackage{mathrsfs}
\usepackage[dvips]{graphicx} % for figures
\usepackage{amsfonts}
\usepackage{amssymb}
\usepackage{amscd}
\usepackage{amsmath}    % need for subequations	
\usepackage{amsthm}
\usepackage{appendix}
\usepackage{bbm}
\usepackage{dsfont}
\usepackage{mathptmx}
\usepackage{enumerate}
\usepackage{enumitem}
\usepackage{epsfig}
\usepackage{float}
\usepackage[colorlinks=true,linkcolor=teal,citecolor=teal,urlcolor=teal]{hyperref}

\usepackage{makecell}
\usepackage{times}
\usepackage{pifont}
\usepackage{subfigure}
\usepackage{subfloat}
\usepackage{xcolor}		
\usepackage{physics}
\usepackage[most]{tcolorbox}
\usepackage{tabularx}
\usepackage{MnSymbol}
\usepackage{soul}
\usepackage{tikz}
\usepackage{tikzpeople}
\usepackage{pgfplots}
\usepackage[normalem]{ulem}
\usetikzlibrary{quantikz}
\usetikzlibrary{arrows}
\usetikzlibrary{shapes,fadings,snakes}
\usetikzlibrary{decorations.pathmorphing,patterns}
\usetikzlibrary{calc}
\usetikzlibrary{positioning}
\usepackage[linesnumbered,ruled,vlined]{algorithm2e}

\newtheorem{theorem}{Theorem}
\newtheorem{fact}{Fact}
\newtheorem{lemma}{Lemma}
\newtheorem{corollary}{Corollary}

\newtheorem{remark}{Remark}

\newcommand{\bZ}{\mathbb{Z}}
\newcommand{\bE}{\mathbb{E}}

\newcommand{\ba}{\mathbf{a}}
\newcommand{\bs}{\mathbf{s}}
\newcommand{\cO}{\mathcal{O}}
\newcommand{\Var}{\mathrm{Var}}
\newcommand{\sym}{\mathrm{sym}}

\newcommand{\poly}{\mathrm{poly}}
\newcommand{\diag}{\mathrm{diag}}

\usepackage[most]{tcolorbox}
\newtcolorbox[auto counter]{mybox}[2][]{
	enhanced,
	breakable,
	colback=blue!5!white,
	colframe=blue!75!black,
	fonttitle=\bfseries,
	title=Box \thetcbcounter: #2,#1
}

\begin{document}
\title{Optimal randomized measurements for a family of non-linear quantum properties}

\author{Zhenyu Du}
\thanks{ZD and YT contributed equally to this work.}
\affiliation{Center for Quantum Information, Institute for Interdisciplinary Information Sciences, Tsinghua University, Beijing 100084, China}

\author{Yifan Tang}
\thanks{ZD and YT contributed equally to this work.}
\affiliation{Dahlem Center for Complex Quantum Systems, Freie Universität Berlin, 14195 Berlin, Germany}

\author{Andreas Elben}
  \affiliation{Laboratory for Theoretical and Computational Physics and ETHZ-PSI Quantum Computing Hub, Paul Scherrer Institute, CH-5232 Villigen-PSI, Switzerland}

\author{Ingo Roth}
\affiliation{Quantum Research Center, Technology Innovation Institute, Masdar City, Abu Dhabi, United Arab Emirates}

\author{Jens Eisert}
\affiliation{Dahlem Center for Complex Quantum Systems, Freie Universität Berlin, 14195 Berlin, Germany}

\author{Zhenhuan Liu}
\thanks{\href{mailto:qubithuan@gmail.com}{qubithuan@gmail.com}}
\affiliation{Center for Quantum Information, Institute for Interdisciplinary Information Sciences, Tsinghua University, Beijing 100084, China}

\begin{abstract}
Quantum learning encounters fundamental challenges when estimating non-linear properties, owing to the inherent linearity of quantum mechanics. 
Although recent advances in single-copy randomized measurement protocols have achieved optimal sample complexity for specific tasks like state purity estimation, generalizing these protocols to estimate broader classes of non-linear properties without sacrificing optimality remains an open problem. 
In this work, we introduce the observable-driven randomized measurement (ORM) protocol enabling the estimation of ${\rm Tr}(O\rho^2)$ for an arbitrary observable $O$---an essential quantity in quantum computing and many-body physics. 
% ORM achieves this by decomposing the observable $O$ into dichotomic observables and extracting the information of each eigenspace through randomized measurements with block-diagonal unitaries. 
We establish an upper bound for ORM's sample complexity and show its optimality for observables with a large trace-norm, including Pauli and local observables, closing a gap in the literature. 
For these observables, ORM admits an efficient implementation with Clifford circuits.
% \revise{Furthermore, we develop a Pauli sampling protocol using ORM for various physical observables, including Pauli observables, local observables, and stabilizer fidelities, with efficient Clifford-circuit implementations.}
Numerical experiments validate that ORM requires substantially fewer state samples to achieve the same precision compared to classical shadows.
Additionally, we introduce a braiding randomized measurement protocol for multiple low-rank non-linear observables, reducing circuit complexities in practical applications.
\end{abstract}

\maketitle

\section{Introduction}
Quantum learning is pivotal for both verifying and extracting information from quantum systems, addressing the growing demands of quantum technologies while at the same time enabling transformative applications~\cite{Eisert2020certification, elben2019toolbox,Kliesch2021Certification,Cerezo2022qml,Eisert2025Mind}. To date, by far the majority of quantum learning theory has focused on the estimation of linear properties of quantum states.
However, non-linear properties, like the state purity $\Tr(\rho^2)$ and the non-linear expectation value $\Tr({O\rho^2})$, play vital roles 
in quantum simulation for 
virtual cooling~\cite{adam2016thermalization, cotler2019cooling}, quantum error mitigation~\cite{cai2023qem,hugginsVirtualDistillationQuantum2021,koczor2021exponential}, and instances of quantum algorithms~\cite{Huang2022QuantumAdvantage,Lloyd2014qpca}.
Specifically, estimating $\Tr(O\rho^2)/\Tr(\rho^2)$ yields estimates of a quantum system at a lower temperature than the physically prepared thermal state, thereby suppressing potential incoherent noise. % in noisy intermediate-scale quantum devices. 
Further applications include inferring 
entropic quantities and
the quantum Fisher information of an unknown quantum state~\cite{rath2021Fisher,yu2021Fisher} and detecting quantum phases of matter in mixed states~\cite{lee2023decoherence,Lee2025symmetryprotected,zhang2025probingmixedstatephases}.

Yet, the intrinsic linearity of quantum mechanics renders the estimation of non-linear properties considerably more challenging than that of a linear expectation value $\Tr(O\rho)$ of an observable $O$. 
For linear estimation, directly measuring single state copies in the observable's eigenbasis has optimal sample complexity in the Hilbert space dimension $d$.  
While quantum memory enables entangled measurements across multiple state copies, such resources only offer exponential advantage in more complex tasks~\cite{aharonov2022quantum,Huang2022QuantumAdvantage,chen2022memory,gong2024sample}.

In contrast, when estimating a non-linear property like $\Tr(O\rho^2)$, directly applying observable-driven techniques from linear estimation unavoidably requires coherently-controlled quantum memory~\cite{Barenco1997Stabilization, Buhrman2001Quantum,ekert2002direct,Islam2015entropy,Zhou2024hybrid,liu2024auxiliary,miller2024experimental,Yao2025nonlinear} and renders implementations on quantum devices challenging~\cite{Zhou2024hybrid,liu2024auxiliary,miller2024experimental}.  
Existing protocols using only single-copy quantum operations instead resort to general informationally complete measurements independent 
of the observable~\cite{elben2020mix,Vitale2024Robust} and post-process tensor products of 
so-called \emph{classical shadows} (CS) of the state~\cite{huang2020predicting,bertoni2022shallow}. 
Current 
lower bounds $\Omega(\sqrt{d})$, where $d$ is the dimension of Hilbert space, suggest that these protocols consuming $\cO(d)$ state copies are indeed sub-optimal~\cite{Huang2022QuantumAdvantage}. 
Purity estimation can be optimally implemented using \emph{randomized measurements}  (RM)~\cite{van2012Measuring,Brydges2019Probing,elben2019toolbox,elben2023randomized,cieslinski2023analysing, aharonov2022quantum,chen2022memory,Ohliger2013efficient}. 
How to optimally estimate $\Tr(O\rho^2)$ given a general observable $O$ from single state copies---a non-linear analog of directly measuring $O$---is in fact an open question. 

In this work, we answer this question by proposing an observable-driven protocol and establishing its optimality for a wide range of observables in the single-copy regime.
Our protocol generalizes the optimal RM protocol for purity estimation to the non-linear expectation value $\Tr(O\rho^2)$ of arbitrary observables $O$. 
The RM ensemble incorporates information about the target observable $O$.
As shown in Fig.~\ref{fig:framework}, our \emph{observable-driven randomized measurement (ORM)} protocol begins by decomposing the target observable into a few dichotomic observables, each of which has only two distinct eigenspaces. 
For each dichotomic component, we design an RM protocol that employs block-diagonal random unitary evolutions to estimate its non-linear expectation value.
Our analysis shows that for general observables with bounded operator norm on a quantum system of dimension $d$, the ORM protocol exhibits a sample complexity scaling of $\cO(\sqrt{d})$. 
Importantly, this scaling matches the fundamental lower bound for a broad class of observables, including all Pauli and local observables, across all single-copy schemes~\cite{Huang2022QuantumAdvantage, Ye2025ExponentialAdvantageReplica}, thereby is provably optimal.
The ORM protocol also offers significant practical advantages. It is sufficient to use a constant number of measurement bases for constant additive error and failure rate, independent of the system size or the target observable. Furthermore, the protocol features exceptionally simple post-processing, alleviating the need for precise, gate-by-gate calibration of the unitary operations.

\begin{figure}[!htbp]
\centering
\includegraphics[width=.98\linewidth]{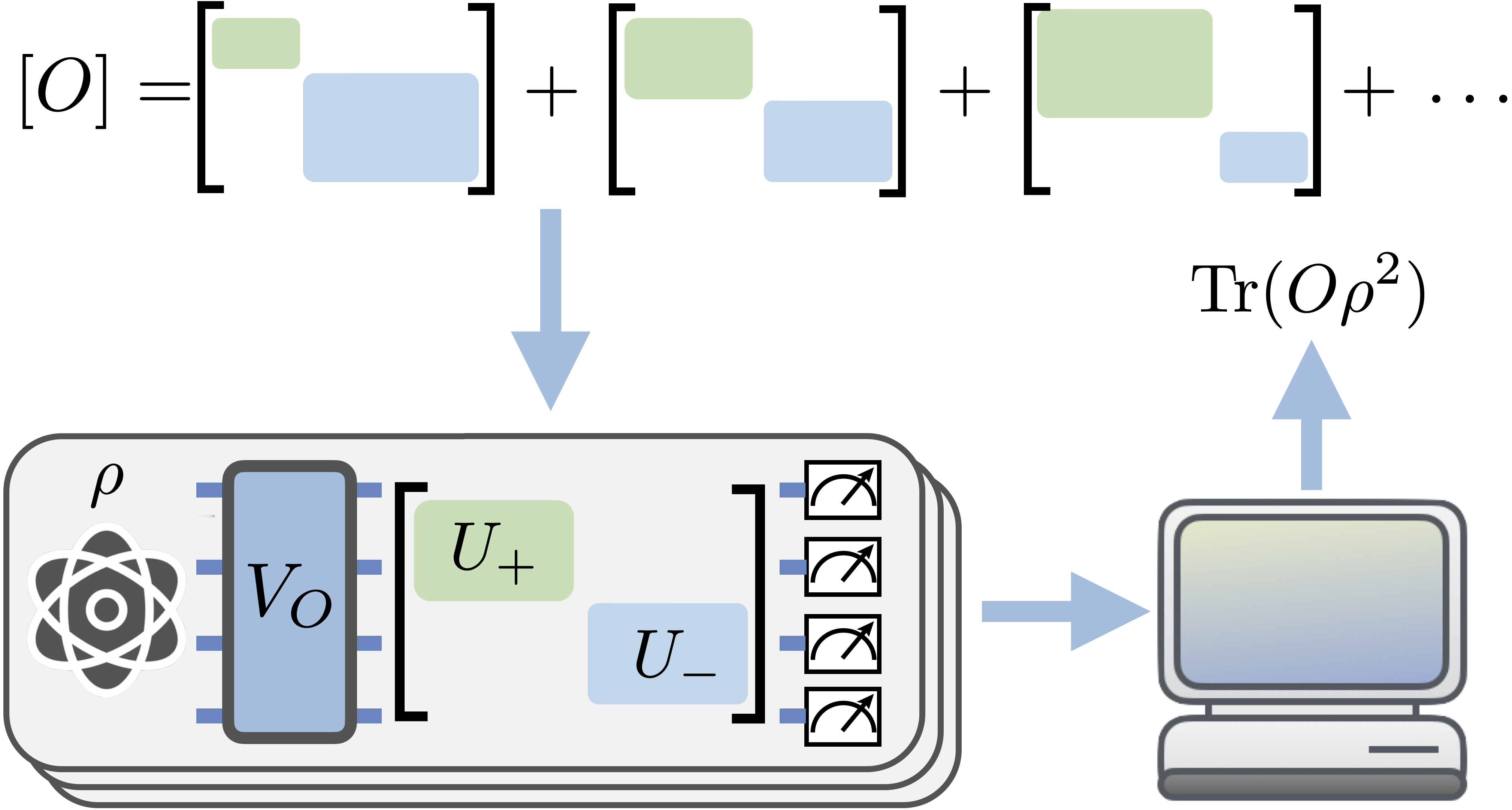}
\caption{The framework of observable-driven randomized measurement protocol. To estimate $\Tr(O\rho^2)$ for arbitrary observables $O$, the ORM protocol first decomposes $O$ into a few dichotomic observables.
For each dichotomic observable, one rotates the target state into the eigenbasis of the observable, evolves it with a block-diagonal random unitary, measures in the computational basis, and post-processes the measurement results to estimate its non-linear expectation value.}
\label{fig:framework}
\end{figure}

To enhance the practicality of our ORM protocol, we construct more efficient circuit implementations for practical observables. 
We first simplify the protocol for Pauli observables by introducing mid-circuit measurements and subsystem random unitaries, and propose a variant requiring only single-qubit gates. 
Built on this, we develop an alternative Clifford-circuit implementation that also achieves the optimal $\mathcal{O}(\sqrt{d})$ sample complexity. 
We further extend our analysis to other unitary ensembles—such as approximate designs, shallow circuits, and local unitaries—to establish the protocol's robustness. 
Finally, via a Pauli sampling protocol, we generalize these efficient primitives to a broad class of physical observables. 
For local Hamiltonians and stabilizer fidelities, this approach preserves the $\cO(\sqrt{d})$ sample complexity while significantly reducing circuit complexities.

Additionally, for constant-rank observables, we propose the \emph{braiding randomized measurement} (BRM) protocol, which simultaneously combines the strengths of randomized measurement and classical shadows. 
Specifically, BRM achieves the same sample complexity scaling as both classical shadows and ORM for constant-rank observables. Moreover, it significantly reduces the number of required measurement bases compared to classical shadows and offers more favorable scaling when simultaneously estimating multiple observables, compared to ORM.

This work is organized as follows. Section~\ref{sec:pre} revisits the RM protocol for purity estimation. Section~\ref{sec:observable_driven} introduces the ORM protocol and analyzes its sample complexity. Section~\ref{sec:efficient_circuit} further simplifies the circuit implementation of our protocols for a broad class of observables. We first present a version tailored to Pauli observables that admits a Clifford-circuit implementation, and then discuss extensions using approximate unitary designs, low-depth circuits, and local measurements. We also propose a Pauli sampling scheme that enables simple implementations for estimating local observables and stabilizer fidelities. 
Section~\ref{sec:applications} demonstrates applications of ORM to error mitigation, mixed-phase detection, and the estimation of entropic quantities; in particular, we numerically apply ORM to a virtual cooling task and show its advantage over classical shadows. Section~\ref{sec:low_rank} introduces the braiding randomized measurement (BRM) protocol for estimating low-rank observables. Finally, Section~\ref{sec:outlook} summarizes our main results and outlines directions for future research.
Our results are summarized in Tables~\ref{tab:CS_ORM} and \ref{tab:low_rank}.

\begin{table*}[htbp]
\centering
\begin{tabular}{c|c|c|c|c|c|c}
\hline 
\hline
 & Observables & Sample complexity & Measurement bases & Random unitary & \makecell{Information of unitary \\ in postprocessing?} & \makecell{Estimating $M$\\ independent observables}\\
\hline
CS  & $\norm{O}_{\infty} \le 1$ &  $\cO(\sqrt{d\Tr(O^2)})$ & $\cO(\sqrt{d\Tr(O^2)})$ & Unitary 3-design & Yes & $\cO(\log(M))$\\
\hline
ORM  & $\norm{O}_{\infty} \le 1$ &  $\cO(\sqrt{d})$ & $\cO(1)$ & \makecell{Block-diagonal \\unitary 4-design} & No & $\cO(M)$\\ 
\hline
ORM  & \makecell{Pauli observables, \\ local observables, \\ and stabilizer fidelities}  &  $\cO(\sqrt{d})$ & $\cO(1)$ & Clifford & No & $\cO(M)$\\
\hline \hline
\end{tabular}
\caption{The comparison of classical shadows and ORM in estimating $\Tr(O\rho^2)$ with constant additive error and failure probability.}\label{tab:CS_ORM}
\end{table*}

\begin{table*}[htbp]
\centering
\begin{tabular}{c|c|c|c|c|c|c}
\hline 
\hline
& Observables & Sample complexity & Measurement bases & Random unitary & \makecell{Information 
of unitary \\ in postprocessing?} & \makecell{Estimating $M$\\ independent observables}\\
\hline
CS  & $\Tr(O^2) = \cO(1)$ & $\cO(\sqrt{d})$ & $\cO(\sqrt{d})$ & Unitary 3-design & Yes & $\cO(\log(M))$\\
\hline
ORM  & $\Tr(O^2) = \cO(1)$ & $\cO(\sqrt{d})$ & $\cO(1)$ & \makecell{Block-diagonal\\ unitary 4-design} & No & $\cO(M)$\\
\hline
ORM  & Stabilizer fidelities & $\cO(\sqrt{d})$ & $\cO(1)$ & Clifford & No & $\cO(M)$\\
\hline
BRM  & $\Tr(O^2) = \cO(1)$ & $\cO(\sqrt{d})$ & $\cO(1)$ & Unitary 6-design & Yes & $\cO(\log(M))$\\
\hline \hline
\end{tabular}
\caption{The comparison of classical shadows, ORM, and BRM protocols in estimating low-rank non-linear observables with constant additive error and failure probability.}\label{tab:low_rank}
\end{table*}

\section{Revisit: Randomized measurement protocol for purity estimation}\label{sec:pre}
Randomness plays a crucial role in quantum information processing, such as in quantum sampling advantage~\citep{Arute2019supermacy,hangleiter2023sampling}, device benchmarking~\cite{knill2008benchmarking,helsen2022general}, noise suppression~\citep{wallmanNoiseTailoringScalable2016,caiConstructingSmallerPauli2019,tsubouchi2024symmetric}, efficient quantum state tomography~\cite{Ohliger2013efficient}, and quantum cryptography~\citep{ji2018pseudo}. RM has emerged as a key tool in quantum learning~\citep{van2012Measuring,elben2019toolbox,elben2023randomized,cieslinski2023analysing}, usually offering highly attractive performance compared to conventional tomography-based methods, while avoiding resource-intensive operations such as quantum memory or deep quantum circuits.

To motivate our method, we first revisit how to use an RM protocol to estimate the purity of a quantum state $\rho$~\cite{Brydges2019Probing,elben2019toolbox}. 
Throughout this work, we assume $\rho$ is an $n$-qubit $d$-dimensional quantum state.
The main idea is to perform repeated measurements on multiple random bases and then post-process the measurement outcomes. 
Specifically, we draw $N_U$ independent random unitaries from a suitable random unitary ensemble (often taken to be a unitary design~\cite{Designs}, see Appendix~\ref{app:random_unitaries}) to rotate the measurement basis. 
For each choice of $U$, we apply it to $\rho$, generating $U \rho U^{\dagger}$, and then perform $N_M$ computational basis measurements on the evolved state. 
Each measurement produces an $n$-bit outcome string $s_i$, and we denote the outcomes of all $N_M$ measurements as $\{s_i\}_{i=1}^{N_M}$.
Subsequently, we perform classical post-processing on these measurement outcomes to estimate the purity. 
When the random unitary is sampled from a unitary two-design, the post-processing function is given by
\begin{equation}\label{eq:RRM_coef_purity}
    X_2(s_1,s_2) = 
    \begin{cases} 
        d, & s_1 = s_2,\\
        -1, & s_1 \neq s_2,
    \end{cases}
\end{equation}
and the unbiased estimator for $\Tr(\rho^2)$ is chosen as
\begin{equation}\label{eq:purity_estimator}
\hat{M}_2^U=\binom{N_M}{2}^{-1}\sum\limits_{1 \le i < j \le N_M}X_2(s_i,s_j).
\end{equation}
We can further average the estimators with all $N_U$ choices of unitaries $\{U_u\}_{u=1}^{N_U}$ to obtain 
\begin{equation}
\hat{M}_2=\frac1{N_U}\sum\limits_{u=1}^{N_U}\hat{M}_2^{U_u}.
\end{equation}
Each $\hat{M}_2^U$ is an unbiased estimator of $\Tr(\rho^2)$~\cite{elben2019toolbox}.
Moreover, if the random unitary ensemble forms a unitary 4-design, then the variance of the estimator $\hat{M}_2$ scales roughly as ${d}/{(N_UN_M^2)}$~\cite{elben2018renyi,anshu2022distributed}. 
Thus, by setting $N_M = \cO(\sqrt{d})$ and $N_U=\cO(1)$, the purity can be estimated with constant additive error, saturating the lower 
bound on sample complexity for purity estimation using single-copy operations ~\cite{aharonov2022quantum,chen2022memory}. 
Another feature of the RM protocol is that, while the number of measurements in a single basis $N_M$ scales exponentially with the system 
size, selecting a constant number of measurement bases $N_U$ suffices for accurate state purity estimation~\cite{elben2019toolbox,Liu2025CertifyingQuantumTemporal}. 
Moreover, its post-processing does not require knowledge of the applied unitaries, providing advantages over classical-shadow-based quantum learning protocols in terms of computational cost for post-processing and robustness against miscalibration.

The core idea of making $\hat{M}_2^U$ an unbiased estimator of $\Tr(\rho^2)$ is that the post-processing function $X_2(s_1,s_2)$ in Eq.~\eqref{eq:RRM_coef_purity} virtually measures the SWAP observable between two copies of $\rho$~\cite{elben2019toolbox}. Specifically, we define a diagonal observable 
\begin{equation}
X_2=\sum\limits_{s_i,s_j}X_2(s_i,s_j)\ketbra{s_i,s_j}{s_i,s_j}
\end{equation}
and draw $U$ from a unitary 2-design ensemble, we have 
\begin{equation}
\underset{U}{\mathbb{E}}\left[U^{\dagger\otimes 2}X_2U^{\otimes 2}\right]=\mathbb{S}, 
\end{equation}
where $\mathbb{S}$ is the SWAP operator. 
Following this line of thought, we get
\begin{equation}\label{eq:purity_expectation}
\begin{aligned}
&\underset{U,\{s_i\}_i}{\mathbb{E}}\left[\hat{M}_2^U\right]\\
=&\underset{U}{\mathbb{E}}\left[\sum\limits_{s_i,s_j}X_2(s_i,s_j)\bra{s_i}U\rho U^\dagger\ket{s_i}\bra{s_j}U\rho U^\dagger\ket{s_j}\right]\\
=&\underset{U}{\mathbb{E}}\left[\sum\limits_{s_i,s_j}X_2(s_i,s_j)\Tr\left(\ketbra{s_i,s_j}{s_i,s_j}U^{\otimes 2}\rho^{\otimes 2} U^{\dagger \otimes 2}\right)\right]\\
=&\underset{U}{\mathbb{E}}\left[\Tr\left(U^{\dagger\otimes 2}X_2U^{\otimes 2}\rho^{\otimes 2}\right)\right]=\Tr(\mathbb{S}\rho^{\otimes 2})=\Tr(\rho^2).
\end{aligned}
\end{equation}
The construction of the RM protocol utilizes the property of purity that it is invariant under the action of an arbitrary unitary operation, i.e., $[\mathbb{S},U^{\otimes 2}]=0$.
Thus, the measurement statistics on a random basis reflect the information of state purity.
Following this logic, by incorporating the properties of the target non-linear observable, we can design an RM protocol to estimate more general state functions.

\section{Observable-driven randomized measurement protocol} \label{sec:observable_driven}
In this section, we present our main result: an observable-driven randomized measurement protocol to estimate the non-linear quantity $\Tr(O\rho^2)$ for arbitrary observables $O$. As the name suggests, the ORM protocol uses random unitaries drawn from an ensemble specific to the target observable $O$. 
We prove that ORM accurately estimates $\Tr(O\rho^2)$ with high probability using only $\cO(\sqrt{d})$ copies of $\rho$ for all observables with $\norm{O}_\infty\le1$. The result is the first matching the lower bounds derived in Refs.~\cite{Huang2022QuantumAdvantage,Ye2025ExponentialAdvantageReplica} and concludes an open theoretical question.

\subsection{Estimating dichotomic observables} \label{sec:dichotomic}
We first consider a dichotomic observable (i.e., a Hermitian matrix whose eigenvalues are $+1$ and $-1$), which is the foundation for developing an estimation protocol for general observables. 
A representative example is $O = Z_1 = Z \otimes I_{n-1}$, which acts as the Pauli-$Z$ operator on the first qubit and trivially on the remaining qubits. 
For any single-copy protocol, estimating $\Tr(Z_1\rho^2)$ to constant precision with constant failure probability has a sample complexity lower bound of $\Omega(\sqrt{d})$~\cite{Huang2022QuantumAdvantage}. 
In Appendix~\ref{app:lower_bound}, we show that this lower bound can be generalized to the estimation of all Pauli observables.
In the meantime, classical shadows require 
\begin{equation}
\cO(\sqrt{d \,\Tr(O^2)}) = \cO(d) 
\end{equation}
many samples~\cite{huang2020predicting, Seif2023ShadowDistillation,hu2022logical}, resulting in a quadratic gap relative to the lower bound. 
This gap primarily arises because the quantum experiments in classical shadows do not utilize information about the target observable. Here, we leverage the knowledge of $O$ to design a protocol that closes this gap.

Given the dichotomic observable $O$ with two eigenspaces $\Pi_{\pm}$ (eigenvalues $\pm1$ and projectors $P_{\pm}$) and corresponding dimensions $d_{\pm}$, the target state $\rho$ can be written in a block form with respect to these subspaces
\begin{equation}\label{eq:block_rho}
\rho = \begin{bmatrix}
        \tilde{\rho}_+ & B \\
        B^\dagger & \tilde{\rho}_-
        \end{bmatrix},
\end{equation}
where $\tilde{\rho}_{\pm}=P_\pm\rho P_\pm$. 
A key observation is that $\Tr(O\rho^2)$ can be expressed as a difference of purities within each eigen-subspace
\begin{equation}\label{eq:observation_block}
\begin{aligned}
\Tr(O\rho^2)
=\Tr\left(
\begin{bmatrix}
1 & 0\\
0 & -1
\end{bmatrix}
\begin{bmatrix}
\tilde{\rho}_+ & B \\
B^\dagger & \tilde{\rho}_-
\end{bmatrix}^2
\right)=\Tr\bigl(\tilde{\rho}_+^2\bigr)\;-\;\Tr\bigl(\tilde{\rho}_-^2\bigr).
\end{aligned}
\end{equation}
\noindent Therefore, our idea for estimating $\Tr(O\rho^2)$ is to separately estimate the purities of $\tilde{\rho}_\pm$ and then compute their difference. This can be done by modifying the random unitaries and the post-processing function.
As shown in Fig.~\ref{fig:framework}, given the dichotomic observable $O$, one first rotates the target state $\rho$ into the eigenbasis of $O$ using $V_O$ (such that $V_O O V_O^\dagger$ is a diagonal matrix), which makes the density matrix the form of Eq.~\eqref{eq:block_rho} in the computational basis.
To estimate purities of $\tilde{\rho}_\pm$, we rotate the state $V_O\rho V_O^\dagger$ with the unitary $U=\tilde{U}_+ \oplus \tilde{U}_-$, where $\tilde{U}_+$ and $\tilde{U}_-$ are independent random unitaries acting on the $d_+$- and $d_-$-dimensional subspaces, respectively.
The diagonal blocks of the rotated density matrix are $\tilde{U}_+\tilde{\rho}_+\tilde{U}_+^\dagger$ and $\tilde{U}_-\tilde{\rho}_-\tilde{U}_-^\dagger$, which corresponds to original blocks evolved under random unitaries $\tilde{U}_+$ and $\tilde{U}_-$.
Similar to the RM protocol, in each random unitary, we perform $N_M$ times of computational basis measurement to collect the outcomes $\{s_i\}_{i=1}^{N_M}$.
We then classify them into two groups depending on their eigenspaces.
Then, one can employ the estimator of Eq.~\eqref{eq:purity_estimator} for each group of data to estimate the purity of the corresponding block and evaluate the difference between them to estimate $\Tr(O\rho^2)$.

To summarize, the unbiased estimator of $\Tr(O\rho^2)$ can be chosen as
\begin{equation}\label{eq:unbiased_estimator_single}
    \omega_U = \binom{N_M}{2}^{-1} \sum\limits_{1 \le i < j \le N_M} X_2^{\mathrm{b}}(s_i,s_j).
\end{equation}
Here, the post-processing function is defined as
\begin{equation}\label{eq:ORM_data_processing}
    X_2^{\mathrm{b}}(s_i,s_j) = 
    \begin{cases}
        \lambda_\ell \,X_2^\ell(s_i,s_j), & \text{if } s_i, s_j \in \Pi_\ell \text{ for some } \ell\in \{+, -\},\\
        0, & \text{otherwise},
    \end{cases}
\end{equation}
where $\lambda_{\pm}=\pm 1$ are the eigenvalues corresponding to $\Pi_\pm$, and $X_2^\pm(s_i,s_j)$ is given by Eq.~\eqref{eq:RRM_coef_purity} with $d$ replaced by $d_\pm$. 
Note that when $O$ is the identity operator, this estimator reduces to the estimator of the original RM protocol (using a global unitary) for purity estimation~\cite{van2012Measuring, Brydges2019Probing, elben2019toolbox}.
Combining the original proof of RM shown in Eq.~\eqref{eq:purity_expectation} and the observation of Eq.~\eqref{eq:observation_block}, one can prove that $\omega_U$ is indeed an unbiased estimator of $\Tr(O\rho^2)$ (see Appendix~\ref{app:observable-driven_analysis} for details). 
To further reduce the variance of the estimator, we can take the average of multiple estimators in Eq.~\eqref{eq:unbiased_estimator_single} estimated with $U$ 
being chosen as $N_U$ independent random unitaries $\{U_1,\cdots,U_{N_U}\}$, leading to 
\begin{equation}\label{eq:unbiased_estimator}
    \omega = \frac1{N_U}\sum\limits_{u=1}^{N_U}\omega_{\,U_u}.
\end{equation}
Furthermore, by repeating the protocol $T = \cO\bigl(\log (\delta^{-1})\bigr)$ times and applying the median-of-means method, we can suppress the failure probability to $\delta$. 
We summarize the full procedure in Protocol~\ref{protocol:dichotomic}.

\begin{algorithm}[htbp] \label{protocol:dichotomic}
\caption{Estimating $\Tr(O\rho^2)$ for a dichotomic observable $O$}
\KwIn{
    $T, N_U, N_M$: number of repetitions, unitaries per repetition, and measurements per unitary.
    $O$: classical description of a dichotomic observable and its two eigenspaces with dimensions $d_+$ and $d_-$. 
    $T \times N_U \times N_M$ copies of $\rho$.
}
\KwOut{$\omega$: an estimator of $\Tr(O\rho^2)$.}

Find the unitary $V_O$ that diagonalizes $O$, i.e., $V_O O V_O^\dagger$ is a diagonal matrix in the computational basis.

\For{$t = 1$ \KwTo $T$} {
\For{$u = 1$ \KwTo $N_U$}{
    Sample random unitaries $\tilde{U}_+$ and $\tilde{U}_-$ with dimensions $d_+$ and $d_-$, respectively, and let $U=\tilde{U}_+\oplus\tilde{U}_-$.\\
    \For{$i = 1$ \KwTo $N_M$} {
        Rotate $\rho$ with $V_O$ and $U$ to get $UV_O\rho V_O^\dagger U^{\dagger}$.\\
        Perform a computational basis measurement and obtain an outcome $s_{u,i}^{(t)}$.
    }
    Compute $\omega_{u}^{(t)} = \binom{N_M}{2}^{-1} \sum\limits_{1\le i < j \le N_M} X_2^{\mathrm{b}}(s_{u,i}^{(t)}, s_{u,j}^{(t)})$.
    }
    Compute $\omega^{(t)} = \frac{1}{N_U} \sum\limits_{u=1}^{N_U} \omega_{u}^{(t)}$.
}
Return $\omega = \mathrm{median}\{\omega^{(1)},\omega^{(2)},\dots,\omega^{(T)}\}$.
\end{algorithm}

Recall that the RM protocol described earlier can estimate the purity $\Tr(\rho^2)$ of a density matrix to constant additive error using $\cO(\sqrt{d})$ samples. The procedure of the ORM protocol for estimating a dichotomic observable is equivalent to performing the RM protocol in two blocks of the target state and then post-processing accordingly.
Consequently, the ORM protocol has a similar performance guarantee to the RM protocol~\cite{anshu2022distributed}, as stated in the following theorem (the proof is provided in Appendix~\ref{app:proof_first_theorem}):

\begin{lemma}[Estimating dichotomic observables with block-diagonal unitary 4-designs]
\label{lem:estimating_dichotomic}
Let $O$ be a dichotomic observable, the dimensions of its two eigenspaces be $d_+, d_-\neq0$, and two random unitaries $\tilde{U}_+$ and $\tilde{U}_-$ be sampled from a unitary 4-design. 
For any input state $\rho$ in a $d$-dimensional qubit system,  Protocol~\ref{protocol:dichotomic} estimates 
the non-linear quantity $\Tr(O\rho^2)$ to additive error $\varepsilon$ with probability at least $1-\delta$,  with a sample complexity of \begin{equation}
\cO\left(\max\left\{\frac{\sqrt{d}}{\varepsilon},\frac{\sqrt{d/d_+}}{\varepsilon^2},\frac{\sqrt{d/d_-}}{\varepsilon^2}\right\} \log (\delta^{-1})\right).
\end{equation}
Specifically, the number of repetitions is $T = \cO\bigl(\log(\delta^{-1})\bigr)$, the number 
of random unitaries per repetition is $N_U = \cO\left(\max\left\{1,\frac1{d_+\varepsilon^2},\frac1{d_-\varepsilon^2}\right\}\right)$, and the number of measurements per unitary is 
$N_M = \cO\bigl(\frac{\sqrt{d}}{\varepsilon\sqrt{N_U}}\bigr)$.
\end{lemma}

Note that the sample complexity given in Lemma~\ref{lem:estimating_dichotomic} cannot cover the sample complexity for estimating purity, which corresponds to the observable of identity $\mathbb{I}$ with $d_-=0$.
Actually, by slightly modifying our derivation in Appendix~\ref{app:proof_first_theorem}, we can derive the sample complexity for purity estimation~\cite{anshu2022distributed}.

\begin{remark}[Purity estimation with unitary 4-designs]\label{rmk:purity}
When $O$ is the identity matrix, $O = \mathbb{I}$, i.e., has a unique eigenvalue $1$, Protocol~\ref{protocol:dichotomic} still applies by setting $d_+=d$, $d_-=0$. In this case, ORM is equivalent to the RM protocol measuring the purity~\cite{elben2019toolbox}. The sample complexity is $\cO\left(\max\left\{\frac{\sqrt{d}}{\varepsilon},\frac{1}{\varepsilon^2}\right\}\log\left(\delta^{-1}\right)\right)$ by setting $T = \cO\bigl(\log(\delta^{-1})\bigr)$, $N_U = \cO\left(\max\left\{1,\frac{1}{d\varepsilon^2}\right\}\right)$, and $N_M=\cO\left(\frac{\sqrt{d}}{\varepsilon\sqrt{N_U}}\right)$.
\end{remark}

Note that when $d_{\pm} = \Omega(d)$, Protocol~\ref{protocol:dichotomic} achieves a sample complexity upper bound as $\cO\left(\max\left\{\frac{\sqrt{d}}{\varepsilon},\frac{1}{\varepsilon^2}\right\}\right)$, which is similar with the sample complexity for estimating purity.
However, when $d_-=\cO(1)$ or $d_+=\cO(1)$, the sample complexity upper bound can reach $\cO(\frac{\sqrt{d}}{\varepsilon^2})$, which is worse than the purity estimation.
We can use an observable decomposition technique to avoid this upper bound, the proof of which is left to  Appendix~\ref{app:decomp_dichotomic_obs}.
\begin{lemma}[Dichotomic decomposition]\label{lem:Dichotomic_decomposition}
For any dichotomic observable $O$ with $\min\{d_+,d_-\} \le d/4$, one can decompose $O$ as
$O = \frac{1}{2}(\pm \mathbb{I}+O_1+O_2+O_3)$,
where each $O_i$ is a dichotomic observable with subspace dimensions being $\Omega(d)$.
\end{lemma}
\noindent Therefore, by estimating $\Tr(\rho^2)$ and each $\Tr(O_i\rho^2)$ to $\cO(\varepsilon)$ precision and combining the results, we have:
\begin{theorem}[Improved scaling in estimating dichotomic observables with block-diagonal unitary 4-designs]\label{thm:good_dicho_estimation}
For any dichotomic observable $O$, $\Tr(O\rho^2)$ can be estimated to within additive error $\varepsilon$ with probability at least $1-\delta$ using $\cO\left(\max\left\{\frac{\sqrt{d}}{\varepsilon},\frac{1}{\varepsilon^2}\right\}\log(\delta^{-1})\right)$
sample complexity with single-copy operations. 
\end{theorem}

The ORM protocol offers several practical advantages. Firstly, it only requires a minimal number of measurement settings. Notably, for any estimation accuracy $\varepsilon$ that is not exponentially small, even a constant number of measurement bases, $N_U$, is sufficient to estimate $\Tr(O\rho^2)$ accurately. This feature significantly reduces the experimental overhead, as reconfiguring the measurements is on many platforms more demanding than repeated sampling in a fixed basis~\cite{Myerson2008High,Lenzini2018Integrated,Dassonneville2020Fast,Opremcak2021High}.

Another crucial advantage of our protocol lies in its simplified gate calibration and post-processing. 
The data processing is independent of the specific random unitary applied, requiring only comparing the sampled bitstrings (see Eq.~\eqref{eq:ORM_data_processing}). This implies that the unitaries do not need meticulous, gate-by-gate calibration; it is sufficient that they are drawn from the required unitary design. 
In contrast to methods like classical shadows, ORM avoids computationally intensive steps like classical state reconstruction, and the post-processing function can be evaluated efficiently in practice. 
The runtime of computing the estimator over $T$ samples scales as $\cO(nT^2)$ under a brute-force enumeration of pairwise contributions, and can be improved to $\cO(nT)$ using standard data structures such as a trie. 
Thus, we do not expect the protocol to be limited in its scalability by the post-processing complexity but rather by the data acquisition.

Furthermore, the ORM protocol achieves the optimal sample complexity for Pauli observables.
According to Lemma~\ref{lem:estimating_dichotomic}, when $\varepsilon$ and $\delta$ are constants, the sample complexity for estimating Pauli observables is upper bounded by $\cO(\sqrt{d})$, which is more efficient than other prevailing single-copy randomized measurement protocols such as classical shadows (see Appendix~\ref{app:classical_shadow} for an analysis of the classical shadows).
In Appendix~\ref{app:lower_bound}, we generalize the proof in Ref.~\cite{Huang2022QuantumAdvantage} and show that $\cO(\sqrt{d})$ is optimal for all Pauli observables, thus answering an open question in quantum learning. 
In the next section, we will show that our protocol is actually optimal for a much broader class of observables (see Fact~\ref{fact:large_one_norm}).

\subsection{Estimating general observables}
We now discuss how to generalize the protocol from dichotomic observables to estimating general observables $O$ subject to $\|O\|_{\infty} \le 1$ (otherwise rescale $O$ according to ${O}/{\|O\|_{\infty}}$).
The main observation is that an arbitrary observable can be approximated, to a desired precision $\varepsilon$, by a sum of $\cO(\log\varepsilon^{-1})$ dichotomic observables.
With the decomposition, we can then apply Protocol~\ref{protocol:dichotomic} to estimate the expectation value of each component and sum the resulting estimators to obtain the final estimator. 
We first show how to perform such a decomposition.

\begin{lemma}[Dichotomic decomposition of general observables]\label{lem:decompose}
    For any observable $O$ satisfying $\|O\|_{\infty} \le 1$, one can decompose $O$ as
    \begin{equation}\label{eq:O_decomposition}
        O = \sum\limits_{l=1}^{k} 2^{-l} O_{l} + O_{\Delta},
    \end{equation}
    where $k \in \mathbb{N}^{+}$, each $O_{l}$ is a dichotomic observable with eigenvalues $\pm 1$, and $\|O_{\Delta}\|_{\infty} \le 2^{-k}$.
\end{lemma}

\begin{proof}
    First, suppose $O$ is a diagonal matrix, $O = \mathrm{diag}\bigl(a_0,a_1,\dots,a_{d-1}\bigr)$,
    where $d = 2^n$ is the dimension of the Hilbert space. Given $k \in \mathbb{N}_{+}$, each diagonal element $a_i$ can be written as
    \begin{equation}
        a_i = \sum\limits_{l=1}^k 2^{-l} a_{i,l} \;+\; a_{i,\Delta},
    \end{equation}
    where $a_{i,l} \in \{\pm 1\}$ and $0 \le \lvert a_{i,\Delta}\rvert \le 2^{-k}$. To see this, define 
    $a_{i,1} = \mathrm{sgn}(a_i)$, recursively 
    \begin{equation}
    a_{i,l} = \mathrm{sgn}\Bigl(a_i - \sum\limits_{j=1}^{l-1} 2^{-j} a_{i,j}\Bigr)\end{equation} for $2\leq l\leq k$, and $a_{i,\Delta} = a_i - \sum\limits_{l=1}^{k} 2^{-l} a_{i,l}$, where the (slightly modified) sign function is given by
    \begin{equation} 
    \mathrm{sgn}(x) 
    = 
    \begin{cases}
        1, & x \ge 0, \\
        -1, & x < 0.
    \end{cases}
    \end{equation} 
    By construction, 
    \(\bigl\lvert a_i - \sum\limits_{j=1}^{l} 2^{-j} a_{i,j}\bigr\rvert \le 2^{-l}\) for all \(1 \le l \le k\). Let
    \begin{equation}
        \begin{split}
            O_l &= \diag(a_{0,l},a_{1,l},\dots,a_{d-1,l}), \quad 1 \le l \le k, \\
            O_{\Delta} &= \diag(a_{0,\Delta},a_{1,\Delta},\dots,a_{d-1,\Delta}).
        \end{split}
    \end{equation}
    We have that
    \begin{equation}\label{eq:obs_decomposition}
        O = \sum\limits_{l=1}^{k} 2^{-l} O_l \;+\; O_{\Delta},
    \end{equation}
    where each $O_l$ is diagonal with entries in $\{\pm 1\}$, and $\|O_{\Delta}\|_{\infty} = \max_i \{\abs{a_{i,\Delta}}\}\le 2^{-k}$.

    If $O$ is not diagonal, let $O = V_O^\dagger D V_O$ be its diagonalization, where $D$ is a real diagonal matrix and $V_O$ is unitary. Decompose $D$ into a sum of diagonal dichotomic matrices $D_{l}$ as above, and define $O_l = V_O^\dagger D_{l} V_O$, $O_{\Delta} = V_O^\dagger D_{\Delta} V_O$. 
    This yields the desired decomposition.
\end{proof}

Lemma~\ref{lem:decompose} shows that a general observable $O$ can be decomposed into a sum of dichotomic observables $\{O_{l}\}_{l=0}^k$. The contribution of residual term $O_{\Delta}$ to $\Tr(O\rho^2)$ decreases exponentially as $k$ increases. Notice that the constructive proof presented here does not necessarily yield the optimal number of decomposed terms, but provides a general decomposing procedure of an arbitrary observable to dichotomic ones.
By applying Protocol~\ref{protocol:dichotomic} to estimate each $O_l$ and ignoring $O_{\Delta}$ for  $k = \cO(\log \varepsilon^{-1})$, we obtain the protocol for estimating $\Tr(O\rho^2)$ for an arbitrary observable $O$ to a desired precision $\varepsilon$.
The full procedure is summarized in Protocol~\ref{protocol:general}, and the resulting sample complexity is stated below, with the proof being presented in Appendix~\ref{app:proof_of_general}.

\begin{theorem}[Estimating general observables with block-diagonal unitary 4-designs] \label{thm:estimating_general}
    Let $O$ be an observable with $\|O\|_{\infty} \le 1$. 
    For any input state $\rho$ in a $d$-dimensional qubit system, Protocol~\ref{protocol:general} estimates $\Tr(O\rho^2)$ to within additive error $\varepsilon$ 
    with probability at least $1-\delta$, using a sample complexity of 
    \begin{equation}
\cO\Bigl(\max\left\{\frac{\sqrt{d}}{\varepsilon}, \frac{1}{\varepsilon^{2}}\right\}\log\bigl[\delta^{-1}\log(\varepsilon^{-1})\bigr]\Bigr).
    \end{equation}
\end{theorem}

\begin{algorithm}[htbp] \label{protocol:general}
\caption{Estimating $\Tr(O\rho^2)$ for general observables (ORM)}
\KwIn{
    $\varepsilon$: desired precision of additive error.
    $\delta$: desired upper bound of failure probability.
    $O$: classical description of an observable with $\norm{O}_{\infty} \le 1$.
    $\cO\Bigl(\max\left\{\frac{\sqrt{d}}{\varepsilon}, \frac{1}{\varepsilon^{2}}\right\}\log\bigl[\delta^{-1}\log(\varepsilon^{-1})\bigr]\Bigr)$ copies of $\rho$.
}
\KwOut{$\omega$: an estimator of $\Tr(O\rho^2)$, satisfying $\Pr\left[\abs{\omega-\Tr(O\rho^2)}<\varepsilon\right]\geq1-\delta$.}

Set $k = \lceil \log_2(\varepsilon^{-1})\rceil + 1$, get the dichotomic deposition $O = \sum\limits_{l=1}^k2^{-l} O_l + O_{\Delta}$ using Lemma~\ref{lem:decompose}.

\For{$l = 1$ \KwTo $k$} {
    Set $\varepsilon'_l = (\frac{3}{2})^l\frac{\varepsilon}{8}, \delta' =\delta /k$.
    Obtain an estimator $\omega_{l}$ for $\Tr(O_{l}\rho^2)$ with precision $\varepsilon'$ and failure probability $\delta'$ using Remark~\ref{rmk:purity} and Theorem~\ref{thm:good_dicho_estimation}. 
}
Return $\omega = \sum\limits_{l=1}^k 2^{-l} \omega_l$.
\end{algorithm}

Furthermore, a matching lower bound exists for a broad class of observables, as established by the following result.
\begin{fact}[Lower bound of sample complexity for estimating observables with large trace-norm  \cite{Ye2025ExponentialAdvantageReplica}]\label{fact:large_one_norm}
Given an observable $O$ with $\norm{O}_1 = \Omega(d)$, any measurement protocol that uses single-copy operations on a $d$-dimensional input state $\rho$ to estimate $\Tr(O\rho^2)$ to within additive error $\varepsilon = \cO(1)$ requires at least $\Omega(\sqrt{d})$ sample complexity.
\end{fact}
The condition $\norm{O}_1 = \Omega(d)$ is met by most observables of physical interest, including Pauli observables, local observables, and typical physical Hamiltonians, establishing the broad optimality of our protocol for these key applications. See Sec.~\ref{sec:efficient_circuit} for a detailed discussion and efficient implementations for estimating these observables.

%As an observable-driven protocol, ORM tailors its measurement procedure to the specific observable of interest. 
%Thus, while estimating non-linear properties, it shows advantages in sample complexity and the required number of measurement bases compared to other prevailing single-copy randomized measurement protocols such as classical shadows (see Appendix~\ref{app:classical_shadow} for an analysis of the classical shadows). 

While the ORM protocol offers significant advantages in terms of sample and computational complexity in post-processing, it also presents challenges to a practical implementation. The protocol requires classically diagonalizing the target observable. Furthermore, it necessitates the implementation of random unitaries drawn from block-diagonal unitary designs, with their structure tailored to the specific observable being measured. Implementing these specialized ensembles generally has a similar circuit complexity as realizing global random unitaries. Additionally, the post-measurement data from an experiment cannot be reused for estimating multiple observables simultaneously.
A comparison between classical shadows and ORM is summarized in Table~\ref{tab:CS_ORM}.

\section{Efficient circuit implementations}\label{sec:efficient_circuit}

While the ORM protocol is theoretically appealing due to its optimal sample complexity, its practical utility additionally requires scalable implementation of the measurements. In this section, we develop efficient circuit implementations of ORM for several physically relevant observables. 
First, we present the ORM protocol tailored for Pauli observables. Then, we provide performance guarantees for ORM when replacing the unitary 4-design with Clifford unitaries for Pauli observables,  and discuss its performance when using approximate unitary designs and local gates. 
Finally, we develop a Pauli sampling protocol that decomposes the target observables into several Pauli strings and performs importance sampling on those strings to estimate their non-linear expectation values. We show that our ORM achieves $\cO(\sqrt{d})$ sample complexity for a wide range of physical observables and is sample-optimal for many of them.

\subsection{Estimating Pauli observables}\label{sec:hamiltonian}

The unitary transformation that rotates an arbitrary observable into its eigenbasis can be prohibitively complex, often requiring a circuit depth that scales exponentially with the system size. This complexity presents a barrier even for the simpler task of estimating a linear expectation value.
Fortunately, many observables of practical interest possess sufficient structure to be diagonalized efficiently. In this section, we focus on one of the most important classes of such observables—Pauli operators—and present an implementation of the ORM protocol that is not only efficient in terms of circuit resources but also achieves the optimal $\cO(\sqrt{d})$ sample complexity.

Starting from the simplest case where the target observable is $O=Z_1$, i.e., the Pauli-$Z$ operator on the first qubit and identity on the others.
In this case, the eigenspace with $+1$ eigenvalue is spanned by state vectors of the form of $\ket{0}\otimes\ket{\psi_{n-1}}$, where $\ket{\psi_{n-1}}$ is an arbitrary pure state defined on the second to the $n$-th qubits.
Similarly, the eigenspace with $-1$ eigenvalue is spanned by state vectors of the form $\ket{1}\otimes\ket{\psi_{n-1}}$.
Now, the gate implementation procedure in Protocol~\ref{protocol:dichotomic} can be greatly simplified into
the following.
\begin{enumerate}
\item Draw random unitaries $U_\pm$ from a unitary 4-design that act on the last $n-1$ qubits.
\item Measure the first qubit in the computational basis to obtain $i\in\{0,1\}$.
\item When the measurement result is $i=0$, apply $U_+$  on the other $(n-1)$ qubits, or apply $U_-$ otherwise.
\end{enumerate}
By judging the first bits of two measurement outcomes, it is easy to decide whether the two outcomes are from the same eigenspace or not, and thus decide the value of the post-processing function.
To summarize, when the target observable is $Z_1$, one does not need to perform the diagonalization unitary but instead two independent $(n-1)$-qubit random unitaries, which can be implemented with only log-depth circuits~\cite{schuster2024random}.
At the same time, the sample complexity remains $\mathcal{O}(\sqrt{d})$ as the procedure described above is actually equivalent to that of Protocol~\ref{protocol:dichotomic}.
From an experimental perspective, chaotic evolution in certain analog systems is known to approximate global random unitaries~\cite{RandomHamiltonians,Nakata2017RandomHamiltonian, Ho2022Emergent, Cotler2023EmergentIndividual}, which can be beneficial for quantum learning tasks~\cite{Tran2023Measuring, liu2024predicting,mark2024efficiently}. 
Recent experimental efforts have already used this chaotic behavior to benchmark quantum devices~\cite{Choi2023PreparingRandom}. 
Implementing our protocol on such experimental platforms is therefore both interesting and promising. 

Furthermore, the $(n-1)$-qubit unitary could also be replaced by the tensor product of $n-1$ local unitaries drawn from a single-qubit unitary 2-design to further reduce the implementation difficulty. 
The post-processing function needs to be modified correspondingly to
\begin{equation}\label{eq:RRM_coef_purity_local}
    X_2^{\mathrm{b}}(s_i,s_j) = 
    \begin{cases}
        2^{n-1}(-2)^{-\Delta(s_i,s_j)}, & \text{if } s_i[1]=s_j[1]=0,\\
        -2^{n-1}(-2)^{-\Delta(s_i,s_j)}, & \text{if } s_i[1]=s_j[1]=1,\\
        0, & \text{otherwise},
    \end{cases}
\end{equation}
where $\Delta(s_i,s_j)$ is the Hamming distance between $s_i$ and $s_j$, and $s_i[1]$ is the measurement outcome of the first qubit~\cite{Brydges2019Probing}.
While, in this case, we cannot guarantee that the protocol can reach the $\cO(\sqrt{d})$ sample complexity for estimating $\Tr(Z_1\rho^2)$.
We will use numerical experiments later to assess the performance of the local-version protocol. 

Notice that when $O$ is a general Pauli observable, we could always find a Clifford circuit $V_C$ rotating it to a single-qubit Pauli $Z$ acting on $O$'s support (say qubit 1), i.e., $Z_1=V_COV_C^\dagger$. 
Experimentally, this means one can first apply the circuit $V_C$ to the state $\rho$, then run the ORM protocol for the simple, fixed observable $Z_1$.
Crucially, when $O$ is a local Pauli operator, the required Clifford circuit $V_C$ is also local and can be implemented with much lower depth.
As illustrated in Fig.~\ref{fig:hamiltonian}, this approach provides an efficient circuit implementation for the entire class of Pauli observables.

\begin{figure}[!htbp]
\centering
\includegraphics[width=.95\linewidth]{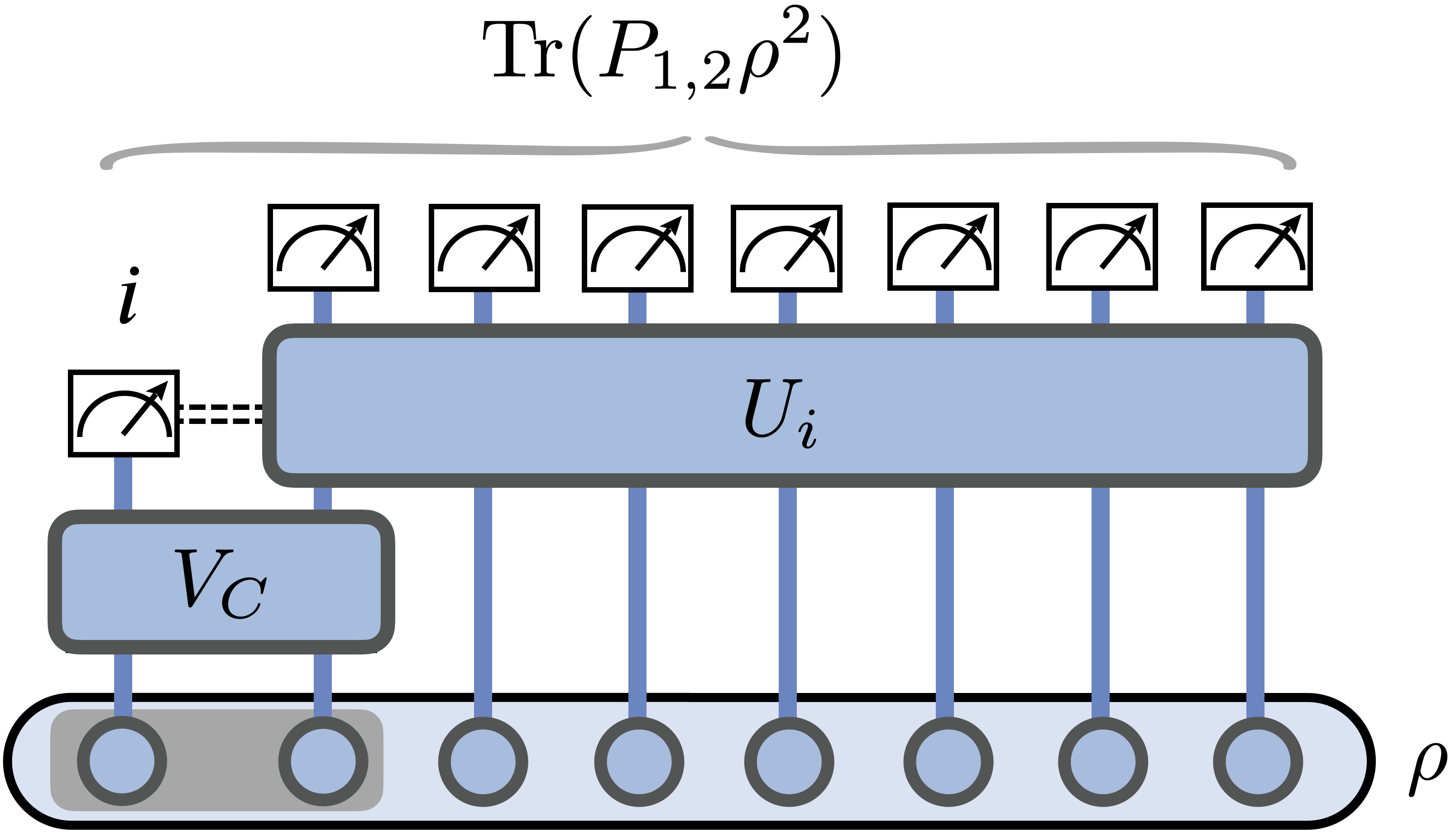}
\caption{Protocol for measuring any Pauli observable. Here, $V_C$ is a Clifford unitary 
that rotates $O$ to the Pauli-$Z$ operator on the first qubit. Depending on the measurement outcome of the first qubit, we then apply $U_i$ to the remaining part of the system, which can be either a global random unitary or a tensor product of single-qubit random unitaries.}
\label{fig:hamiltonian}
\end{figure}

\subsection{Replacing the exact $4$-design with Clifford unitaries, approximate design, and local gates}

While our protocol for Pauli observables achieves the optimal performance with an exact $4$-design, implementing such an ensemble is  demanding~\cite{Nakata2021ExactDesign}. In this section, we demonstrate that this requirement can be significantly relaxed by using a more practical ensemble of random Clifford unitaries~\cite{zhu2017multiqubit,zhu2016clifford}, which can be implemented efficiently with $\mathcal{O}(n)$ circuit depth on 1D architectures~\cite{Maslov2018Stabilizer, Bravyi2021Clifford}. This requirement parallels the global Clifford protocol used in classical shadows~\cite{huang2020predicting}.
Specifically, we prove that replacing the exact $4$-design with random Clifford unitaries preserves the optimal sample complexity scaling of $\mathcal{O}\left(\sqrt{d}\right)$, albeit at the cost of a less favorable dependence on the desired precision $\varepsilon$.
\begin{lemma}[Estimating traceless dichotomic observables with block-diagonal Clifford gates]
\label{lem:estimating_dichotomic_Clifford}
Let $O$ be a traceless dichotomic observable, i.e., the dimensions of its two eigenspaces are $d_+, d_-=d/2=2^{n-1}$, and two random unitaries $\tilde{U}_+$ and $\tilde{U}_-$ be sampled from the $(n-1)$-qubit Clifford group. 
For any input $n$-qubit state $\rho$ with $n\geq3$,  Protocol~\ref{protocol:dichotomic} estimates 
the non-linear quantity $\Tr(O\rho^2)$ to additive error $\varepsilon$ with probability at least $1-\delta$,  with a sample complexity of \begin{equation}
\cO\left(\frac{\sqrt{d}}{\varepsilon^2}\log (\delta^{-1})\right).
\end{equation}
Specifically, the number of repetitions is $T = \cO\bigl(\log(\delta^{-1})\bigr)$, the number 
of random unitaries per repetition is $N_U = \cO\left(\frac1{\varepsilon^2}\right)$, and the number of measurements per unitary is 
$N_M = \cO\left(\sqrt{d}\right)$.
\end{lemma}

The proof of Lemma~\ref{lem:estimating_dichotomic_Clifford} is detailed in Appendix~\ref{app:proof_dichotomic_Clifford}. 
With this result, we can now replace the unitary $U_i$ in Fig.~\ref{fig:hamiltonian} by a random Clifford unitary and achieve the sample-optimal single-copy estimation of Pauli observables with solely Clifford circuits.
Similar to the unitary 4-design case, our analysis also extends to the estimation of purity and inner product between two states using multi-qubit Clifford gates with only minor modifications, which is consistent with the results reported in Ref.~\cite{Zheng2025Distributed}.

Moreover, the requirement for an exact unitary 4-design can be relaxed to an approximate one. 
Although a rigorous proof is beyond the scope of this work, the reasoning is straightforward. 
If one uses an $\epsilon$-approximate unitary design, the resulting deviations in quantities like the expectation value (e.g., Eq.~\eqref{eq:purity_expectation}) and its variance will be bounded by $d^{c'}\epsilon$, where $c'$ is a constant. By selecting an approximation error that is polynomially small in the system dimension, specifically $\epsilon = d^{-c}$ for a constant $c > c'$, the error in the final estimate becomes negligible. The performance thus remains nearly identical to that achieved with an exact design. Crucially, such approximate designs can be implemented using a random circuit with depth of 
\cite{schuster2024random}
\begin{equation}
\log(\epsilon^{-1}) = c \log d = \cO(n),
\end{equation}
which is comparable to that of the random Clifford ensemble. 

The unitary ensemble can be simplified even further to shallow random circuits or even local random unitaries by adjusting the estimator form (see Eq.~\eqref{eq:RRM_coef_purity_local}). 
While these more practical ensembles do introduce some additional sample complexity overhead—for example, average sample complexities of $\cO(\sqrt{2.5^n})$ and $\cO(\sqrt{2.18^n})$ for local and shallow unitary ensembles~\cite{Zheng2025Distributed}—the cost remains significantly lower than the $\cO(2^n)$ complexity guarantee of classical shadow protocols. In Sec.~\ref{sec:applications}, we present numerical experiments that validate this advantage, showing that our protocol achieves better sample complexity than classical shadows even when using local unitaries.

\subsection{Estimating physical observables via Pauli sampling}\label{sec:pauli_sampling}

Building upon our schemes for estimating Pauli observables with Clifford unitaries, we now show that the non-linear expectation values of a wide range of observables—such as local observables and fidelity to stabilizer states—can be estimated using efficient circuit implementation schemes while maintaining the $\cO(\sqrt{d})$ sample complexity. This is optimal for most of these observables.

Our central method is built on a powerful technique: importance sampling of Pauli operators~\cite{Flammia2011DirectFidelityEstimation,huang2021efficient}. Given any observable $O$, we can decompose it in the Pauli basis as follows:
\begin{equation}\label{eq:Pauli_basis}
    O = \frac{1}{\sqrt{d}} \sum_{i=1}^K  \chi_O(P_i) P_i ,
\end{equation}
where $P_i \in \{I,X,Y,Z\}^{\otimes n}$ is a Pauli operator, and $\chi_O(P) \coloneqq \Tr(O P )/\sqrt{d}$ represents the corresponding Pauli coefficient. Here, $K$ denotes the number of Pauli operators needed to decompose $O$. 
The non-linear expectation values can thus be expressed as
\begin{equation}
    \Tr(O\rho^2) = \sum_{i=1}^K \chi_O(P_i) \chi_{\rho^2}(P_i).
\end{equation}
This expression can be further simplified to
\begin{equation}
    \Tr(O\rho^2) = \sum_{i=1}^K p_i X_i, 
\end{equation}
where the terms $p_i$ and $X_i$ are defined as:
\begin{equation}
\begin{split}
    p_i &= \frac{\chi^2_O(P_i)}{\norm{O}_2^2}, \\
    X_i &= \norm{O}_2^2 \frac{\chi_{\rho^2}(P_i)}{\chi_O(P_i)}.
\end{split}
\end{equation}
Note that $\{p_i\}_{i=1}^K$ forms a probability distribution over the Pauli operators that decompose $M$.

Our protocol proceeds as follows.  
\begin{enumerate}
    \item Sample $l = \cO\left(\frac{K \norm{O}_2^2}{d \varepsilon^2}\right)$ Pauli strings $P_i$ according to the probability distribution $\{p_i\}$.
    \item For each $1 \le i \le l$, use 
    \begin{equation}
        m_i = \cO\left( \frac{\norm{O}_2^4}{ \sqrt{d}l \varepsilon^2 \chi_O^2(P_i) } \right)
    \end{equation}
    samples of $\rho$ to obtain an estimator $\omega_i$ of $\Tr(P_i \rho^2)$ in Eq.~\eqref{eq:unbiased_estimator} using Protocol~\ref{protocol:dichotomic} and Lemma~\ref{lem:estimating_dichotomic_Clifford} with random Clifford gates.

    \item Compute the averaged estimator
    \begin{equation}\label{eq:estimator_Pauli_sampling}
        \omega = \frac{1}{l}\sum_{i=1}^l \frac{\norm{O}_2^2}{\Tr(OP_i)} \omega_i.
    \end{equation}

    \item Repeat steps 1 to 3 for $T = \cO(\log \delta^{-1})$ times and apply the median-of-means method to the obtained $T$ estimators to output the final estimators.
\end{enumerate}

We prove that the above Pauli sampling protocol achieves the following performance in Appendix~\ref{app:Pauli_sampling}:
\begin{theorem}[Estimating observables via Pauli sampling and block-diagonal Clifford gates]\label{thm:Pauli_sampling}
     Let $O$ be an observable that admits a Pauli decomposition, with the number of Pauli basis elements given by $K$ (see Eq.~\eqref{eq:Pauli_basis}). Then, the Pauli sampling protocol estimates the non-linear quantity $\Tr(O\rho^2)$ to additive error $\varepsilon$ with probability at least $1-\delta$, with an expected sample complexity of 
    \begin{equation}
        \cO\left(\frac{K \norm{O}_2^2}{\sqrt{d} \varepsilon^2} \log\left(\delta^{-1}\right) \right).
    \end{equation}
\end{theorem} 

The actual sample complexity is unlikely to exceed its expected value by much, as guaranteed by Markov's inequality. 
Furthermore, we note that the sample complexity can be further improved using unitary 4-designs, which can result in a better dependence on $\varepsilon^{-1}$ rather than $\varepsilon^{-2}$ for the precision $\varepsilon$.

\subsection{Local observables and Hamiltonians}

We now apply the Pauli sampling protocol to a wide range of observables and demonstrate that it achieves $\cO(\sqrt{d})$ complexity.

Local observables, which act non-trivially on only a constant number of qubits, are ubiquitous in quantum many-body physics (e.g., correlation functions and local energy terms).   
For any $k$-local observable $O = O_A \otimes \mathbb{I}_{[n] \backslash A}$ with $\abs{A} = k$, we can always decompose it using $K = 4^k$ Pauli operators, each acting non-trivially only on $A$. Assuming the proper normalization condition $\norm{O}_{\infty} \le 1$, we have $\norm{O}_2^2 \le d$, leading to the following performance guarantee. 
\begin{corollary}[Performance guarantee for local observables and Hamiltonians]
    Let $O$ be a $k$-local observable satisfying $\norm{O}_{\infty} \le 1$. Then, the Pauli sampling protocol estimates $\Tr(O\rho^2)$ to additive precision $\varepsilon$ with probability $1-\delta$, with a sample complexity of $\cO\left(K \frac{\sqrt{d}}{\varepsilon^2} \log\left(\delta^{-1}\right)\right)$, where $K \le 4^k$ is the number of nonzero Pauli terms in the decomposition of $O$.
\end{corollary}
This method extends directly to local Hamiltonians, which model most practical quantum systems and give the energy of the target states. Such Hamiltonians can be written as
\begin{equation}\label{eq:local_Hamiltonian}
    H = \sum_{i=1}^K a_i P_i, 
\end{equation}
where $K = \poly(n)$ and each $P_i$ is a local Pauli operator. The 2-norm of $H$ is given by $\norm{H}_2^2 = \sum_i a_i^2 d$, which yields the sample complexity
\begin{equation}
    \cO\left(\Bigl(\sum_i a_i^2 \Bigr) K \frac{\sqrt{d}}{\varepsilon^2} \log\left(\delta^{-1}\right)\right).
\end{equation}

For $k$-local observables, we have that $\norm{O}_1 = 2^{n-k} \norm{O_A}_1 = \Omega(d)$ for constant $k$ and $\norm{O_A}_1 = \Omega(1)$. By Fact~\ref{fact:large_one_norm}, estimating the non-linear expectation value of such observables requires at least $\Omega(\sqrt{d})$ sample complexity, so our protocol is sample-optimal across all single-copy schemes. 
Similarly, since the energy is non-vanishing for most eigenstates of practical Hamiltonians $H$, we have $\norm{H}_1 = \Omega(d)$, and thus our protocol achieves the optimal sample complexity $\tilde{\cO}(\sqrt{d})$ up to a logarithmic factor $\Bigl(\sum_i a_i^2 \Bigr) K$.

\subsection{Stabilizer fidelities and subspaces}
Another important class of observables is the fidelities to stabilizer states $\ket{\psi}$, which play a fundamental role in quantum information theory~\cite{gottesman1997stabilizer}. A stabilizer state $\ket{\psi}$ can be represented as a uniform sum over Pauli strings in its stabilizer group:
\begin{equation}
    \ketbra{\psi} = \frac{1}{2^n}\sum_{i=1}^{2^n} P_i, \quad P_i \in \mathcal{G}_\psi,
\end{equation}
where $\abs{\mathcal{G}_{\psi}} = 2^n$. This can be generalized to a stabilizer subspace, represented by the projector
\begin{equation}\label{eq:projector_stabilizer_subspace}
    \Pi = \frac{1}{2^{k}} \sum_{i=1}^{2^{k}} P_i, \quad P_i \in \mathcal{G}_{\Pi}.
\end{equation}
where $\mathcal{G}_{\Pi}$ contains all Pauli strings that stabilize this subspace, and $2^{n-k}$ is the rank of the subspace. For example, when $\Pi$ represents the set of all code states of a stabilizer code, $\Pi$ becomes the projector onto the code space. This is important in tasks like subspace verification for quantum error-correcting codes~\cite{chen2024quantumsubspaceverificationerror}. When $k=n$, this reduces to a single stabilizer state.

For any $1 \le k \le n$, we have $\norm{\Pi}_2^2 = 2^{n-k}$ and $K = 2^{k}$, which leads to the following sample complexity guarantees.
\begin{corollary}
    Let $1 \le k \le n$ and $\Pi$ be the projector onto a stabilizer subspace as given in Eq.~\eqref{eq:projector_stabilizer_subspace}. Then, the Pauli sampling protocol estimates $\Tr(\Pi \rho^2)$ to additive precision $\varepsilon$ with probability $1-\delta$, with a sample complexity of $\cO\left(\frac{\sqrt{d}}{\varepsilon^2} \log\left(\delta^{-1}\right)\right)$.
\end{corollary}
\noindent When $k$ is constant, our protocol achieves optimal sample complexity by Fact~\ref{fact:large_one_norm}. To the best of our knowledge, it remains open whether the sample complexity $\Omega(\sqrt{d})$ is optimal for low-rank projectors. See Sec.~\ref{sec:low_rank} for further discussion on estimating non-linear expectation values of low-rank observables. Moreover, for states with low non-stabilizerness, our protocol can potentially preserve the $\cO(\sqrt{d})$ sample complexity~\cite{Hinsche2025DistributedPauliSampling}.

\section{Applications}\label{sec:applications}
Estimating higher-order expectation values of quantum states is both a fundamental task and a key primitive in various fields. 
In this section, we highlight three applications: quantum virtual cooling, mixed-state phase detection, and entropy estimation.
We perform numerical experiments to confirm the advantages of ORM compared with classical shadows.

\subsection{Quantum virtual cooling} \label{sec:app_virtual_cooling}

A particularly important case that has attracted much recent attention is error mitigation on noisy intermediate-scale quantum devices and slightly beyond~\cite{Cai2023mitigation}. In quantum experiments, we often aim to obtain the expectation values of observables on a quantum state prepared by a quantum device. However, due to imperfections such as gate errors or insufficient cooling, the actual state is typically a noisy mixed state $\rho$ rather than a pure state. 
Although some severe obstructions against the scalability of classes of quantum error mitigation methods have been identified \cite{ErrorMitigationObstructionsOld,ErrorMitigationObstructions,PhysRevLett.131.210602,kento2023bound_qem}, various such notions (also involving quantum processing aspects) remain a key ingredient in realizing functioning near-term quantum architectures.

The most important information is usually encoded in the principal component $\ket{\psi}$ of $\rho$, which can be expressed as
\begin{equation}
\rho = (1-p)\ketbra{\psi} + p \rho_{\perp},
\end{equation}
where $\rho_{\perp}$ has orthogonal support to $\ket{\psi}$. 
When the noise strength $p$ is small, the error contributed by the noisy part $\rho_{\perp}$ can be substantially suppressed by considering the second-order expectation value ${\Tr(O\rho^2)}/{\Tr(\rho^2)}$~\cite{hugginsVirtualDistillationQuantum2021,koczor2021exponential}.
When $\rho_{\perp}$ has exponentially large rank, such as in the case of white noise, ${\Tr(O\rho^2)}/{\Tr(\rho^2)}$ can even exponentially suppress the noise with respect to the number of qubits.
Therefore, estimating $\Tr(O\rho^2)$ and the purity $\Tr(\rho^2)$ provides a significantly improved estimation of $\bra{\psi}O\ket{\psi}$ compared to directly measuring $\Tr(O\rho)$.

When the noise is primarily introduced by 
thermal excitations, the target state can often be described by a Gibbs state of the form $\rho \propto \exp(-\beta H)$, where $H$ is the Hamiltonian and $\beta>0$ is the inverse temperature.
In this setting, estimating ${\Tr(O\rho^2)}/{\Tr(\rho^2)}$ effectively predicts the expectation value on a lower-temperature state, as $\rho^2 \propto \exp(-2\beta H)$, providing access to physical properties at temperatures beyond the experimental cooling capability~\cite{cotler2019cooling}.

Here, we numerically simulate the quantum virtual cooling task using the global and local versions of ORM. 
We consider the Gibbs state of a system of $L$ spins governed by the 
Heisenberg XX Hamiltonian 
\begin{equation}\label{eq:H_XY}
H_{XX} = \sum\limits_{1\leq i<j\leq L} 2hJ_{i,j}(X_iX_j + Y_iY_j) + hB_z\sum\limits_{i=1}^L Z_i
\end{equation}
with a transverse field, which is well studied and implemented in trapped-ion platforms~\cite{jurcevic2014quasiparticle,Brydges2019Probing}.
The antiferromagnetic coupling strength decays according to a power law $J_{i,j} = J_0/{|i-j|^\alpha}$ with $J_0 = 420\,\mathrm{s}^{-1}$ and $\alpha = 1.24$. 
The transverse field is set to $B_z = 50J_0$ so that $H_{XX}$ well approximates the system of hopping hard-core bosons~\cite{jurcevic2014quasiparticle}.
The observable of interest is the two-local Pauli correlator $O = Z_1 Z_2$. 
Starting from a reference temperature $T_{0}=200J_{0}$ (i.e.\ $\beta_{0}=1/T_{0}$), we examine the thermal states
\begin{equation}
    \rho(\beta)=\frac{e^{-\beta H_{XX}}}{\Tr\bigl(e^{-\beta H_{XX}}\bigr)}
\end{equation}
for $\beta \in\{\beta_{0},\,2\beta_{0},\,3\beta_{0},\,4\beta_{0},\,5\beta_{0}\}$.  
For each $\beta$, we estimate
\begin{equation}
    \frac{\Tr\left[Z_1Z_2\rho(\beta)^2\right]}{\Tr\left[\rho(\beta)^2\right]} = \Tr\left[Z_1Z_2\rho(2\beta)\right]
\end{equation}
using both \emph{global-unitary ORM} (GORM) and 
\emph{local-unitary ORM} (LORM).  We numerically simulate a 6-qubit system and present the results in Fig.~\ref{fig:cooling}, where each data point is the mean of ten experiments with $N_U=10$ random unitaries and $N_M=1000$ measurements per unitary, with error bars indicating the standard error.  
The random unitaries are drawn from the global Haar measure or the tensor product of single-qubit Haar measure.
Fig.~\ref{fig:cooling} also shows 
the exact theoretical values of $\Tr\left[Z_{1}Z_{2}\rho(\beta)\right]$ and $\Tr\left[Z_{1}Z_{2}\rho(2\beta)\right]$ 
(dashed lines) for comparison.
The numerical results confirm that both GORM and LORM accurately reproduce the expectation value of $Z_{1}Z_{2}$ at 
half the temperature of the original thermal state with reasonable experiment repetitions.  
Because local unitaries are easier to implement experimentally, LORM is often preferred despite its moderately larger estimator variance.

\begin{figure}[!htbp]
\centering
\includegraphics[width=0.98\linewidth]{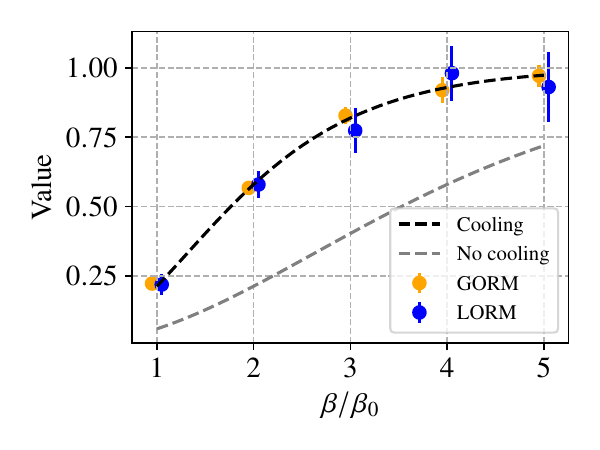}
\caption{
Virtual cooling of 6-qubit thermal states using GORM and LORM protocols.
The black dashed curve shows the theoretical value $\Tr[Z_1Z_2\,\rho(2\beta)]$ on cooled states, while the grey dashed curve shows the value $\Tr[Z_1Z_2\,\rho(\beta)]$ on original thermal states. 
Markers give ORM estimates of $\Tr\left[Z_1Z_2\,\rho(\beta)^{2}\right] / \Tr\left[\rho(\beta)^{2}\right]$ for five values of $\beta$ and are slightly displaced along the horizontal axis to enhance visibility. 
Each point is the mean of ten independent runs with 
$N_{U}=10$ random unitaries (drawn from the global Haar measure or the tensor product of single-qubit Haar measure) and $N_{M}=1000$ measurements 
per unitary, with error bars denoting the standard error.
Both protocols accurately reproduce the 
cooled-state expectation values.
}
\label{fig:cooling}
\end{figure}

Furthermore, we compare the number of state copies needed to estimate the non-linear expectation value 
$\Tr\bigl(Z_1Z_2\rho^{2}\bigr)$ to a target precision using four protocols: global and local versions of classical shadows and ORM.
The input state throughout is the Gibbs state at temperature $\beta = 8\beta_{0}$.
For each system size and each protocol we perform $T=100$ experiments, obtaining estimators $\omega_{1},\dots,\omega_{T}$.  
We quantify accuracy via the \emph{mean-squared error} (MSE),
\begin{equation}
  \Delta^{2}
  =
  \frac{1}{T}\sum_{i=1}^{T}
  \bigl[
      \omega_{i}-\Tr(Z_{1}Z_{2}\rho^{2})
  \bigr]^{2},
\end{equation}
and evaluate the minimum number of state copies required to achieve  
$\Delta^{2}\le 0.01$.

For ORM protocols, we sweep over the number~$N_{U}$ of random unitaries drawn from the Haar measure and, for every $N_{U}$, increase the number of measurements per unitary~$N_{M}$ until the MSE falls below~$0.01$.  
The sample complexity is given as  $\min_{\Delta^2(N_U,N_M) \le 0.01} \bigl(N_{U}N_{M}\bigr)$.

For classical shadows, we vary the total number of shadows~$N_{s}$.
In each experiment $i = 1,2,\dots,T$, we generate $2 N_s$ classical shadows $\{\hat{\rho}^{(i)}_j\}_{j=1}^{2N_s}$ using random unitaries drawn from the Haar measure, split them into two equally sized groups,
\begin{equation}
  \hat{\sigma}^{(i,1)}
  =
  \frac{1}{N_{s}}\sum_{j=1}^{N_{s}}
  \hat{\rho}^{(i)}_{j},
\quad
\hat{\sigma}^{(i,2)}
  =
  \frac{1}{N_{s}}\sum_{j=N_{s}+1}^{2N_{s}}
  \hat{\rho}^{(i)}_{j},  
\end{equation}
and form the estimator  
\begin{equation}
  \omega_{i}
  =
  \Tr\bigl[
    Z_{1}Z_{2}\,
    \tfrac{1}{2}\bigl(
      \hat{\sigma}^{(i,1)}\hat{\sigma}^{(i,2)}
      +
      \hat{\sigma}^{(i,2)}\hat{\sigma}^{(i,1)}
    \bigr)
  \bigr].
\end{equation}
Because each $\omega_{i}$ uses only $N_{s}^{2}$ distinct shadow pairs, rather than the full $(2N_{s})^{2}/2 = 2N_{s}^{2}$ possible pairs, we count the effective sample complexity as $\sqrt{2}N_{s}$ instead of $2N_{s}$. The reported sample complexity is the smallest $\sqrt{2}N_{s}$ that achieves an MSE below $0.01$.

The numerical results are summarized in Fig.~\ref{fig:shots}.  
Both global and local versions of ORM require far fewer state copies than their classical-shadow counterparts (GCS and LCS, respectively), especially when the number of qubits is large.  
These findings confirm that classical shadows do not attain optimal sample-complexity scaling.  
Whereas the classical-shadow curves grow exponentially with system size, the global ORM curve increases only slowly, remaining almost flat up to $L \le 10$ qubits, underscoring its excellent sample efficiency in practice.

\begin{figure}[!htbp]
\centering
\includegraphics[width=0.98\linewidth]{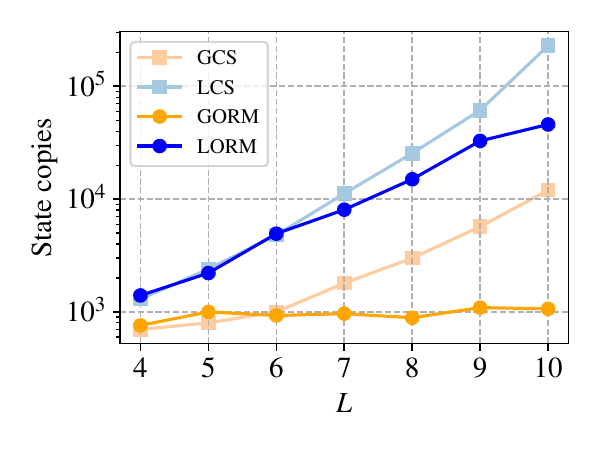}
\caption{
Number of state copies needed to estimate $\Tr(Z_1Z_2\,\rho^{2})$ with $\text{MSE}\le$ 0.01 averaged over 100 experiments, as a function of system size~$L$.  
Within both the global-unitary and local-unitary settings, ORM requires far fewer copies than classical shadows, offering a significant advantage in sample complexity with increasing system size.%
}
\label{fig:shots}
\end{figure}

\subsection{Detection of mixed-state
quantum phases}
The study of phases of matter and phase transitions is one of the most fundamental branches of physics, providing the conceptual foundation for understanding diverse phenomena across condensed matter physics, statistical mechanics, and even emerging connections to artificial intelligence~\cite{Sachdev}.
Traditionally, investigations of quantum phase transitions have largely focused on pure quantum states, particularly at zero temperature, where the ground state or low-lying excitations dominate the system’s behavior. 
In reality, perfectly pure states at zero temperature do not exist. 
Any physical system inevitably interacts with its environment, leading to decoherence and resulting in mixed states. 
As such, pure states should be viewed as a highly useful, yet ultimately idealized abstraction.

Motivated by this realization, recent efforts have sought to extend the study of quantum phases and phase transitions beyond pure states to encompass mixed states. 
These investigations have revealed novel physical phenomena with no direct counterparts in the pure-state regime~\cite{zhu2023nishimori,lee2023decoherence,lessa2025swssb}.
One prominent example is the discovery of symmetry-protected topological phases for mixed states, where researchers have found that certain topological properties, traditionally thought to be fragile under decoherence, can in fact remain robust even in noisy, open-system conditions~\cite{Lee2025symmetryprotected,deGroot2022symmetryprotected,zhang2025probingmixedstatephases} in a precise sense.

In the course of studying mixed-state phases, researchers have realized that traditional linear quantities, such as correlation functions, are often insufficient to fully capture the phase behavior and critical phenomena of mixed states. 
This has motivated the use of various nonlinear quantities, such as entanglement entropy, Rényi correlators, and fidelity correlators, that are more sensitive to quantum correlations and coherence.
Among these, our target quantity in this work, $\Tr(O\rho^2)$ with $O$ being a Pauli correlator, plays a particularly important role in characterizing mixed-state phases and their transitions~\cite{lee2023decoherence,Lee2025symmetryprotected}.
However, previous studies have primarily relied on the classical shadows to estimate such quantities~\cite{zhang2025probingmixedstatephases}. 
As demonstrated in Fig.~\ref{fig:shots}, the ORM protocol proposed in this work offers a significant advantage in sample complexity compared to classical shadows, where the observable is also set as the Pauli correlator. 
This further highlights the value 
and the practical usefulness of the ORM approach as a powerful tool for probing mixed-state phases and uncovering new physical insights beyond what is accessible through traditional techniques.

\subsection{Estimation of entropic quantities}\label{sec:entropy}

We here briefly mention a further application of $\Tr(O\rho^2)$ not considered before, that of the estimation of certain entropic quantities.
Specifically, the 
\emph{Petz–R{\'e}nyi relative entropy}~\cite{Petz1986QuasiEntropies} 
is defined for
$\alpha\in (0,1)\cup(1,\infty)$ as
\begin{equation}
D_
\alpha(\rho\| \sigma) := 
\frac{1}{\alpha-1}
\log_2 \Tr(\rho^\alpha \sigma^{1-\alpha}),
\end{equation}
a quantity with applications 
and a crisp interpretation in quantum information theory~\cite{Wilde2017QuantumInformation}.
For $\alpha=2$ and $\sigma$ being a full-rank density matrix, this gives
\begin{equation}
D_2 (\rho\| \sigma) := 
\log_2 \Tr(\rho^2 
\sigma^{-1}).
\end{equation}
As $\sigma^{-1}$ is Hermitian and can be seen as a quantum observable $O$, the ORM protocol can be used to estimate the value of $D_2 (\rho\| \sigma)$. 

When $\sigma$ is the maximally mixed state, $\sigma = \mathbb{I}/d$, $D_2(\rho\|\sigma)$ provides an alternative metric of how close $\rho$ is to the maximally mixed state, a task fundamental to the resource theory of purity~\cite{Streltsov2018ResourcePurity} and closedly related to quantum mixedness testing~\cite{chen2022tight}. 
Let 
\begin{equation}
    c \coloneqq D_2\Bigl(\rho\Bigm\| \frac{\mathbb{I}}{d}\Bigr) = n + \log \Tr(\rho^2),
\end{equation}
so that $\Tr(\rho^2) = 2^{-(n-c)}$. 
To estimate $c$ within additive error~$\varepsilon$, we need to produce an estimator~$\omega$ of $\Tr(\rho^{2})$ satisfying
$\omega\in[2^{-(n-c+\varepsilon)},2^{-(n-c-\varepsilon)}]$, 
which translates to an additive precision requirement $\varepsilon_{\omega} = \Theta\left(\varepsilon 2^{-(n-c)}\right)$.
By Remark~\ref{rmk:purity}, achieving this accuracy requires 
\begin{equation}
    \cO\left(\frac{1}{\varepsilon_{\omega}^2}\right) = \cO\left(\frac{d^2}{\varepsilon^2}\right)
\end{equation}
copies of $\rho$, reflecting the fundamentally exponential sample complexity of mixedness testing~\cite{chen2022tight}. 

The ORM protocol applies to estimating $D_2(\rho\|\sigma)$ for a general state $\sigma$. Let
\begin{equation}
    \kappa = \lambda_{\max}(\sigma)/\lambda_{\min}(\sigma)
\end{equation}
be the condition number of~$\sigma$, where 
$\lambda_{\max}$ and $\lambda_{\min}$ are its largest and smallest eigenvalues. 
We can rewrite
\begin{equation}
   c\coloneqq D_2(\rho \| \sigma)  = n + \log \kappa + \log \Tr(\rho^2 (\kappa d \sigma)^{-1}).
\end{equation}
To estimate $c$ within additive error~$\varepsilon$, it suffices to obtain an estimator $\omega$ of $\Tr(\rho^2 (\kappa d \sigma)^{-1})$ satisfying $\omega\in[\kappa^{-1}2^{-(n-c+\varepsilon)}, \kappa^{-1}2^{-(n-c-\varepsilon)}]$, which requires an additive precision 
\begin{equation}
\varepsilon_{\omega} = \Theta(\kappa^{-1}\varepsilon 2^{-(n-c)}).
\end{equation}
 As the matrix $(\kappa d \sigma)^{-1}$ is Hermitian and obeys 
 \begin{equation}
 \norm{(\kappa d \sigma)^{-1}}_{\infty} \le 1.
 \end{equation}
 Theorem \ref{thm:estimating_general} guarantees us to achieve this precision with $\tilde{O}({\kappa^2 d^2}/{\varepsilon^2})$ copies of $\rho$. 
When $\kappa=\mathcal{O}(1)$, the resulting sample complexity (up to logarithmic factors) matches that of the maximally mixed reference state.

\section{An alternate approach for low-rank observable estimation}\label{sec:low_rank}
In the previous section, we have utilized the local structure of the target observable to simplify the implementation of the ORM protocol, replacing the block-diagonal random unitary with two independent random unitaries. 
In this section, we will develop an alternative RM protocol to show that, when estimating $\Tr(O\rho^2)$ for low-rank observables, we can also replace the block-diagonal random unitary with a normal random unitary.
This protocol is partly inspired by the thrifty shadow protocol~\cite{helsen2023thrifty,Zhou2023performanceanalysis}.
Note that estimating such quantities is particularly important for error-mitigated fidelity estimation, where $O = \ketbra{\psi}$ and $\ket{\psi}$ is the target state vector.

For low-rank observable, Protocol~\ref{protocol:general} achieves a sample complexity of $\cO(\sqrt{d})$, as shown in Theorem~\ref{thm:estimating_general}. 
Interestingly, this favorable scaling does not necessarily require an observable-dependent protocol. 
In fact, the classical shadows can also estimate $\Tr(O\rho^2)$ with sample complexity $\cO(\sqrt{d})$ (see Appendix~\ref{app:classical_shadow}).
However, a major limitation of the classical shadows lies in its large number of required measurement settings. For example, when $\varepsilon, \delta = \cO(1)$, our protocol only requires $N_U = \cO(1)$ distinct measurement settings, whereas classical shadows demands $\cO(\sqrt{d})$ random unitaries. Reducing the number of measurement settings is crucial for practical quantum platforms such as superconducting and optical qubits~\cite{blais2021cqed,pan2012multiphoton}, where measurements under a single setting can be repeated efficiently, but changing the measurement basis is time-consuming.
Furthermore, the measurement cost can be exponentially reduced on platforms that support ensemble-average measurements across many copies, such as nuclear magnetic resonance systems, where a thermodynamically large number of quantum states can be simultaneously accessed~\cite{Liu2025CertifyingQuantumTemporal}.

The new protocol combines the central ideas of the RM protocol introduced in Section~\ref{sec:pre} and the classical shadows, and thus is named \emph{braiding randomized measurement} (BRM). 
Specifically, it shares the same quantum data-obtaining process with RM and borrows the idea from classical shadows in data post-processing.
To illustrate this, we first show how the purity-estimation protocol can be generalized to estimate $\Tr(\rho^3)$. 
Recall that the key step in estimating $\Tr(\rho^2)$ is to virtually implement a SWAP operator $\mathbb{S}$ using the data post-processing function (see Eq.~\eqref{eq:RRM_coef_purity} and Eq.~\eqref{eq:purity_expectation}). By a similar argument, an estimator for $\Tr(\rho^3)$ should virtually implement the third-order permutation operator. Concretely, the estimator takes the form
\begin{equation}\label{eq:estimator_third_moment}
\hat{M}_3^U=\binom{N_M}{3}^{-1}\sum\limits_{1 \le i < j < k \le N_M}X_3(s_i,s_j,s_k),
\end{equation}
where $X_3(s_i,s_j,s_k)$ is chosen so that
\begin{equation}\label{eq:RRM_third_twirling}
    \underset{U}{\mathbb{E}}\bigl[U^{\otimes 3} \,X_3\, U^{\dagger\otimes 3}\bigr]
    = \tfrac{1}{2}\bigl(\mathbb{P}^3_{(132)} + \mathbb{P}^3_{(123)}\bigr)
\end{equation}
for $U$ being sampled from a
unitary 3-design~\cite{Designs}.
Here, 
\begin{equation}
X_3 = \sum\limits_{s_1, s_2, s_3} X_3(s_1, s_2, s_3)\, \ketbra{s_1, s_2, s_3},
\end{equation}
and $\mathbb{P}^3_{(132)}, \mathbb{P}^3_{(123)}$ are third-order cyclic permutation operators satisfying $\mathbb{P}^3_{(132)}\,\ket{\phi_1, \phi_2, \phi_3}
= \ket{\phi_2, \phi_3, \phi_1}$ and $
\mathbb{P}^3_{(123)}\,\ket{\phi_1, \phi_2, \phi_3}
= \ket{\phi_3, \phi_1, \phi_2}$ for arbitrary $\ket{\phi_1}$, $\ket{\phi_2}$ and $\ket{\phi_3}$.
To fulfill Eq.~\eqref{eq:RRM_third_twirling} and make $\hat{M}_3^U$ unbiased, one can select $X_3(s_1,s_2,s_3) = X_3({\abs{\{s_1, s_2, s_3\}}})$, with 
\begin{equation}\label{eq:X3_coef}
X_3(1) = \frac{1 + d^2}{2},\ 
X_3(2) = \frac{1 - d}{2},\ X_3(3)=1
\end{equation}
for $d=2^n$~\cite{zhou2020Single}.
We now make a key observation: 
since 
\begin{equation}
\Tr(\rho^3) = \Tr(\rho^2\rho),
\end{equation}
replacing one instance of $\rho$ in $\Tr(\rho^3)$ with an observable $O$ yields $\Tr(O\rho^2)$. Leveraging the unbiasedness of the estimator $\hat{M}_3^U$, we replace one ``\emph{physical} probability'' $\bra{s}U\rho U^\dagger\ket{s}$ with a ``\emph{virtual} probability'' $\bra{s}UO U^\dagger\ket{s}$, which can be computed classically from $O$ and $U$. This replacement still preserves unbiasedness, producing an estimator for $\Tr(O\rho^2)$.
Specifically, the resulting single-round estimator is 
\begin{equation}\label{eq:RRM_estimator_single_round}
    \omega_u 
    = \binom{N_M}{2}^{-1}
      \sum\limits_{1 \le i < j \le N_M}
      \sum\limits_{\sigma \in \{0,1\}^n}
      X_3(s_i, s_j, \sigma)\,
      \bra{\sigma}\,U\,O\,U^\dagger\ket{\sigma}.
\end{equation}
In fact, by a minor modification to the estimator $\omega_u$, we could remove its dependency on $O$ and construct the unbiased estimator
\begin{equation}
\hat{\rho^2}=\binom{N_M}{2}^{-1}
      \sum\limits_{1 \le i < j \le N_M}
      \sum\limits_{\sigma \in \{0,1\}^n}
      X_3(s_i, s_j, \sigma)\,
      \,U^\dagger\ketbra{\sigma}{\sigma}\,U\,
\end{equation}
 for $\rho^2$. This could be regarded as the ``shadow'' of $\rho^2$ while estimating properties in the form of $\Tr(O\rho^2)$.
Similarly to the protocols proposed in Sec.~\ref{sec:observable_driven}, we can average over $N_U$ random unitaries, and then apply the median-of-means procedure to further suppress the failure probability.
We summarize the full procedure in Protocol~\ref{protocol:low-rank observables}.

\begin{algorithm}[H] \label{protocol:low-rank observables}
\caption{Estimating $\Tr(O\rho^2)$ for many low-rank observables (BRM)}
\KwIn{
    $T, N_U, N_M$: Number of repetitions, unitaries per repetition, and measurements per unitary. 
    $\{O_m\}_{m=1}^M$: classical description of $M$ observables satisfying $\Tr(O_i^2)=\cO(1)$. 
    $T \cdot N_U \cdot N_M$ copies of $\rho$.
}
\KwOut{$\{\omega_m\}_{m=1}^M$: estimators of $\left\{\Tr(O_m\rho^2)\right\}_{m=1}^M$, respectively.}
\For{$t = 1$ \KwTo $T$} {
\For{$u = 1$ \KwTo $N_U$}{
    Sample a random unitary $U_u$.
    
    \For{$i = 1$ \KwTo $N_M$} {
        Apply $U_u$ to $\rho$.
        Perform computational-basis measurement and obtain an outcome $s_{u,i}^{(t)}$.
    }
}
}
\For{$m = 1$ \KwTo $M$} {
    \For{$t = 1$ \KwTo $T$} {
    Compute $\omega_{m}^{(t)} 
    = \frac{1}{N_U} \sum\limits_{u=1}^{N_U}\binom{N_M}{2}^{-1}
      \sum\limits_{1 \le i < j \le N_M}
      \sum\limits_{\sigma \in \{0,1\}^n}$ $ X_3(s_{u,i}^{(t)}, s_{u,j}^{(t)}, \sigma)\,
      \bra{\sigma}UO_mU^\dagger\ket{\sigma}$
    .
      }
      Compute $\omega_m = \mathrm{median}\{\omega^{(1)}_m,\omega^{(2)}_m,\dots,\omega^{(T)}_m\}$.
    }

Return $\{\omega_m\}_{m=1}^M$.
\end{algorithm}

An advantage of BRM over ORM is that BRM does not require a block-diagonal unitary. Instead, it uses a global random unitary, which simplifies implementation in real experiments. Moreover, this observable-independent BRM protocol offers a notable advantage in estimating multiple low-rank observables compared to ORM. 
Specifically, suppose we wish to estimate $M$ observables $O_1, O_2, \dots, O_M$, each satisfying $\Tr(O_m^2)=\cO(1)$, with a failure probability $\delta$. By setting $\delta'=\delta/M$ and selecting the number of experiments to guarantee failure probability $\delta'$ for estimating a single observable, the union bound ensures that all $M$ observables can be estimated simultaneously with failure probability $\delta$. This approach introduces only an $\cO\left[\log (\delta'^{-1})\right]=\cO(\log M)$ overhead, rather than a multiplicative factor of $M$ that would result from naively repeating the entire protocol $M$ times.

We further remark that third-order twirling is required to virtually implement the permutation operators, so unitaries should be sampled from a unitary 3-design to ensure the estimator is unbiased. Furthermore, sampling from a 6-design is sufficient to guarantee the performance stated in Theorem~\ref{thm:low_rank} (see proof in Appendix~\ref{app:low_rank} for the single-observable case and the discussion above for the multi-observable case). 
\begin{theorem}[Estimating many low-rank observables]
\label{thm:low_rank}
    Let $\{O_m\}_{m=1}^M$ be a collection of observables satisfying $\Tr(O_m^2) =\cO(1)$ for all $i$. 
    For an input state $\rho$ in a $d$-dimension system, Protocol~\ref{protocol:low-rank observables} simultaneously estimates all quantities  $\Tr(O_m\rho^2)$ to additive error at most $\varepsilon>0$ with probability at least $1 - \delta$, with a sample complexity of $\cO(\sqrt{d} \varepsilon^{-2} \log (\frac{M}{\delta}))$. Specifically, the number of repetitions is $T = \cO\bigl(\log(\frac{M}{\delta})\bigr)$, the number of random unitaries per repetition is $N_U = \cO\bigl(\varepsilon^{-2} \bigr)$, and the number of measurements per unitary is $N_M = \cO\bigl(\sqrt{d}\bigr)$. 
\end{theorem}

We use Table~\ref{tab:low_rank} to compare the three protocols (classical shadows, ORM, and BRM) in estimating low-rank non-linear expectation value. Subject to constant additive error and failure probability, all three protocols have a sample complexity scaling with the square root of the system dimension. 
Our proposed protocols, both ORM and BRM, remarkably reduce the number of required measurement bases to a constant number, compared to the exponentially many with system size in classical shadows. BRM also shares the favorable scaling in simultaneously estimating multiple observables, since it is also an observable-independent protocol. 

\section{Summary and outlook}\label{sec:outlook}

In this work, we have presented observable-driven randomized measurement protocols for estimating non-linear properties of the form $\Tr(O\rho^2)$.
Our approach requires only single-copy quantum experiments and achieves provably optimal sample complexity among all observables with $\norm{O}_1=\Omega(d)$.
While previous single-copy protocols were known to be optimal only for purity and inner product estimation, we bridge this gap by incorporating observable information into the RM framework, achieving a sample complexity of $\mathcal{O}(\sqrt{d})$ for estimating any observable with bounded norm.
Note that the observable-driven idea was also adopted for classical shadows to reduce the complexity in estimating linear properties~\cite{huang2021efficient,van2024derandomized}.

For practical applications, we present an efficient Clifford-circuit implementation for estimating Pauli observables that preserves the optimal sample complexity. 
We then extend this implementation to a broad class of physical observables, such as local Hamiltonians and stabilizer fidelities, via a Pauli sampling protocol. We further simplify circuit requirements by employing approximate unitary designs, shallow circuits, and local unitaries. 
A systematic study of randomized measurement under approximate state designs—particularly the performance of shallow-depth approximate unitary designs~\cite{schuster2024random}—is an interesting direction for future work.

The theoretical intuition behind our protocol is rooted in Schur-Weyl duality, which refers to the fact that the action of a group of unitaries can filter out properties invariant under their action. This is the fundamental reason why the RM protocol can estimate the purity $\Tr(\rho^2)$, as the action of any unitary on the state $\rho$ leaves the purity unchanged.
Following this intuition, when estimating $\Tr(O\rho^2)$, it is natural to choose a unitary ensemble consisting of all unitaries that commute with the observable $O$.
In the case of a dichotomic observable, this insight leads naturally to Protocol~\ref{protocol:dichotomic}, which employs block-diagonal random unitaries.
Therefore, an intriguing direction for future work is to develop a unified framework capable of estimating functions like $\Tr(O_1\rho O_2\rho)$, which plays a vital role in mixed state quantum phase transition~\cite{lessa2025swssb} or more general functions $\Tr(O\rho^{\otimes t})$ by appropriately adjusting the random unitary ensemble. 
From a practical standpoint, one of the key motivations for estimating non-linear expectation values is to amplify the information of the principal component in a noisy target state.
While our work addresses the estimation of the second-order non-linear expectation value, a natural question is how to extend these methods to estimate higher-order values, $\Tr(O\rho^t)$ for $t>2$, or even to directly obtain $\bra{\psi}O\ket{\psi}$ using single-copy operations, where $\ket{\psi}$ is the principal component of the noisy mixed state $\rho$.
A related approach that combines randomized measurements with quantum error mitigation has also been developed in Ref.~\cite{Onorati2024Noise}. 
The distinction lies in that Ref.~\cite{Onorati2024Noise} mitigates the measurement process itself through the estimation of an appropriate frame.

In addition to this, we show that to estimate low-rank observables, the observable-independent BRM protocol proposed in Section~\ref{sec:low_rank} can achieve sample complexity and measurement bases equivalent to those of ORM. 
This shows once again that low-rank structure can reduce
the sample complexity of quantum learning protocols~\cite{Compressed}.
The BRM protocol integrates techniques from both RM and classical shadows, offering distinct advantages over these conventional approaches: compared with RM, its capability extends beyond the estimation of state moments, and compared with classical shadows, it significantly reduces the number of required measurement bases. 
This result represents an intriguing step toward a unified framework combining RM and classical shadows, two of the most prevalent approaches based on measurements in random bases~\cite{elben2023randomized, helsen2023thrifty,Zhou2023performanceanalysis}.

We further note that, to the best of our knowledge, it remains an open question whether $\Omega(\sqrt{d})$ is a fundamental lower bound for estimating $\Tr(O\rho^2)$ for low-rank observables under single-copy measurement schemes. 
Therefore, there is still potential to further optimize the sample complexity by exploring novel ways of incorporating observable information into the randomized measurement framework. 
While our work has focused on single-copy measurement schemes, a promising avenue for future research is the estimation of non-linear properties in other settings. These include multi-copy measurement scenarios~\cite{Chen2025SimultaneousEstimation}, models with access to oracles preparing purifications of the input state~\cite{Zhang2025MeasuringLearnMoreQuadratic, Tang2025ConjugateQueries}, and protocols designed for a limited number of copies~\cite{Chen2021HierarchyReplicaQuantumAdvantage, Ye2025ExponentialAdvantageReplica}.
We hope that the present work stimulates further endeavors in these directions.  

\begin{acknowledgments}
We appreciate the valuable discussion with Hong-Ye Hu, Weiyuan Gong, Shu Chen, You Zhou, Pei Zeng, and Zihao Li. We thank Lorenzo Leone for his help in proving Lemma~\ref{lem:upper_bound_Q_sigma}. Z.~D.~and Z.~L.~acknowledge the support from the National Natural Science Foundation of China Grant No.~12174216 and the Innovation Program for Quantum Science and Technology Grant No.~2021ZD0300804 and No.~2021ZD0300702. Y.~T.~is supported by the Quantum Flagship MILLENION. 
J.~E.~is supported by the 
BMFTR (DAQC, MuniQC-Atoms, Hybrid++, QuSol), the Munich Quantum Valley, Berlin Quantum, the Quantum Flagship programmes MILLENION  and PASQUANS2, the DFG (CRC 183),
and the European Research Council (DebuQC).
Part of the numerical results in this work is obtained using the Python package Qiskit~\cite{qiskit2024Quantum}.
\end{acknowledgments}

\bibliography{bib}

\clearpage

\appendix
\onecolumngrid

\section{Preliminaries}
\subsection{Random unitaries} \label{app:random_unitaries}
In this section, we will briefly introduce the foundations of random unitary integral, especially the twirling function. One could refer to the review article  Ref.~\cite{Mele2024introductiontohaar} for more information on random unitaries.
A unitary distribution is called Haar 
measure if and only if the following holds
\begin{equation}
\int_{\mathrm{Haar}}f(U)dU=\int_{\mathrm{Haar}}f(UV)dU=\int_{\mathrm{Haar}}f(VU)dU
\end{equation}
for arbitrary continuous function $f(\cdot)$ and unitary $V$.
According to the Schur-Weyl duality, the $t$-th order integral over Haar random unitaries is related to the $t$-th order permutation group. For $X\in\mathcal{L}\left(\left(\mathbb{C}^d\right)^{\otimes t}\right)$, we define its $t$-th order twirling and write it as a linear combination of permutation operators
\begin{equation}\label{eq:t-th_order_twirling}
\Phi_t(X)\coloneqq\int_\mathrm{Haar}U^{\otimes t}XU^{\dagger\otimes t}dU=\sum\limits_{\pi,\sigma\in S_t}\text{Wg}(\pi^{-1}\sigma,d)\Tr\left(\mathbb{P}^t_{\sigma^{-1}}X\right)\mathbb{P}^t_{\pi},
\end{equation}
where $\mathbb{P}^t_{\pi}\in\mathcal{L}\left(\left(\mathbb{C}^d\right)^{\otimes t}\right)$ acts as permutation $\pi\in S_t$ on $t$ states in $\mathbb{C}^d$ according to
\begin{equation}
\mathbb{P}^t_{\pi}=\sum\limits_{s_1,\dots,s_t}\ket{s_{\pi^{-1}(1)},\dots,s_{\pi^{-1}(t)}}\bra{s_1,\dots,s_t}.
\end{equation}
The Weingarten coefficient $\text{Wg}(\pi^{-1}\sigma,d)$~\cite{Weingarten1978Asymptotic,Collins2003Moments,köstenberger2021weingartencalculus} is defined as the element in the pseudo-inverse of the Gram matrix,
\begin{equation}\label{eq:weingarten}
\begin{aligned}
\text{Wg}(\pi^{-1}\sigma,d)&\coloneqq\left(G^+\right)_{\pi,\sigma},
\\
\text{where } \,G_{\pi,\sigma}&=\Tr\left(\mathbb{P}^t_{\pi^{-1}}\mathbb{P}^t_\sigma\right).
\end{aligned}
\end{equation}
The design of post-processing functions in the RM protocol aims to construct a suitable operator $X_t$, such that $\Phi_t(X_t)$ is a combination of desired permutation operators. For example, when estimating purities, we want to realize $\Phi_2(X_2)=\mathbb{S}$, which is then reduced to solving the set
\begin{equation}
\begin{aligned}
\Tr(X_2)&=d,\\
\Tr(\mathbb{S}X_2)&=d^2
\end{aligned}
\end{equation}of linear equations~\cite{elben2019toolbox}.
A direct construction is to make $X_2$ diagonal as $X_2=\sum\limits_{s_1,s_2}X(s_1,s_2)\ket{s_1,s_2}\bra{s_1,s_2}$ and determine a function $X_2(s_1,s_2)$ satisfying
\begin{equation}
\begin{aligned}
\sum\limits_{s_1,s_2}X_2(s_1,s_2)&=d,\\
\sum\limits_{s_1}X_2(s_1,s_1)&=d^2.
\end{aligned}
\end{equation}
The function in Eq.~\eqref{eq:RRM_coef_purity} provides such a solution. Furthermore, the function $X_3$ defined in Eq.~\eqref{eq:X3_coef} constructs an operator whose third-order twirling evaluates as a linear combination of cyclic permutation operators. Similarly, one can generalize this theory to design functions for constructing high-order permutation operators. It is worth mentioning that the inherent symmetry of the symmetric group may forbid a solution to equations corresponding to an arbitrary combination of permutation operators in Eq.~\eqref{eq:t-th_order_twirling}. One example is that when $X_3$ is subjected to a diagonal operator, the two cyclic permutation operators $\mathbb{P}^3_{(123)},\mathbb{P}^3_{(132)}$ always share the same coefficients~\cite{zhou2020Single}.

\subsection{Median-of-means method} \label{app:median-of-mean}
Here, we briefly introduce how the median-of-means method exponentially suppresses the failure probability. The median-of-means estimator is summarized in the following lemma:

\begin{lemma}[Median-of-means estimator]\label{lem:median-of-mean}
    For given parameters $0 < \delta < 1$ and $0 < c < \tfrac{1}{2}$, let
    \begin{equation}\label{eq:def_T}
        T = \frac{1}{2(\frac{1}{2} - c)^2} \log \left(\delta^{-1}\right).
    \end{equation}
    Suppose \(\{\omega^{(t)}\}_{t=1}^T\) are independent random variables satisfying $\mathbb{E}\left[\omega^{(i)}\right] = \bar{\omega}$ and \(\mathrm{Var}\left[\omega^{(t)}\right]\le c\,\varepsilon^2\) for any $t$. Define the median-of-means estimator
    \begin{equation}
        \omega = \mathrm{median}\left\{\omega^{(1)},\omega^{(2)},\dots,\omega^{(T)}\right\}.
    \end{equation}
    Then
    \begin{equation}
        \Pr\bigl[|\omega - \bar{\omega}| > \varepsilon\bigr] \le \delta.
    \end{equation}
\end{lemma}

\begin{proof}
    The event \(\{|\omega - \bar{\omega}| > \varepsilon\}\) is contained by the event
    \begin{equation}
        \frac{1}{T} \sum\limits_{t=1}^T \mathbbm{1}\bigl(|\omega^{(t)} - \bar{\omega}| > \varepsilon\bigr) 
        \ge \frac{1}{2},
    \end{equation}
    where \(\mathbbm{1}(\cdot)\) is the indicator function. Define
    \begin{equation}
        X_t = \frac1T\mathbbm{1}\bigl(|\omega^{(t)} - \bar{\omega}| > \varepsilon\bigr).
    \end{equation}
    Then \(X_t\) are independent random variables taking values in \(\{0,\frac1T\}\). By Markov's inequality,
    \begin{equation}
        \Pr\left[|\omega^{(t)} - \bar{\omega}| > \varepsilon\right]
        \le \frac{\Var[\omega^{(t)}]}{\varepsilon^2}
        \le c.
    \end{equation}
    Hence, \(\mathbb{E}[X_t] \le \frac cT\) and $\mathbb{E}\left[\sum\limits_{t=1}^T X_t\right]\leq c$. Applying Hoeffding's inequality, we have
    \begin{equation}
        \Pr\left[
          \sum\limits_{t=1}^T X_t \ge \frac{1}{2}
        \right]
        \leq
        \Pr\left[
          \sum\limits_{t=1}^T X_t-\mathbb{E}\left[\sum\limits_{t=1}^T X_t\right] \ge \frac{1}{2}-c
        \right]
        \le
        \exp\bigl(-2\bigl(\tfrac{1}{2}-c\bigr)^2\,T\bigr).
    \end{equation}
    Substituting
    \begin{equation}
        T = \frac{1}{2(\frac{1}{2} - c)^2} \log \left(\delta^{-1}\right)
    \end{equation}
    completes the proof.
\end{proof}

Now suppose we have a protocol that produces an estimator \(\omega^{(t)}\) with expectation \(\bar{\omega}\) and variance
\begin{equation}
    \Var\left[\omega^{(t)}\right] \le c\varepsilon^2.
\end{equation}
We run this protocol \(T = \cO(\log(\delta^{-1}))\) times (given by~\eqref{eq:def_T}), thereby obtaining estimator \(\omega\) to be the median of these \(T\) estimators:
\begin{equation}
    \omega = \mathrm{median}\left\{\omega^{(1)},\omega^{(2)},\dots,\omega^{(T)}\right\}.
\end{equation}
By Lemma~\ref{lem:median-of-mean}, \(\omega\) satisfies
\begin{equation}
    \Pr\left[|\omega-\bar{\omega}|\ge\varepsilon\right] \le\delta.
\end{equation}
Thus, the median-of-means method exponentially suppresses the failure probability compared to the average estimator $\frac1T\sum\limits_{t=1}^T\omega^{(t)}$.

\subsection{Classical shadows} \label{app:classical_shadow}

Here, we briefly review the classical shadows using global random unitaries introduced in Ref.~\cite{huang2020predicting} and discuss its relevant performance guarantees. In this scheme, one performs $N_s$ single-copy experiments on the input state $\rho$. In each round, one draws a global random unitary $U$ from a unitary 3-design and applies it to the state: $\rho \rightarrow U\rho U^{\dagger}$. A computational-basis measurement is then performed on $U\rho U^{\dagger}$. 
Suppose the measurement outcome for the $i$th unitary $U_i$ is $\ket{b_i}$, a ``classical shadow'' of $\rho$ is constructed as 
\begin{equation}\label{eq:classical_shadow}
    \hat{\rho}_i = (2^n+1)U_i^{\dagger}\ketbra{b_i} U_i - \mathbb{I},
\end{equation}
which satisfies $\mathbb{E}[\hat{\rho}_i] = \rho$.
To estimate $\Tr(O\rho^2)$, one defines the estimator~\cite{huang2020predicting}
\begin{equation}
    \omega_{\mathrm{CS}} = \binom{N_s}{2}^{-1} \sum\limits_{1 \le i < j \le N_s} \Tr(O_{\sym} \hat{\rho}_i \otimes \hat{\rho}_j),
\end{equation}
where 
\begin{equation}\label{eq:O_sym}
    O_{\sym} = \frac{1}{2} \mathbb{S} (O \otimes \mathbb{I} + \mathbb{I} \otimes O)
\end{equation}
is the symmetrized version of $O$, and $\mathbb{S}$ denotes the swap operator acting on $\left(\mathbb{C}^{2^n}\right)^{\otimes2}$. It follows that
\begin{equation}
    \mathbb{E}[\omega_{\mathrm{CS}}] = \binom{N_s}{2}^{-1}\sum\limits_{1 \le i < j \le N_s} \Tr(O_{\sym} \mathbb{E}[\hat{\rho}_i] \otimes \mathbb{E}[\hat{\rho}_j]) = \binom{N_s}{2}^{-1}\sum\limits_{1 \le i < j \le N_s} \Tr(O_{\sym} \rho \otimes \rho) = \Tr(O\rho^2).
\end{equation}
Hence, $\omega_{\mathrm{CS}}$ is an unbiased estimator of $\Tr(O\rho^2)$. Next, we show that the associated variance obeys the following bound.

\begin{lemma}\label{lem:variance_shadow_second_moment}
    The variance of the estimator $\omega_{\mathrm{CS}}$ satisfies
    \begin{equation}
        \Var[\omega_{\mathrm{CS}}] = \frac{1}{N_s^2}\cO(\Tr(O)^2 + d\Tr(O^2)) + \frac{1}{N_s} \cO(\norm{O}_{\infty}^2)
    \end{equation}
\end{lemma}

We postpone the proof of Lemma~\ref{lem:variance_shadow_second_moment} and first discuss its implications. When $O = \mathbb{I}$, the classical-shadow estimator provides an estimation of state purity $\Tr(\rho^2)$, yielding $\Var[\omega_{\mathrm{CS}}] = \cO(\frac{d^2}{N_s^2} + \frac{1}{N_s})$.
To estimate $\Tr(\rho^2)$ to constant additive error with high probability, the variance should be suppressed to a constant. This implies a sample complexity of $N_s = \cO(d)$, which is worse than the $\cO(\sqrt{d})$ scaling of the RM protocol in Sec.~\ref{sec:pre}. For a general traceless observable $O_0$ satisfying $\Tr(O_0) = 0$, achieving a constant-precision estimate of $\Tr(O\rho^2_0)$ requires a sample complexity
\begin{equation}
     N_s = \cO\left(\sqrt{d\,\Tr(O_0^2)}\right).
\end{equation}

Before proving Lemma~\ref{lem:variance_shadow_second_moment}, we first establish the following auxiliary lemma.

\begin{lemma}[Auxiliary lemma]\label{lem:O2}
For any state $\rho$ and Hermitian observable $O$ act on $\rho$, $\Tr(\rho O \rho O) \le \norm{O^2}_{\infty}$.
\end{lemma}

\begin{proof}
Let $\rho=\sum\limits_i p_i\ketbra{v_i}$ be the spectral decomposition of $\rho$, then
\begin{equation}
\begin{split}
    \Tr(\rho O \rho O) &\le \norm{O \rho O}_{\infty} = \norm{\sum\limits_i p_i O \ketbra{v_i} O}_{\infty} \le \sum\limits_i p_i \norm{O\ketbra{v_i}O}_{\infty} \le \sum\limits_i p_i \norm{O^2}_{\infty} = \norm{O^2}_{\infty}.
\end{split}
\end{equation}
\end{proof}

\begin{proof}[Proof of Lemma~\ref{lem:variance_shadow_second_moment}] 
 We use the following properties of the classical shadow estimator. 
\begin{fact}[{\cite[Lemma S1]{huang2020predicting}}]\label{fact:var_o_CS}
Fix observable $O$ and set $\hat{o}=\Tr(O\hat{\rho})$, where $\rho$ is a classical shadow in Eq.~\eqref{eq:classical_shadow}. Then the variance of $\hat{o}$ satisfies
\begin{equation}
\Var[\hat{o}]\leq3\Tr(O^2).
\end{equation}
\end{fact}
\begin{fact}[{\cite[Lemma S5]{huang2020predicting}}]\label{fact:var_omega_CS}
The variance of $\omega_{\mathrm{CS}}$ satisfies
\begin{equation}
    \Var[\omega_{\mathrm{CS}}] \le \frac{4}{N_s^2} \Var[\Tr(O_{\sym} \hat{\rho}_1 \otimes \hat{\rho_2})] + \frac{4}{N_s} \Var[\Tr(O_{\sym} \hat{\rho}_1 \otimes \rho)],
\end{equation}
where $\hat{\rho}_1$ and $\hat{\rho}_2$ are two shadow estimators defined in Eq.~\eqref{eq:classical_shadow} constructed in independent experiments.
\end{fact}
\noindent We bound these two terms in Fact~\ref{fact:var_omega_CS} separately. For the first term, we use the following fact:
\begin{fact}[{\cite[Lemma S6]{huang2020predicting}}]
    For the symmetrized observable $O_{\sym}$,
    \begin{equation}
        \Var[\Tr(O_{\sym} \hat{\rho}_1 \otimes \hat{\rho}_2)] = \cO\left(\Tr(O_{\sym}^2)\right)
    \end{equation}
\end{fact}
\noindent Incorporating Eq.~\eqref{eq:O_sym} into the above, we find
\begin{equation}\label{eq:CS_variance_1}
    \frac{4}{N_s^2} \Var[\Tr(O_{\sym} \hat{\rho}_1 \otimes \hat{\rho_2})] = \frac{1}{N_s^2}\cO(\Tr(O_\sym^2)) = \frac{1}{N_s^2}\cO(\Tr(O)^2 + d \Tr(O^2)).
\end{equation}
\noindent For the second term, we have
\begin{equation}\label{eq:CS_variance_2}
\begin{split}
    \frac{4}{N_s}\Var[\Tr(O_{\sym} \hat{\rho}_1 \otimes \rho)]
    &= \frac{2}{N_s} \Var[\Tr((\rho O + O\rho) \hat{\rho}_1)]\\
    &= \frac{1}{N_s} \cO(\Tr(\rho O \rho O + O^2\rho^2)) \\
    &= \frac{1}{N_s} \cO(\norm{O}_{\infty}^2)
\end{split}
\end{equation}
where in the second step we
have used  Fact~\ref{fact:var_o_CS} and in the last step we apply Lemma~\ref{lem:O2}.
Combining Eqs.~\eqref{eq:CS_variance_1} and~\eqref{eq:CS_variance_2}, we obtain the stated variance bound for $\omega_{\mathrm{CS}}$.
\end{proof}

\section{Additional analysis on estimating dichotomic observables}\label{app:observable-driven_analysis}
Here, we provide the analysis of the protocols introduced in Section~\ref{sec:dichotomic}. First, we present the proof of Lemma~\ref{lem:estimating_dichotomic}. Then, we give a decomposition of the dichotomic observable that leads to the sample complexity stated in Theorem~\ref{thm:good_dicho_estimation}.

\subsection{Proof of Lemma \ref{lem:estimating_dichotomic}}\label{app:proof_first_theorem}
First, we analyze the expectation value of the estimator $\omega$ defined in Eq.~\eqref{eq:unbiased_estimator}.

\begin{lemma}[Expectation value of estimator $\omega$]\label{lem:expectation_value_w}
    Given $N_M$ and a dichotomic observable $O$, the estimator $\omega_U$ defined in Eq.~\eqref{eq:unbiased_estimator} is unbiased, satisfying 
    \begin{equation}
        \mathbb{E}[\omega_U] = \Tr(O\rho^2).
    \end{equation}
\end{lemma}

\begin{proof}
First consider a general observable $O$ in the Hilbert space of an $n$-qubit system $\mathcal{H}_n = (\mathbb{C}^2)^{\otimes n}$. Let
\begin{equation}
O = \sum\limits_{\ell=1}^L \lambda_\ell P_\ell
\end{equation}
be its spectral decomposition, where $P_\ell$ is the projector onto the eigenspace $\Pi_\ell$ corresponding to the eigenvalue $\lambda_\ell$. Our protocol proceeds as follows. For each eigenspace $\Pi_\ell$, we draw a Haar-random unitary $\tilde{U}_\ell$, and then apply their direct sum
\begin{equation}\label{eq:direct_sum_U}
U = \bigoplus_{\ell=1}^L \tilde{U}_\ell
\end{equation}
to the state $\rho$. Next, we measure the resulting state on the eigenbasis of $O$ and record the measurement outcomes $\{s_1,\ldots,s_{N_M}\}$. Suppose we generate these random unitaries $\{\tilde{U}_\ell\}_{\ell=1}^L$ for $N_U$ times, and in each time we perform $N_M$ repeated measurements. For each unitary $U$, we construct the estimator
\begin{equation}\label{eq:direct_sum_estimator}
\omega_U
= \binom{N_M}{2}^{-1} 
  \sum\limits_{\{i,j\} \in \binom{[N_M]}{2}} 
  X_2^{\mathrm{b}}(s_i,s_j),
\end{equation}
where
\begin{equation}
X_2^{\mathrm{b}}(s_i,s_j) =
\begin{cases}
\lambda_\ell \,X_2^\ell(s_i,s_j), & s_i,s_j \in \Pi_\ell \text{ for some } \ell \in [L],\\
0, & \text{otherwise}.
\end{cases}
\end{equation}
Here, $\{i,j\} \in \binom{[N_M]}{2}$ 
means $1 \le i < j \le N_M, i\neq j$, and $X_2^\ell(s_i,s_j)$ is defined in Eq.~\eqref{eq:RRM_coef_purity}, with $d$ replaced by $d_\ell$, the dimension of $\Pi_\ell$. The expectation value of $\omega_U$ is
\begin{equation}\label{eq:direct_sum_exp_1}
\begin{aligned}
\mathbb{E}\left[\omega_U\right]&=\underset{U}{\mathbb{E}}\left[\sum\limits_{s_1,s_2}\langle s_1|U\rho U^\dagger|s_1\rangle\langle s_2|U\rho U^\dagger|s_2\rangle X_2^{\mathrm{b}}(s_1,s_2)\right]
\\&=\underset{U}{\mathbb{E}}\left[\sum\limits_{\ell=1}^L\sum\limits_{s_1,s_2}\langle s_1|P_\ell U\rho U^\dagger P_\ell|s_1\rangle\langle s_2|P_\ell U\rho U^\dagger P_\ell|s_2\rangle \lambda_\ell X_2^\ell(s_1,s_2)\right].
\end{aligned}
\end{equation}
Note that $U$, constructed as in Eq.~\eqref{eq:direct_sum_U}, commutes with each projector $P_\ell$ (and thus commutes with $O$). Define $U_\ell = \tilde{U}_\ell \oplus \mathbf{0}$ to be the unitary acting on the whole space $\mathcal{H}_n$ but keeping only terms in subspace $\Pi_\ell$, and denote
\begin{equation}
\tilde{\rho}_\ell = (P_\ell\rho P_\ell)\bigl\vert_{\Pi_\ell}, 
\quad
\vert \tilde{s}\rangle = \ket{s}\bigl\vert_{\Pi_\ell}\text{ for }s\in\Pi_\ell.
\end{equation}
Then it follows that
\begin{equation}
\begin{aligned}
U_\ell&=P_\ell UP_\ell=P_\ell U=UP_\ell,\\
U_\ell^\dagger&=P_\ell U^\dagger P_\ell=P_\ell U^\dagger=U^\dagger P_\ell.
\end{aligned}
\end{equation}
Using these definitions, we can simplify the terms in Eq.~\eqref{eq:direct_sum_exp_1}:
\begin{equation}\label{eq:direct_sum_exp_2}
\begin{aligned}
\mathbb{E}\left[\omega_U\right]&=\underset{\{U_\ell\}}{\mathbb{E}}\left[\sum\limits_{\ell=1}^L\sum\limits_{s_1,s_2}\langle s_1|U_\ell\rho U_\ell^\dagger |s_1\rangle\langle s_2|U_\ell\rho U_\ell^\dagger|s_2\rangle \lambda_\ell X_2^\ell(s_1,s_2)\right]
\\&=\underset{\{\tilde{U}_\ell\}}{\mathbb{E}}\left[\sum\limits_{\ell=1}^L\sum\limits_{\tilde{s}_1,\tilde{s}_2}\langle \tilde{s}_1|\tilde{U}_\ell\tilde{\rho}_\ell \tilde{U}_\ell^\dagger |\tilde{s}_1\rangle\langle \tilde{s}_2|\tilde{U}_\ell\tilde{\rho}_\ell \tilde{U}_\ell^\dagger|\tilde{s}_2\rangle \lambda_\ell X_2^\ell(\tilde{s}_1,\tilde{s}_2)\right]
\\&=\sum\limits_{\ell=1}^L\lambda_\ell\Tr(\tilde{\rho}_\ell^2).
\end{aligned}
\end{equation}
In the last step, we use Eq.~\eqref{eq:purity_expectation}. When $O$ is a dichotomic observable with eigenvalues $\pm 1$, Eq.~\eqref{eq:direct_sum_exp_2} simplifies to
\begin{equation}\label{eq:Pauli_exp}
\mathbb{E}\left[\omega_U\right]=\Tr(\tilde{\rho}_+^2)-\Tr(\tilde{\rho}_-^2) = \Tr(O\rho^2),
\end{equation}
where the last equality uses Eq.~\eqref{eq:observation_block}.
\end{proof}

Furthermore, we could sample $U$ for $N_U$ times independently and construct $\omega_u$ according to Eq.~\eqref{eq:unbiased_estimator} for each $u\in[N_U]$. Then their average $\omega$ satisfies 
\begin{equation}
    \mathbb{E}[\omega] = \frac{1}{N_U} \sum\limits_{u=1}^{N_U} \mathbb{E}[\omega_u] = \Tr(O\rho^2),
\end{equation}
thus is also an unbiased estimator. We establish the variance of the estimator $\omega$. 

\begin{lemma}[Variance of estimator $\omega$ with block-diagonal unitary 4-design]\label{lem:variance_w}
    Given $N_U, N_M$, and a dichotomic observable $O$, the estimator $\omega$ defined in Eq.~\eqref{eq:unbiased_estimator}, obtained by following Protocol~\ref{protocol:dichotomic} with $\tilde{U}_\pm$ sampled from a unitary 4-design, has variance
    \begin{equation}
        \Var(\omega) = \cO\left(\frac{1}{N_U}\bigl(\frac{d}{N_M^2} + \frac{1}{N_M} + \frac{1}{d_+} + \frac{1}{d_-})\right),
    \end{equation}
    where $d_+$ and $d_-$ denote the dimensions of the eigenspaces $\Pi_+$ and $\Pi_-$ of $O$, respectively, and $d=d_++d_-$.
\end{lemma}
\begin{proof}
We calculate the variance of the estimator $\omega_u$ (see  Eq.~\eqref{eq:direct_sum_estimator}), to get
\begin{equation}\label{eq:var_omega_u}
\begin{aligned}
\Var(\omega_u)&=\mathbb{E}[\omega_u^2]-\mathbb{E}[\omega_u]^2
\\&={\binom{N_M}2}^{-2}\sum\limits_{\substack{\{i,j\}\in\binom{[N_M]}2\\\{i',j'\}\in\binom{[N_M]}2}}\mathbb{E}\left[X_2^{\mathrm{b}}(s_i,s_j)X_2^{\mathrm{b}}(s_{i'},s_{j'})\right]-\mathbb{E}[\omega_u]^2.
\end{aligned}
\end{equation}

Define $\Delta(i,j;i',j')=\left|\{i,j\}\cup\{i',j'\}\right|\in\{2,3,4\}$ and split the above summation into three terms
\begin{equation}
\begin{aligned}
\sum\limits_{\substack{\{i,j\}\in\binom{[N_M]}2\\\{i',j'\}\in\binom{[N_M]}2}}&=\sum\limits_{\substack{\{i,j\}\in\binom{[N_M]}2\\\{i',j'\}\in\binom{[N_M]}2\\\Delta(i,j;i',j')=2}}+\sum\limits_{\substack{\{i,j\}\in\binom{[N_M]}2\\\{i',j'\}\in\binom{[N_M]}2\\\Delta(i,j;i',j')=3}}+\sum\limits_{\substack{\{i,j\}\in\binom{[N_M]}2\\\{i',j'\}\in\binom{[N_M]}2\\\Delta(i,j;i',j')=4}}.
\end{aligned}
\end{equation}

\begin{itemize}
\item The first term is
\begin{equation}\label{eq:Delta=2}
\begin{aligned}
&\sum\limits_{\substack{\{i,j\}\in\binom{[N_M]}2\\\{i',j'\}\in\binom{[N_M]}2\\\Delta(i,j;i',j')=2}}\mathbb{E}\left[X_2^{\mathrm{b}}(s_i,s_j)X_2^{\mathrm{b}}(s_{i'},s_{j'})\right]
\\=&\binom{N_M}2\underset{U_+}{\mathbb{E}}\left[\sum\limits_{s_1,s_2}\langle s_1|U_+\rho U_+^\dagger|s_1\rangle\langle s_2|U_+\rho U_+^\dagger|s_2\rangle X_2^+(s_1,s_2)^2\right]+\left(+\Longleftrightarrow-\right)
\\=&\binom{N_M}2\underset{\tilde{U}_+}{\mathbb{E}}\left[\sum\limits_{\tilde{s}_1,\tilde{s}_2}\langle \tilde{s}_1|\tilde{U}_+\tilde{\rho}_+ \tilde{U}_+^\dagger|\tilde{s}_1\rangle\langle \tilde{s}_2|\tilde{U}_+\tilde{\rho}_+ \tilde{U}_+^\dagger|\tilde{s}_2\rangle X_2^+(\tilde{s}_1,\tilde{s}_2)^2\right]+\left(+\Longleftrightarrow-\right)
\\=&\binom{N_M}2\underset{\tilde{U}_+}{\mathbb{E}}\left[\Tr\left(\tilde{U}_+^{\dagger\otimes2}\left(\sum\limits_{\tilde{s}_1,\tilde{s}_2}X_2^+(\tilde{s}_1,\tilde{s}_2)^2|\tilde{s}_1\tilde{s}_2\rangle\langle\tilde{s}_1\tilde{s}_2|\right)\tilde{U}_+^{\otimes2}\tilde{\rho}_+^{\otimes2}\right)\right]+\left(+\Longleftrightarrow-\right),
\end{aligned}
\end{equation}
where $\left(+\Longleftrightarrow-\right)$ is obtained by replacing all the label ``$+$'' with label ``$-$'' (and vice versa) in the previous terms.
We 
define
\begin{equation}
X_+^{(1,2)2}=\sum\limits_{\tilde{s}_1,\tilde{s}_2}X_2^+(\tilde{s}_1,\tilde{s}_2)^2|\tilde{s}_1\tilde{s}_2\rangle\langle\tilde{s}_1\tilde{s}_2|
\end{equation}
and calculate its second-order twirling~\cite{zhou2020Single}
\begin{equation}
\Phi_2(X_+^{(1,2)2})=d_+\mathbb{I}+(d_+-1)\mathbb{S},
\end{equation}
where $\mathbb{I}$ and $\mathbb{S}$ are respectively the identity and swap operator acting on $\Pi_+^{\otimes2}$. In this way, the contribution of this term becomes
\begin{equation}\label{eq:A_2}
A_{2}=\binom{N_M}2\left(d_+\Tr(\tilde{\rho}_+)^2+(d_+-1)\Tr(\tilde{\rho}_+^2)\right)+\left(+\Longleftrightarrow-\right).
\end{equation}

\item The second term is
\begin{equation}\label{eq:Delta=2}
\begin{aligned}
&\sum\limits_{\substack{\{i,j\}\in\binom{[N_M]}2\\\{i',j'\}\in\binom{[N_M]}2\\\Delta(i,j;i',j')=3}}\mathbb{E}\left[X_2^{\mathrm{b}}(s_i,s_j)X_2^{\mathrm{b}}(s_{i'},s_{j'})\right]
\\
=&6\binom{N_M}3\underset{U_+}{\mathbb{E}}\Biggl[\sum\limits_{s_1,s_2,s_3}\langle s_1|U_+\rho U_+^\dagger|s_1\rangle\langle s_2|U_+\rho U_+^\dagger|s_2\rangle\langle s_3|U_+\rho U_+^\dagger|s_3\rangle X_2^+(s_1,s_2)X_2^+(s_1,s_3)\Biggr]+\left(+\Longleftrightarrow-\right)
\\=&6\binom{N_M}3\underset{\tilde{U}_+}{\mathbb{E}}\left[\sum\limits_{\tilde{s}_1,\tilde{s}_2,\tilde{s}_3}\langle \tilde{s}_1|\tilde{U}_+\tilde{\rho}_+ \tilde{U}_+^\dagger|\tilde{s}_1\rangle\langle \tilde{s}_2|\tilde{U}_+\tilde{\rho}_+ \tilde{U}_+^\dagger|\tilde{s}_2\rangle\langle \tilde{s}_3|\tilde{U}_+\tilde{\rho}_+ \tilde{U}_+^\dagger|\tilde{s}_3\rangle X_2^+(\tilde{s}_1,\tilde{s}_2)X_2^+(\tilde{s}_1,\tilde{s}_3)\right]+\left(+\Longleftrightarrow-\right)
\\=&6\binom{N_M}3\underset{\tilde{U}_+}{\mathbb{E}}\left[\Tr\left(\tilde{U}_+^{\dagger\otimes3}\left(\sum\limits_{\tilde{s}_1,\tilde{s}_2,\tilde{s}_3}X_2^+(\tilde{s}_1,\tilde{s}_2)X_2^+(\tilde{s}_1,\tilde{s}_3)|\tilde{s}_1\tilde{s}_2\tilde{s}_3\rangle\langle\tilde{s}_1\tilde{s}_2\tilde{s}_3|\right)\tilde{U}_+^{\otimes3}\tilde{\rho}_+^{\otimes3}\right)\right]+\left(+\Longleftrightarrow-\right).
\end{aligned}
\end{equation}
We further define the operator
\begin{equation}
X_+^{(1,2),(1,3)}=\sum\limits_{\tilde{s}_1,\tilde{s}_2,\tilde{s}_3}X_2^+(\tilde{s}_1,\tilde{s}_2)X_2^+(\tilde{s}_1,\tilde{s}_3)|\tilde{s}_1\tilde{s}_2\tilde{s}_3\rangle\langle\tilde{s}_1\tilde{s}_2\tilde{s}_3|
\end{equation}
and calculate its third-order twirling~\cite{zhou2020Single}
\begin{equation}
\Phi_3(X_+^{(1,2),(1,3)})=-\frac1{d_++2}\left(\mathbb{P}^3_\mathbb{I}+\mathbb{P}^3_{(12)}+\mathbb{P}^3_{(13)}\right)+\frac{d_++1}{d_++2}\left(\mathbb{P}^3_{(23)}+\mathbb{P}^3_{(123)}+\mathbb{P}^3_{(132)}\right).
\end{equation}
Here $\mathbb{P}^3_\pi$ is the permutation operator corresponding to $\pi\in S_3$ acting on $\Pi_+^{\otimes3}$ and we write $\pi$ in its cycle decomposition. Notice that each $l$-cycle contributes a factor $\Tr(\tilde{\rho}_+^l)$, so the summation simplifies to
\begin{equation}\label{eq:A_3}
A_{3}=6\binom{N_M}3\left(\frac{d_+-1}{d_++2}\Tr(\tilde{\rho}_+^2)\Tr(\tilde{\rho}_+)+\frac{2(d_++1)}{d_++2}\Tr(\tilde{\rho}_+^3)-\frac1{d_++2}\Tr(\tilde{\rho}_+)^3\right)+\left(+\Longleftrightarrow-\right).
\end{equation}

\item The last term is
given by
\begin{equation}\label{eq:Delta=2}
\begin{aligned}
&\sum\limits_{\substack{\{i,j\}\in\binom{[N_M]}2\\\{i',j'\}\in\binom{[N_M]}2\\\Delta(i,j;i',j')=4}}\mathbb{E}\left[X_2^{\mathrm{b}}(s_i,s_j)X_2^{\mathrm{b}}(s_{i'},s_{j'})\right]
\\
=&6\binom{N_M}4\underset{U_+}{\mathbb{E}}\Biggl[\sum\limits_{s_1,s_2,s_3,s_4}\langle s_1|U_+\rho U_+^\dagger|s_1\rangle\langle s_2|U_+\rho U_+^\dagger|s_2\rangle\langle s_3|U_+\rho U_+^\dagger|s_3\rangle\langle s_4|U_+\rho U_+^\dagger|s_4\rangle X_2^+(s_1,s_2)X_2^+(s_3,s_4)\Biggr]
\\&-6\binom{N_M}4\underset{U_+,U_-}{\mathbb{E}}\Biggl[\sum\limits_{s_1,s_2,s_3,s_4}\langle s_1|U_+\rho U_+^\dagger|s_1\rangle\langle s_2|U_+\rho U_+^\dagger|s_2\rangle\langle s_3|U_-\rho U_-^\dagger|s_3\rangle\langle s_4|U_-\rho U_-^\dagger|s_4\rangle X_2^+(s_1,s_2)X_2^-(s_3,s_4)\Biggr]
\\&+\left(+\Longleftrightarrow-\right)
\\=&6\binom{N_M}4\underset{\tilde{U}_+}{\mathbb{E}}\left[\sum\limits_{\tilde{s}_1,\tilde{s}_2,\tilde{s}_3,\tilde{s}_4}\langle \tilde{s}_1|\tilde{U}_+\tilde{\rho}_+ \tilde{U}_+^\dagger|\tilde{s}_1\rangle\langle \tilde{s}_2|\tilde{U}_+\tilde{\rho}_+ \tilde{U}_+^\dagger|\tilde{s}_2\rangle\langle \tilde{s}_3|\tilde{U}_+\tilde{\rho}_+ \tilde{U}_+^\dagger|\tilde{s}_3\rangle\langle \tilde{s}_4|\tilde{U}_+\tilde{\rho}_+ \tilde{U}_+^\dagger|\tilde{s}_4\rangle X_2^+(\tilde{s}_1,\tilde{s}_2)X_2^+(\tilde{s}_3,\tilde{s}_4)\right]
\\&-6\binom{N_M}4\underset{\tilde{U}_+,\tilde{U}_-}{\mathbb{E}}\left[\sum\limits_{\tilde{s}_1,\tilde{s}_2,\tilde{s}_3,\tilde{s}_4}\langle \tilde{s}_1|\tilde{U}_+\tilde{\rho}_+ \tilde{U}_+^\dagger|\tilde{s}_1\rangle\langle \tilde{s}_2|\tilde{U}_+\tilde{\rho}_+ \tilde{U}_+^\dagger|\tilde{s}_2\rangle\langle \tilde{s}_3|\tilde{U}_-\tilde{\rho}_- \tilde{U}_-^\dagger|\tilde{s}_3\rangle\langle \tilde{s}_4|\tilde{U}_-\tilde{\rho}_- \tilde{U}_-^\dagger|\tilde{s}_4\rangle X_2^+(\tilde{s}_1,\tilde{s}_2)X_2^-(\tilde{s}_3,\tilde{s}_4)\right]
\\&+\left(+\Longleftrightarrow-\right)
\\=&6\binom{N_M}4\underset{\tilde{U}_+}{\mathbb{E}}\left[\Tr\left(\tilde{U}_+^{\dagger\otimes4}\left(\sum\limits_{\tilde{s}_1,\tilde{s}_2,\tilde{s}_3,\tilde{s}_4}X_2^+(\tilde{s}_1,\tilde{s}_2)X_2^+(\tilde{s}_3,\tilde{s}_4)|\tilde{s}_1\tilde{s}_2\tilde{s}_3\tilde{s}_4\rangle\langle\tilde{s}_1\tilde{s}_2\tilde{s}_3\tilde{s}_4|\right)\tilde{U}_+^{\otimes4}\tilde{\rho}_+^{\otimes4}\right)\right]
\\&-6\binom{N_M}4\underset{\tilde{U}_+,\tilde{U}_-}{\mathbb{E}}\left[\Tr\left(\tilde{U}_+^{\dagger\otimes2}\left(\sum\limits_{\tilde{s}_1,\tilde{s}_2}X_2^+(\tilde{s}_1,\tilde{s}_2)|\tilde{s}_1\tilde{s}_2\rangle\langle\tilde{s}_1\tilde{s}_2|\right)\tilde{U}_+^{\otimes2}\tilde{\rho}_+^{\otimes2}\right)\Tr\left(\tilde{U}_-^{\dagger\otimes2}\left(\sum\limits_{\tilde{s}_3,\tilde{s}_4}X_2^+(\tilde{s}_3,\tilde{s}_4)|\tilde{s}_3\tilde{s}_4\rangle\langle\tilde{s}_3\tilde{s}_4|\right)\tilde{U}_-^{\otimes2}\tilde{\rho}_-^{\otimes2}\right)\right]
\\&+\left(+\Longleftrightarrow-\right).
\end{aligned}
\end{equation}
We define
\begin{equation}\label{eq:X_+_tensor_2}
X_+^{(1,2)\otimes2}=\sum\limits_{\tilde{s}_1,\tilde{s}_2,\tilde{s}_3,\tilde{s}_4}X_2^+(\tilde{s}_1,\tilde{s}_2)X_2^+(\tilde{s}_3,\tilde{s}_4)|\tilde{s}_1\tilde{s}_2\tilde{s}_3\tilde{s}_4\rangle\langle\tilde{s}_1\tilde{s}_2\tilde{s}_3\tilde{s}_4|,
\end{equation}
and its fourth-order twirling is~\cite{liu2022permutation}
\begin{equation}\label{eq:4-order_twirling_X_tensor_2}
\begin{aligned}
\Phi_4(X_+^{(1,2)\otimes2})=&\left(1+\frac2{d_+(d_++2)(d_++3)}\right)\mathbb{P}^4_{(12)(34)}\\
&+\frac2{d_+(d_++2)(d_++3)}\left(\mathbb{P}^4_{\mathbb{I}}+\mathbb{P}^4_{(12)}+\mathbb{P}^4_{(34)}\right)\\
&+\frac{d_++1}{d_+(d_++3)}\left(\mathbb{P}^4_{(13)(24)}+\mathbb{P}^4_{(14)(23)}+\mathbb{P}^4_{(1324)}+\mathbb{P}^4_{(1423)}\right)\\
&-\frac{d_++1}{d_+(d_++2)(d_++3)}\left(\text{the other elements in }S_4\right).
\end{aligned}
\end{equation}
Where ``the other elements in $S_4$'' contain all the elements not listed previously. Represented by the lengths of cycles in cycle decomposition, they consist of four of type $1+1+2$, eight of type $1+3$, and four of type $4$.
Consequently, the contribution of this summation is
\begin{equation}\label{eq:A_4}
\begin{aligned}
A_{4}=&6\binom{N_M}4\Biggl[\frac2{d_+(d_++2)(d_++3)}\Tr(\tilde{\rho}_+)^4-\frac4{(d_++2)(d_++3)}\Tr(\tilde{\rho}_+)^2\Tr(\tilde{\rho}_+^2)\\
&-\frac{8(d_++1)}{d_+(d_++2)(d_++3)}\Tr(\tilde{\rho}_+)\Tr(\tilde{\rho}_+^3)+\left(1+\frac{2d_+^2+6d_++6}{d_+(d_++2)(d_++3)}\right)\Tr(\tilde{\rho}_+^2)^2\\
&+\frac{2(d_++1)}{(d_++2)(d_++3)}\Tr(\tilde{\rho}_+^4)-\Tr(\tilde{\rho}_+^2)\Tr(\tilde{\rho}_-^2)\Biggr]+\left(+\Longleftrightarrow-\right).
\end{aligned}
\end{equation}
\end{itemize}

Inserting Eqs.~\eqref{eq:A_2},~\eqref{eq:A_3},~\eqref{eq:A_4}, and~\eqref{eq:Pauli_exp} to Eq.~\eqref{eq:var_omega_u}, we have
\begin{equation}\label{eq:variance_final}
\begin{aligned}
\Var(\omega_u)=&{\binom{N_M}2}^{-2}\left(A_2+A_3+A_4\right)-\left(\Tr(\tilde{\rho}_+^2)-\Tr(\tilde{\rho}_-^2)\right)^2
\\=&\frac{2d_+}{N_M(N_M-1)}\Tr(\tilde{\rho}_+)^2+\frac{2d_-}{N_M(N_M-1)}\Tr(\tilde{\rho}_-)^2
\\&+\frac{2(d_+-1)}{N_M(N_M-1)}\Tr(\tilde{\rho}_+^2)+\frac{2(d_--1)}{N_M(N_M-1)}\Tr(\tilde{\rho}_-^2)
\\&+\frac{4(d_+-1)(N_M-2)}{(d_++2)N_M(N_M-1)}\Tr(\tilde{\rho}_+^2)\Tr(\tilde{\rho}_+)+\frac{4(d_--1)(N_M-2)}{(d_-+2)N_M(N_M-1)}\Tr(\tilde{\rho}_-^2)\Tr(\tilde{\rho}_-)
\\&+\frac{8(d_++1)(N_M-2)}{(d_++2)N_M(N_M-1)}\Tr(\tilde{\rho}_+^3)+\frac{8(d_-+1)(N_M-2)}{(d_-+2)N_M(N_M-1)}\Tr(\tilde{\rho}_-^3)
\\&-\frac{4(N_M-2)}{(d_++2)N_M(N_M-1)}\Tr(\tilde{\rho}_+)^3-\frac{4(N_M-2)}{(d_-+2)N_M(N_M-1)}\Tr(\tilde{\rho}_-)^3
\\&+\frac{2(N_M-2)(N_M-3)}{d_+(d_++2)(d_++3)N_M(N_M-1)}\Tr(\tilde{\rho}_+)^4+\frac{2(N_M-2)(N_M-3)}{d_-(d_-+2)(d_-+3)N_M(N_M-1)}\Tr(\tilde{\rho}_-)^4
\\&-\frac{4(N_M-2)(N_M-3)}{(d_++2)(d_++3)N_M(N_M-1)}\Tr(\tilde{\rho}_+)^2\Tr(\tilde{\rho}_+^2)-\frac{4(N_M-2)(N_M-3)}{(d_-+2)(d_-+3)N_M(N_M-1)}\Tr(\tilde{\rho}_-)^2\Tr(\tilde{\rho}_-^2)
\\&-\frac{8(d_++1)(N_M-2)(N_M-3)}{d_+(d_++2)(d_++3)N_M(N_M-1)}\Tr(\tilde{\rho}_+)\Tr(\tilde{\rho}_+^3)-\frac{8(d_-+1)(N_M-2)(N_M-3)}{d_-(d_-+2)(d_-+3)N_M(N_M-1)}\Tr(\tilde{\rho}_-)\Tr(\tilde{\rho}_-^3)
\\&+\frac{(2d_+^2+6d_++6)(N_M-2)(N_M-3)}{d_+(d_++2)(d_++3)N_M(N_M-1)}\Tr(\tilde{\rho}_+^2)^2+\frac{(2d_-^2+6d_-+6)(N_M-2)(N_M-3)}{d_-(d_-+2)(d_-+3)N_M(N_M-1)}\Tr(\tilde{\rho}_-^2)^2
\\&+\frac{2(d_++1)(N_M-2)(N_M-3)}{(d_++2)(d_++3)N_M(N_M-1)}\Tr(\tilde{\rho}_+^4)+\frac{2(d_-+1)(N_M-2)(N_M-3)}{(d_-+2)(d_-+3)N_M(N_M-1)}\Tr(\tilde{\rho}_-^4)
\\&-\frac{4N_M-6}{N_M(N_M-1)}\left[\Tr(\tilde{\rho}_+^2)-\Tr(\tilde{\rho}_-^2)\right]^2.
\end{aligned}
\end{equation}
Notice that $\tilde{\rho}_\pm$ are principal submatrices of $\rho$, so that they are positive semi-definite and $\Tr(\tilde{\rho}_\pm^l)\leq1$ for $l\in\mathbb{N}_+$. We obtain the upper bound for $\Var(\omega_u)$
\begin{equation}
\begin{aligned}
\Var(\omega_u)\leq&\frac{4}{N_M(N_M-1)}(d_++d_-)+\frac{4(N_M-2)}{N_M(N_M-1)}\left(\frac{3d_++1}{d_++2}+\frac{3d_-+1}{d_-+2}\right)
\\&+\frac{4(N_M-2)(N_M-3)}{N_M(N_M-1)}\left(\frac{d_+^2+2d_++2}{d_+(d_++2)(d_++3)}+\frac{d_-^2+2d_-+2}{d_-(d_-+2)(d_-+3)}\right)
\\=&\cO\left(\frac{d}{N_M^2} + \frac{1}{N_M} + \frac{1}{d_+} + \frac{1}{d_{-}}\right).
\end{aligned}
\end{equation}

Furthermore, suppose we draw $N_U$ independent random unitaries and construct an estimator by averaging the individual estimators obtained for each choice of random unitary, i.e.,
\begin{equation}
\omega =\frac{1}{N_U}\sum\limits_{u=1}^{N_U}\omega_u.
\end{equation}
Then we can estimate $\Tr(O \rho^2)$ unbiasedly, with the variance given by
\begin{equation}
\Var(\omega) = \frac{1}{N_U}\Var(\omega_u) = \frac{1}{N_U} \cO\left(\frac{d}{N_M^2} + \frac{1}{N_M} +  \frac{1}{d_+} + \frac{1}{d_{-}}\right).
\end{equation}
% \lzh{
% Following the same logic, when $d_-=0$ or $d_+=0$, the variance reduces to the variance of estimating purity~\cite{anshu2022distributed}
% \begin{equation}
% \Var(\omega) = \frac{1}{N_U}\Var(\omega_u) = \frac{1}{N_U} \cO\left(\frac{d}{N_M^2} + \frac{1}{N_M} +  \frac{1}{d}\right).
% \end{equation}
% }
\end{proof}
Based on Lemmas~\ref{lem:median-of-mean},~\ref{lem:expectation_value_w}, and~\ref{lem:variance_w}, we can now prove Lemma~\ref{lem:estimating_dichotomic} in the main text.

\begin{proof}[Proof of Lemma~\ref{lem:estimating_dichotomic}]
To apply Lemma~\ref{lem:median-of-mean}, it is sufficient to suppress $\Var(\omega)\leq\varepsilon^2/3<\varepsilon^2/2$. Define $\eta=\max\left\{\frac{1}{d_+},\frac{1}{d_-}\right\}=\Omega\left(\frac{1}{d}\right)$.
When $\varepsilon\ge\Omega\left(\eta^{1/2}\right)$, we can set $N_U=\cO(1)$ and $N_M=\cO\left(\frac{d^{1/2}}{\varepsilon}\right)$ and the total sample complexity to be $N_U\times N_M=\cO\left(\frac{d^{1/2}}{\varepsilon}\right)$.
When $\varepsilon\le\cO\left(\eta^{1/2}\right)$, we should set $N_U=\cO\left(\frac{\eta}{\varepsilon^2}\right)$ and $N_M=\cO\left(\frac{d^{1/2}}{\eta^{1/2}}\right)$, which results in the total sample complexity 
$N_U\times N_M=\cO\left(\frac{d^{1/2}\eta^{1/2}}{\varepsilon^2}\right)$.
In conclusion, the sample complexity is upper bounded by $\cO\left(\max\left\{\frac{\sqrt{d}}{\varepsilon},\frac{\sqrt{d\eta}}{\varepsilon^2}\right\}\right)$.
\end{proof}

\subsection{Proof of Lemma~\ref{lem:estimating_dichotomic_Clifford}}\label{app:proof_dichotomic_Clifford}

In this section, we prove Lemma~\ref{lem:estimating_dichotomic_Clifford}. Notice that it only differs from Lemma~\ref{lem:estimating_dichotomic} by drawing $\tilde{U}_\pm$ randomly from the group of multi-qubit Clifford gates (which is a unitary 3-design~\cite{zhu2017multiqubit} but not a 4-design~\cite{zhu2016clifford}), the proof is similar to that in Appendix~\ref{app:proof_first_theorem}. The unbiasedness of $\omega$ follows directly from Lemma~\ref{lem:expectation_value_w}. For the variance, we obtain the following upper bound (cf. Lemma~\ref{lem:variance_w}):

\begin{lemma}[Variance of estimator $\omega$ with block-diagonal Clifford gates]\label{lem:variance_w_Clifford}
    Given $N_U, N_M$, and a traceless dichotomic observable $O$, the estimator $\omega$ defined in Eq.~\eqref{eq:unbiased_estimator}, obtained by following Protocol~\ref{protocol:dichotomic} with $\tilde{U}_\pm$ sampled from the multi-qubit Clifford group, has variance
    \begin{equation}
        \Var^{\operatorname{Cl}}(\omega) = \cO\left(\frac{1}{N_U}\bigl(\frac{d}{N_M^2} + \frac{1}{N_M} + 1)\right),
    \end{equation}
    where $d$ denotes the dimension of the qubit system.
\end{lemma}

\subsubsection{Notations}
We first introduce some notations for a given state $\rho$ that will be used in the proof. We define the state moments $R_k=\Tr(\rho^k)$ for $k\in\mathbb{N}_+$. It is direct to see $R_k\leq1$ for any $k$. We denote the $n$-qubit Pauli group as $\mathcal{P}_n=\{I,X,Y,Z\}^{\otimes n}$, the projector $Q=\frac1{d^2}\sum\limits_{P\in\mathcal{P}_n}P^{\otimes4}$, and its complementary projector $Q^\perp=I-Q$. For a 4-order permutation operator $\sigma\in S_4$, define $Q_\sigma=\Tr\left(\mathbb{P}_\sigma^4Q\rho^{\otimes4}\right)$. Since $Q_\sigma$ is determined by the cycle type $\lambda$ of $\sigma$, we also denote $Q_\lambda=Q_\sigma$:

\begin{equation}
\begin{aligned}
Q_{1,1,1,1}&=\frac1{d^2}\sum\limits_{P\in\mathcal{P}_n}\Tr(P\rho)^4,\\
Q_{2,1,1}&=\frac1{d^2}\sum\limits_{P\in\mathcal{P}_n}\Tr((P\rho)^2)\Tr(P\rho)^2,\\
Q_{2,2}&=\frac1{d^2}\sum\limits_{P\in\mathcal{P}_n}\Tr((P\rho)^2)^2,\\
Q_{3,1}&=\frac1{d^2}\sum\limits_{P\in\mathcal{P}_n}\Tr((P\rho)^3)\Tr(P\rho),\\
Q_{4}&=\frac1{d^2}\sum\limits_{P\in\mathcal{P}_n}\Tr((P\rho)^4).
\end{aligned}
\end{equation}

We obtain a tight upper bound for these quantities as follows.

\begin{lemma}[Upper bound of $Q_\sigma$]\label{lem:upper_bound_Q_sigma}
For $n$-qubit state $\rho$ and $\sigma\in S_4$, we have $Q_\sigma=\Tr\left(\mathbb{P}_\sigma^4Q\rho^{\otimes4}\right)\leq\frac1d$, where $d=2^n$. The bound is saturated when $\rho$ is a pure stabilizer state.
\end{lemma}

\begin{proof}[Proof of Lemma~\ref{lem:upper_bound_Q_sigma}]
Let the spectral decomposition be $\rho=\sum\limits_ip_i\ket{\psi_i}\bra{\psi_i}$. Then we have
\begin{equation}
Q_\sigma=\sum\limits_{i,j,k,l}p_ip_jp_kp_l\Tr\left(\mathbb{P}_\sigma^4Q\Psi_{ijkl}\right),
\end{equation}
where $\Psi_{ijkl}=\bigotimes\limits_{\alpha=i,j,k,l}\ket{\psi_\alpha}\bra{\psi_\alpha}$. Notice that $Q$ and $\Psi_{ijkl}$ are projectors and $[\mathbb{P}_\sigma^4,Q]=0$, $\forall\sigma\in S_4$, we then have
\begin{equation}
\begin{aligned}
\Tr\left(\mathbb{P}_\sigma^4Q\Psi_{ijkl}\right)&=\Tr\left(\mathbb{P}_\sigma^4QQ\Psi_{ijkl}\Psi_{ijkl}\right)\\
&=\Tr\left(\Psi_{ijkl}Q\mathbb{P}_\sigma^4Q\Psi_{ijkl}\right).
\end{aligned}
\end{equation}
Using the Cauchy-Schwarz inequality of the Hilbert-Schmidt norm, we have
\begin{equation}
\begin{aligned}
Q_\sigma&=\sum\limits_{i,j,k,l}p_ip_jp_kp_l\Tr\left(\Psi_{ijkl}Q\mathbb{P}_\sigma^4Q\Psi_{ijkl}\right)
\\&\leq\sum\limits_{i,j,k,l}p_ip_jp_kp_l\abs{\Tr\left(\Psi_{ijkl}Q\mathbb{P}_\sigma^4Q\Psi_{ijkl}\right)}
\\&\leq\sum\limits_{i,j,k,l}p_ip_jp_kp_l\sqrt{\Tr\left(Q\Psi_{ijkl}\Psi_{ijkl}Q\right)\Tr\left(Q\Psi_{ijkl}\mathbb{P}_\sigma^{4\dagger}\mathbb{P}_\sigma^4\Psi_{ijkl}Q\right)}
\\&=\sum\limits_{i,j,k,l}p_ip_jp_kp_l\Tr\left(Q\Psi_{ijkl}\right)
\\&=\Tr(Q\rho^{\otimes4})=\frac1{d^2}\sum\limits_{P\in\mathcal{P}_n}\Tr(P\rho)^4\leq\frac1{d^2}\sum\limits_{P\in\mathcal{P}_n}\Tr(P\rho)^2=\frac1dR_2\leq\frac1d,
\end{aligned}
\end{equation}
where we use $\Tr(P\rho)\in[-1,1]$ for $P\in\mathcal{P}_n$.
When $\rho=\ket{\psi}\bra{\psi}$ and $\ket{\psi}$ is a stabilizer states, there are exactly $2^n=d$ Pauli operators in $\mathcal{P}_n$ such that $\Tr(P\rho)=1$ and for others $\Tr(P\rho)=0$, we have $Q_\sigma=\frac1d$.
\end{proof}

In the following, we use $R_k^\pm$ and $Q_\lambda^\pm$ to denote the corresponding quantities defined for $\tilde{\rho}_\pm$.

\subsubsection{Formula of 4-order Clifford twirling}
The multi-qubit Clifford forms a unitary 3-design but fails gracefully to become a unitary 4-design~\cite{zhu2016clifford,zhu2017multiqubit,Bittel2025Complete}. Based on the proof of Lemma~\ref{lem:variance_w}, we only need to reanalyze the 4-order term (cf. $A_4$ in Eq.~\eqref{eq:A_4}), by replacing the 4-order unitary twirling formula in Eq.~\eqref{eq:4-order_twirling_X_tensor_2} with the corresponding Clifford twirling formula. Denote the $n$-qubit Clifford group as $\text{Cl}(d)$ with $d=2^n$, its 4-order twirling acting on an operator $X\in\mathcal{L}(\mathbb{C}^d)$ is stated as Theorem S.5 in Ref.~\cite{Roth2018Recovering}:

\begin{equation}
\begin{aligned}
E_{\Delta_{\text{Cl}(d)}^4}(X)&\coloneqq\int_{\mathrm{\text{Cl}}(d)}U^{\otimes 4}XU^{\dagger\otimes 4}dU\\&=\frac1{24}\sum\limits_{\lambda\vdash4,l(\lambda)\leq d}d_\lambda\sum\limits_{\sigma\in S_4}\left[\frac1{D_\lambda^+}\Tr\left(XQ\mathbb{P}_{\sigma^{-1}}^4\right)Q+\frac1{D_\lambda^-}\Tr\left(XQ^\perp\mathbb{P}_{\sigma^{-1}}^4\right)Q^\perp\right]\mathbb{P}_\sigma^4P_\lambda,
\end{aligned}
\end{equation}
where $\lambda\vdash4$ is a partition of 4 and $l(\lambda)$ is its length. The dimensions of the Specht modules $d_\lambda$ and Weyl modules $D_\lambda^\pm$ are provided in Table 1 of Ref.~\cite{zhu2017multiqubit}. The orthogonal projectors to the submodules are (cf. Eq.~(S.8) in Ref.~\cite{Roth2018Recovering})
\begin{equation}
P_\lambda=\frac{d_\lambda}{24}\sum\limits_{\sigma\in S_4}\chi^\lambda(\sigma)\mathbb{P}_\sigma^4,
\end{equation}
where the character values $\chi^\lambda(\sigma)$ can be calculated using the Murnaghan-Nakayama rule and are listed in Table~\ref{tab:character_S_4}.

\begin{table}[h!]
\centering
\begin{tabular}{c|ccccc}
\hline
$\chi^\lambda(\sigma)$ & $(1,1,1,1)$ & $(2,1,1)$ & $(2,2)$ & $(3,1)$ & $(4)$ \\ \hline
$(4)$       & 1 & 1  & 1  & 1  & 1  \\
$(3,1)$     & 3 & 1  & -1 & 0  & -1 \\
$(2,2)$     & 2 & 0  & 2  & -1 & 0  \\
$(2,1,1)$   & 3 & -1 & -1 & 0  & 1  \\
$(1,1,1,1)$ & 1 & -1 & 1  & 1  & -1 \\ \hline
\end{tabular}
\caption{Character table of the symmetric group $S_4$. 
Rows correspond to irreducible representations (partitions $\lambda \vdash 4$), 
columns to conjugacy classes (cycle types of $\sigma\in S_4$).}
\label{tab:character_S_4}
\end{table}

We will use the 4-order Clifford twirling of the operator (cf. Eq.~\eqref{eq:RRM_coef_purity} and Eq.~\eqref{eq:X_+_tensor_2}, ignoring script $+$)
\begin{equation}\label{eq:X_tensor_2}
X^{(1,2)\otimes2}=\sum\limits_{\tilde{s}_1,\tilde{s}_2,\tilde{s}_3,\tilde{s}_4}X_2(\tilde{s}_1,\tilde{s}_2)X_2(\tilde{s}_3,\tilde{s}_4)|\tilde{s}_1\tilde{s}_2\tilde{s}_3\tilde{s}_4\rangle\langle\tilde{s}_1\tilde{s}_2\tilde{s}_3\tilde{s}_4|.
\end{equation}

We consider $d\geq4$ to include all partitions of 4, consistent with implied the assumption $d\geq8$ in Lemma~\ref{lem:estimating_dichotomic_Clifford}. Here, we list all the terms in
\begin{equation}
\Tr\left(E_{\Delta_{\text{Cl}(d)}^4}\left(X^{(1,2)\otimes2}\right)\rho^{\otimes4}\right)=\frac1{24}\sum\limits_{\lambda\vdash4,l(\lambda)\leq d}\sum\limits_{\sigma\in S_4}T(\lambda,\sigma),
\end{equation}
where we define
\begin{equation}
\begin{aligned}
T(\lambda,\sigma)\coloneqq&\Tr\left(d_\lambda\times\left[\frac1{D_\lambda^+}\Tr\left(X^{(1,2)\otimes2}Q\mathbb{P}^4_{\sigma^{-1}}\right)Q+\frac1{D_\lambda^-}\Tr\left(X^{(1,2)\otimes2}Q^\perp\mathbb{P}_{\sigma^{-1}}^4\right)Q^\perp\right]\mathbb{P}_\sigma^4P_\lambda\rho^{\otimes4}\right)\\
=&d_\lambda\left(\frac1{D_\lambda^+}\Tr\left(X^{(1,2)\otimes2}Q\mathbb{P}_{\sigma^{-1}}^4\right)
-\frac1{D_\lambda^-}\Tr\left(X^{(1,2)\otimes2}\mathbb{P}_{\sigma^{-1}}^4\right)
+\frac1{D_\lambda^-}\Tr\left(X^{(1,2)\otimes2}Q\mathbb{P}_{\sigma^{-1}}^4\right)\right)\Tr\left(Q\mathbb{P}_\sigma^4P_\lambda\rho^{\otimes4}\right)
\\&+d_\lambda\left(\frac1{D_\lambda^-}\Tr\left(X^{(1,2)\otimes2}\mathbb{P}_{\sigma^{-1}}^4\right)
-\frac1{D_\lambda^-}\Tr\left(X^{(1,2)\otimes2}Q\mathbb{P}_{\sigma^{-1}}^4\right)\right)\Tr\left(\mathbb{P}_\sigma^4P_\lambda\rho^{\otimes4}\right).
\end{aligned}
\end{equation}

\begin{itemize}
\item[(1)] $\lambda=[1,1,1,1]$:
$d_\lambda=1$, $D_\lambda^+=\frac{(d-1)(d-2)}6$, $D_\lambda^-=\frac{(d-4)(d-2)(d-1)(d+1)}{24}$.
\begin{itemize}
\item[\ding{172}] $\sigma\sim[1,1,1,1]$
\begin{itemize}
\item $\sigma=\mathbb{P}_{(1)}^4$
\begin{equation}
\begin{aligned}
T(\lambda,\sigma)=&\frac{d^2(d^2-3d-1)}{4(d-4)(d-2)(d-1)}\left(1\times Q_{1,1,1,1}+(-6)\times Q_{2,1,1}+3\times Q_{2,2}+8\times Q_{3,1}+(-6)\times Q_4\right)\\
&-\frac{d}{(d-4)(d-2)}\left(1\times R_1+(-6)\times R_2+3\times R_2^2+8\times R_3+(-6)\times R_4\right),
\end{aligned}
\end{equation}
\end{itemize}

\item[\ding{173}] $\sigma\sim[2,1,1]$
\begin{itemize}
\item $\sigma=\mathbb{P}_{(12)}^4, \mathbb{P}_{(34)}^4$
\begin{equation}
\begin{aligned}
T(\lambda,\sigma)=&\frac{d^3}{4(d-2)(d-1)}\left((-1)\times Q_{1,1,1,1}+6\times Q_{2,1,1}+(-3)\times Q_{2,2}+(-8)\times Q_{3,1}+6\times Q_4\right),
\end{aligned}
\end{equation}
\end{itemize}

\begin{itemize}
\item $\sigma=\mathbb{P}_{(13)}^4, \mathbb{P}_{(14)}^4, \mathbb{P}_{(23)}^4, \mathbb{P}_{(24)}^4$
\begin{equation}
\begin{aligned}
T(\lambda,\sigma)=&\frac{d(d^2-3d-1)}{4(d-4)(d-2)(d-1)}\left((-1)\times Q_{1,1,1,1}+6\times Q_{2,1,1}+(-3)\times Q_{2,2}+(-8)\times Q_{3,1}+6\times Q_4\right)\\
&-\frac{1}{(d-4)(d-2)}\left((-1)\times R_1+6\times R_2+(-3)\times R_2^2+(-8)\times R_3+6\times R_4\right),
\end{aligned}
\end{equation}
\end{itemize}

\item[\ding{174}] $\sigma\sim[2,2]$
\begin{itemize}
\item $\sigma=\mathbb{P}_{(12)(34)}^4$
\begin{equation}
\begin{aligned}
T(\lambda,\sigma)=&\frac{d^2(d^2-7d+3)}{4(d-4)(d-2)(d-1)}\left(1\times Q_{1,1,1,1}+(-6)\times Q_{2,1,1}+3\times Q_{2,2}+8\times Q_{3,1}+(-6)\times Q_4\right)\\
&+\frac{d(d-1)}{(d-4)(d-2)}\left(1\times R_1+(-6)\times R_2+3\times R_2^2+8\times R_3+(-6)\times R_4\right),
\end{aligned}
\end{equation}
\end{itemize}

\begin{itemize}
\item $\sigma=\mathbb{P}_{(13)(24)}^4, \mathbb{P}_{(14)(23)}^4$
\begin{equation}
\begin{aligned}
T(\lambda,\sigma)=&\frac{d^3+d^2-d}{4(d-2)(d-1)}\left(1\times Q_{1,1,1,1}+(-6)\times Q_{2,1,1}+3\times Q_{2,2}+8\times Q_{3,1}+(-6)\times Q_4\right),
\end{aligned}
\end{equation}
\end{itemize}

\item[\ding{175}] $\sigma\sim[3,1]$
\begin{itemize}
\item $\sigma=\mathbb{P}_{(123)}^4, \mathbb{P}_{(132)}^4, \mathbb{P}_{(124)}^4, \mathbb{P}_{(142)}^4, \mathbb{P}_{(134)}^4, \mathbb{P}_{(143)}^4, \mathbb{P}_{(234)}^4, \mathbb{P}_{(243)}^4$
\begin{equation}
\begin{aligned}
T(\lambda,\sigma)=&\frac{d^2}{4(d-2)(d-1)}\left(1\times Q_{1,1,1,1}+(-6)\times Q_{2,1,1}+3\times Q_{2,2}+8\times Q_{3,1}+(-6)\times Q_4\right),
\end{aligned}
\end{equation}
\end{itemize}

\item[\ding{176}] $\sigma\sim[4]$
\begin{itemize}
\item $\sigma=\mathbb{P}_{(1234)}^4, \mathbb{P}_{(1243)}^4, \mathbb{P}_{(1342)}^4, \mathbb{P}_{(1432)}^4$
\begin{equation}
\begin{aligned}
T(\lambda,\sigma)=&\frac{d(d^2-7d+3)}{4(d-4)(d-2)(d-1)}\left((-1)\times Q_{1,1,1,1}+6\times Q_{2,1,1}+(-3)\times Q_{2,2}+(-8)\times Q_{3,1}+6\times Q_4\right)\\
&+\frac{d-1}{(d-4)(d-2)}\left((-1)\times R_1+6\times R_2+(-3)\times R_2^2+(-8)\times R_3+6\times R_4\right),
\end{aligned}
\end{equation}
\end{itemize}

\begin{itemize}
\item $\sigma=\mathbb{P}_{(1324)}^4, \mathbb{P}_{(1423)}^4$
\begin{equation}
\begin{aligned}
T(\lambda,\sigma)=&\frac{d^3}{4(d-2)(d-1)}\left((-1)\times Q_{1,1,1,1}+6\times Q_{2,1,1}+(-3)\times Q_{2,2}+(-8)\times Q_{3,1}+6\times Q_4\right),
\end{aligned}
\end{equation}
\end{itemize}
\end{itemize}

\item[(2)] $\lambda=[2,1,1]$:
$d_\lambda=3$, $D_\lambda^+=0$, $D_\lambda^-=\frac{d(d-2)(d-1)(d+1)}{8}$.
\begin{itemize}
\item[\ding{172}] $\sigma\sim[1,1,1,1]$
\begin{itemize}
\item $\sigma=\mathbb{P}_{(1)}^4$
\begin{equation}
\begin{aligned}
T(\lambda,\sigma)=&\frac{3}{d-2}\left(3\times Q_{1,1,1,1}+(-6)\times Q_{2,1,1}+(-3)\times Q_{2,2}+0\times Q_{3,1}+6\times Q_4\right)\\
&-\frac{3}{d-2}\left(3\times R_1+(-6)\times R_2+(-3)\times R_2^2+0\times R_3+6\times R_4\right),
\end{aligned}
\end{equation}
\end{itemize}

\item[\ding{173}] $\sigma\sim[2,1,1]$
\begin{itemize}
\item $\sigma=\mathbb{P}_{(12)}^4, \mathbb{P}_{(34)}^4$
\begin{equation}
\begin{aligned}
T(\lambda,\sigma)=0,
\end{aligned}
\end{equation}
\end{itemize}

\begin{itemize}
\item $\sigma=\mathbb{P}_{(13)}^4, \mathbb{P}_{(14)}^4, \mathbb{P}_{(23)}^4, \mathbb{P}_{(24)}^4$
\begin{equation}
\begin{aligned}
T(\lambda,\sigma)=&\frac{3}{d(d-2)}\left((-1)\times Q_{1,1,1,1}+2\times Q_{2,1,1}+1\times Q_{2,2}+0\times Q_{3,1}+(-2)\times Q_4\right)\\
&-\frac{3}{d(d-2)}\left((-1)\times R_1+2\times R_2+1\times R_2^2+0\times R_3+(-2)\times R_4\right),
\end{aligned}
\end{equation}
\end{itemize}

\item[\ding{174}] $\sigma\sim[2,2]$
\begin{itemize}
\item $\sigma=\mathbb{P}_{(12)(34)}^4$
\begin{equation}
\begin{aligned}
T(\lambda,\sigma)=&\frac{3(1-d)}{d-2}\left((-1)\times Q_{1,1,1,1}+2\times Q_{2,1,1}+1\times Q_{2,2}+0\times Q_{3,1}+(-2)\times Q_4\right)\\
&+\frac{3(d-1)}{d-2}\left((-1)\times R_1+2\times R_2+1\times R_2^2+0\times R_3+(-2)\times R_4\right),
\end{aligned}
\end{equation}
\end{itemize}

\begin{itemize}
\item $\sigma=\mathbb{P}_{(13)(24)}^4, \mathbb{P}_{(14)(23)}^4$
\begin{equation}
\begin{aligned}
T(\lambda,\sigma)=0,
\end{aligned}
\end{equation}
\end{itemize}

\item[\ding{175}] $\sigma\sim[3,1]$
\begin{itemize}
\item $\sigma=\mathbb{P}_{(123)}^4, \mathbb{P}_{(132)}^4, \mathbb{P}_{(124)}^4, \mathbb{P}_{(142)}^4, \mathbb{P}_{(134)}^4, \mathbb{P}_{(143)}^4, \mathbb{P}_{(234)}^4, \mathbb{P}_{(243)}^4$
\begin{equation}
\begin{aligned}
T(\lambda,\sigma)=0,
\end{aligned}
\end{equation}
\end{itemize}

\item[\ding{176}] $\sigma\sim[4]$
\begin{itemize}
\item $\sigma=\mathbb{P}_{(1234)}^4, \mathbb{P}_{(1243)}^4, \mathbb{P}_{(1342)}^4, \mathbb{P}_{(1432)}^4$
\begin{equation}
\begin{aligned}
T(\lambda,\sigma)=&\frac{3(1-d)}{d(d-2)}\left(1\times Q_{1,1,1,1}+(-2)\times Q_{2,1,1}+(-1)\times Q_{2,2}+0\times Q_{3,1}+2\times Q_4\right)\\
&+\frac{3(d-1)}{d(d-2)}\left(1\times R_1+(-2)\times R_2+(-1)\times R_2^2+0\times R_3+2\times R_4\right),
\end{aligned}
\end{equation}
\end{itemize}

\begin{itemize}
\item $\sigma=\mathbb{P}_{(1324)}^4, \mathbb{P}_{(1423)}^4$
\begin{equation}
\begin{aligned}
T(\lambda,\sigma)=0,
\end{aligned}
\end{equation}
\end{itemize}
\end{itemize}

\item[(3)] $\lambda=[2,2]$:
$d_\lambda=2$, $D_\lambda^+=\frac{(d-1)(d+1)}3$, $D_\lambda^-=\frac{(d-2)(d-1)(d+1)(d+2)}{12}$.
\begin{itemize}
\item[\ding{172}] $\sigma\sim[1,1,1,1]$
\begin{itemize}
\item $\sigma=\mathbb{P}_{(1)}^4$
\begin{equation}
\begin{aligned}
T(\lambda,\sigma)=&\frac{d^2}{2(d-2)(d-1)(d+1)(d+2)}\left(2\times Q_{1,1,1,1}+0\times Q_{2,1,1}+6\times Q_{2,2}+(-8)\times Q_{3,1}+0\times Q_4\right)\\
&-\frac{2d}{(d-2)(d+2)}\left(2\times R_1+0\times R_2+6\times R_2^2+(-8)\times R_3+0\times R_4\right),
\end{aligned}
\end{equation}
\end{itemize}

\item[\ding{173}] $\sigma\sim[2,1,1]$
\begin{itemize}
\item $\sigma=\mathbb{P}_{(12)}^4, \mathbb{P}_{(34)}^4$
\begin{equation}
\begin{aligned}
T(\lambda,\sigma)=0,
\end{aligned}
\end{equation}
\end{itemize}

\begin{itemize}
\item $\sigma=\mathbb{P}_{(13)}^4, \mathbb{P}_{(14)}^4, \mathbb{P}_{(23)}^4, \mathbb{P}_{(24)}^4$
\begin{equation}
\begin{aligned}
T(\lambda,\sigma)=0,
\end{aligned}
\end{equation}
\end{itemize}

\item[\ding{174}] $\sigma\sim[2,2]$
\begin{itemize}
\item $\sigma=\mathbb{P}_{(12)(34)}^4$
\begin{equation}
\begin{aligned}
T(\lambda,\sigma)=&\frac{d^3(d^2-3d-1)}{2(d-2)(d-1)(d+1)(d+2)}\left(2\times Q_{1,1,1,1}+0\times Q_{2,1,1}+6\times Q_{2,2}+(-8)\times Q_{3,1}+0\times Q_4\right)\\
&+\frac{2d(d-1)}{(d-2)(d+2)}\left(2\times R_1+0\times R_2+6\times R_2^2+(-8)\times R_3+0\times R_4\right),
\end{aligned}
\end{equation}
\end{itemize}

\begin{itemize}
\item $\sigma=\mathbb{P}_{(13)(24)}^4, \mathbb{P}_{(14)(23)}^4$
\begin{equation}
\begin{aligned}
T(\lambda,\sigma)=&\frac{d^3+d^2-d}{2(d-1)(d+1)}\left(2\times Q_{1,1,1,1}+0\times Q_{2,1,1}+6\times Q_{2,2}+(-8)\times Q_{3,1}+0\times Q_4\right),
\end{aligned}
\end{equation}
\end{itemize}

\item[\ding{175}] $\sigma\sim[3,1]$
\begin{itemize}
\item $\sigma=\mathbb{P}_{(123)}^4, \mathbb{P}_{(132)}^4, \mathbb{P}_{(124)}^4, \mathbb{P}_{(142)}^4, \mathbb{P}_{(134)}^4, \mathbb{P}_{(143)}^4, \mathbb{P}_{(234)}^4, \mathbb{P}_{(243)}^4$
\begin{equation}
\begin{aligned}
T(\lambda,\sigma)=&\frac{d^2}{2(d-1)(d+1)}\left((-1)\times Q_{1,1,1,1}+0\times Q_{2,1,1}+(-3)\times Q_{2,2}+4\times Q_{3,1}+0\times Q_4\right),
\end{aligned}
\end{equation}
\end{itemize}

\item[\ding{176}] $\sigma\sim[4]$
\begin{itemize}
\item $\sigma=\mathbb{P}_{(1234)}^4, \mathbb{P}_{(1243)}^4, \mathbb{P}_{(1342)}^4, \mathbb{P}_{(1432)}^4$
\begin{equation}
\begin{aligned}
T(\lambda,\sigma)=0,
\end{aligned}
\end{equation}
\end{itemize}

\begin{itemize}
\item $\sigma=\mathbb{P}_{(1324)}^4, \mathbb{P}_{(1423)}^4$
\begin{equation}
\begin{aligned}
T(\lambda,\sigma)=0,
\end{aligned}
\end{equation}
\end{itemize}
\end{itemize}

\item[(4)] $\lambda=[3,1]$:
$d_\lambda=3$, $D_\lambda^+=0$, $D_\lambda^-=\frac{d(d-1)(d+1)(d+2)}{8}$.
\begin{itemize}
\item[\ding{172}] $\sigma\sim[1,1,1,1]$
\begin{itemize}
\item $\sigma=\mathbb{P}_{(1)}^4$
\begin{equation}
\begin{aligned}
T(\lambda,\sigma)=&\frac{3}{d+2}\left(3\times Q_{1,1,1,1}+6\times Q_{2,1,1}+(-3)\times Q_{2,2}+0\times Q_{3,1}+(-6)\times Q_4\right)\\
&-\frac{3}{d+2}\left(3\times R_1+6\times R_2+(-3)\times R_2^2+0\times R_3+(-6)\times R_4\right),
\end{aligned}
\end{equation}
\end{itemize}

\item[\ding{173}] $\sigma\sim[2,1,1]$
\begin{itemize}
\item $\sigma=\mathbb{P}_{(12)}^4, \mathbb{P}_{(34)}^4$
\begin{equation}
\begin{aligned}
T(\lambda,\sigma)=0,
\end{aligned}
\end{equation}
\end{itemize}

\begin{itemize}
\item $\sigma=\mathbb{P}_{(13)}^4, \mathbb{P}_{(14)}^4, \mathbb{P}_{(23)}^4, \mathbb{P}_{(24)}^4$
\begin{equation}
\begin{aligned}
T(\lambda,\sigma)=&\frac{3}{d(d+2)}\left(1\times Q_{1,1,1,1}+2\times Q_{2,1,1}+(-1)\times Q_{2,2}+0\times Q_{3,1}+(-2)\times Q_4\right)\\
&-\frac{3}{d(d+2)}\left(1\times R_1+2\times R_2+(-1)\times R_2^2+0\times R_3+(-2)\times R_4\right),
\end{aligned}
\end{equation}
\end{itemize}

\item[\ding{174}] $\sigma\sim[2,2]$
\begin{itemize}
\item $\sigma=\mathbb{P}_{(12)(34)}^4$
\begin{equation}
\begin{aligned}
T(\lambda,\sigma)=&\frac{3(1-d)}{d+2}\left((-1)\times Q_{1,1,1,1}+(-2)\times Q_{2,1,1}+1\times Q_{2,2}+0\times Q_{3,1}+2\times Q_4\right)\\
&+\frac{3(d-1)}{d+2}\left((-1)\times R_1+(-2)\times R_2+1\times R_2^2+0\times R_3+2\times R_4\right),
\end{aligned}
\end{equation}
\end{itemize}

\begin{itemize}
\item $\sigma=\mathbb{P}_{(13)(24)}^4, \mathbb{P}_{(14)(23)}^4$
\begin{equation}
\begin{aligned}
T(\lambda,\sigma)=0,
\end{aligned}
\end{equation}
\end{itemize}

\item[\ding{175}] $\sigma\sim[3,1]$
\begin{itemize}
\item $\sigma=\mathbb{P}_{(123)}^4, \mathbb{P}_{(132)}^4, \mathbb{P}_{(124)}^4, \mathbb{P}_{(142)}^4, \mathbb{P}_{(134)}^4, \mathbb{P}_{(143)}^4, \mathbb{P}_{(234)}^4, \mathbb{P}_{(243)}^4$
\begin{equation}
\begin{aligned}
T(\lambda,\sigma)=0,
\end{aligned}
\end{equation}
\end{itemize}

\item[\ding{176}] $\sigma\sim[4]$
\begin{itemize}
\item $\sigma=\mathbb{P}_{(1234)}^4, \mathbb{P}_{(1243)}^4, \mathbb{P}_{(1342)}^4, \mathbb{P}_{(1432)}^4$
\begin{equation}
\begin{aligned}
T(\lambda,\sigma)=&\frac{3(1-d)}{d(d+2)}\left((-1)\times Q_{1,1,1,1}+(-2)\times Q_{2,1,1}+1\times Q_{2,2}+0\times Q_{3,1}+2\times Q_4\right)\\
&+\frac{3(d-1)}{d(d+2)}\left((-1)\times R_1+(-2)\times R_2+1\times R_2^2+0\times R_3+2\times R_4\right),
\end{aligned}
\end{equation}
\end{itemize}

\begin{itemize}
\item $\sigma=\mathbb{P}_{(1324)}^4, \mathbb{P}_{(1423)}^4$
\begin{equation}
\begin{aligned}
T(\lambda,\sigma)=0,
\end{aligned}
\end{equation}
\end{itemize}
\end{itemize}

\item[(5)] $\lambda=[4]$:
$d_\lambda=1$, $D_\lambda^+=\frac{(d+1)(d+2)}6$, $D_\lambda^-=\frac{(d-1)(d+1)(d+2)(d+4)}{24}$.
\begin{itemize}
\item[\ding{172}] $\sigma\sim[1,1,1,1]$
\begin{itemize}
\item $\sigma=\mathbb{P}_{(1)}^4$
\begin{equation}
\begin{aligned}
T(\lambda,\sigma)=&\frac{d^2(d^2+5d+7)}{4(d+1)(d+2)(d+4)}\left(1\times Q_{1,1,1,1}+6\times Q_{2,1,1}+3\times Q_{2,2}+8\times Q_{3,1}+6\times Q_4\right)\\
&-\frac{d}{(d+2)(d+4)}\left(1\times R_1+6\times R_2+3\times R_2^2+8\times R_3+6\times R_4\right),
\end{aligned}
\end{equation}
\end{itemize}

\item[\ding{173}] $\sigma\sim[2,1,1]$
\begin{itemize}
\item $\sigma=\mathbb{P}_{(12)}^4, \mathbb{P}_{(34)}^4$
\begin{equation}
\begin{aligned}
T(\lambda,\sigma)=&\frac{d^3}{4(d+1)(d+2)}\left(1\times Q_{1,1,1,1}+6\times Q_{2,1,1}+3\times Q_{2,2}+8\times Q_{3,1}+6\times Q_4\right),
\end{aligned}
\end{equation}
\end{itemize}

\begin{itemize}
\item $\sigma=\mathbb{P}_{(13)}^4, \mathbb{P}_{(14)}^4, \mathbb{P}_{(23)}^4, \mathbb{P}_{(24)}^4$
\begin{equation}
\begin{aligned}
T(\lambda,\sigma)=&\frac{d(d^2+5d+7)}{4(d+1)(d+2)(d+4)}\left(1\times Q_{1,1,1,1}+6\times Q_{2,1,1}+3\times Q_{2,2}+8\times Q_{3,1}+6\times Q_4\right)\\
&-\frac{1}{(d+2)(d+4)}\left(1\times R_1+6\times R_2+3\times R_2^2+8\times R_3+6\times R_4\right),
\end{aligned}
\end{equation}
\end{itemize}

\item[\ding{174}] $\sigma\sim[2,2]$
\begin{itemize}
\item $\sigma=\mathbb{P}_{(12)(34)}^4$
\begin{equation}
\begin{aligned}
T(\lambda,\sigma)=&\frac{d^2(d^2+d+3)}{4(d+1)(d+2)(d+4)}\left(1\times Q_{1,1,1,1}+6\times Q_{2,1,1}+3\times Q_{2,2}+8\times Q_{3,1}+6\times Q_4\right)\\
&+\frac{d(d-1)}{(d+2)(d+4)}\left(1\times R_1+6\times R_2+3\times R_2^2+8\times R_3+6\times R_4\right),
\end{aligned}
\end{equation}
\end{itemize}

\begin{itemize}
\item $\sigma=\mathbb{P}_{(13)(24)}^4, \mathbb{P}_{(14)(23)}^4$
\begin{equation}
\begin{aligned}
T(\lambda,\sigma)=&\frac{d^3+d^2-d}{4(d+1)(d+2)}\left(1\times Q_{1,1,1,1}+6\times Q_{2,1,1}+3\times Q_{2,2}+8\times Q_{3,1}+6\times Q_4\right),
\end{aligned}
\end{equation}
\end{itemize}

\item[\ding{175}] $\sigma\sim[3,1]$
\begin{itemize}
\item $\sigma=\mathbb{P}_{(123)}^4, \mathbb{P}_{(132)}^4, \mathbb{P}_{(124)}^4, \mathbb{P}_{(142)}^4, \mathbb{P}_{(134)}^4, \mathbb{P}_{(143)}^4, \mathbb{P}_{(234)}^4, \mathbb{P}_{(243)}^4$
\begin{equation}
\begin{aligned}
T(\lambda,\sigma)=&\frac{d^2}{4(d+1)(d+2)}\left(1\times Q_{1,1,1,1}+6\times Q_{2,1,1}+3\times Q_{2,2}+8\times Q_{3,1}+6\times Q_4\right),
\end{aligned}
\end{equation}
\end{itemize}

\item[\ding{176}] $\sigma\sim[4]$
\begin{itemize}
\item $\sigma=\mathbb{P}_{(1234)}^4, \mathbb{P}_{(1243)}^4, \mathbb{P}_{(1342)}^4, \mathbb{P}_{(1432)}^4$
\begin{equation}
\begin{aligned}
T(\lambda,\sigma)=&\frac{d(d^2+d+3)}{4(d+1)(d+2)(d+4)}\left(1\times Q_{1,1,1,1}+6\times Q_{2,1,1}+3\times Q_{2,2}+8\times Q_{3,1}+6\times Q_4\right)\\
&+\frac{d-1}{(d+2)(d+4)}\left(1\times R_1+6\times R_2+3\times R_2^2+8\times R_3+6\times R_4\right),
\end{aligned}
\end{equation}
\end{itemize}

\begin{itemize}
\item $\sigma=\mathbb{P}_{(1324)}^4, \mathbb{P}_{(1423)}^4$
\begin{equation}
\begin{aligned}
T(\lambda,\sigma)=&\frac{d^3}{4(d+1)(d+2)}\left(1\times Q_{1,1,1,1}+6\times Q_{2,1,1}+3\times Q_{2,2}+8\times Q_{3,1}+6\times Q_4\right),
\end{aligned}
\end{equation}
\end{itemize}
\end{itemize}

Summarizing all the terms, we obtain
\begin{equation}\label{eq:4-order_Clifford_X2}
\begin{aligned}
\Tr\left(E_{\Delta_{\text{Cl}(d)}^4}\left(X^{(1,2)\otimes2}\right)\rho^{\otimes4}\right)=&
\frac{5d^5-d^4-5d^3-43d^2+48}{24(d-2)(d-1)(d+1)(d+2)}Q_{1,1,1,1}
+\frac{d^2+d+2}{2(d+2)}Q_{2,1,1}
\\&+\frac{5d^5-9d^4-21d^3+29d^2+16d-16}{8(d-2)(d-1)(d+1)(d+2)}Q_{2,2}
\\&+\frac{d^2(d^3+d^2-d-5)}{6(d-2)(d-1)(d+1)(d+2)}Q_{3,1}
\\&+\frac{d-1}{2}Q_4
-\frac1{d+2}R_1
-\frac2{d+2}R_2
+\frac{d+1}{d+2}R_2^2.
\end{aligned}
\end{equation}

\end{itemize}

\subsubsection{Proof of Lemma~\ref{lem:variance_w_Clifford}}

\begin{proof}

The proof is similar to that of Lemma~\ref{lem:variance_w}. The first two terms $A_2^{\text{Cl}}$ and $A_3^{\text{Cl}}$ are the same as Eq.~\eqref{eq:A_2} and Eq.~\eqref{eq:A_3} since the multi-qubit Clifford group is a unitary 3-design~\cite{zhu2017multiqubit}. Inserting $d_+=d_-=d/2$, we have
\begin{equation}\label{eq:A_2_Clifford}
\begin{aligned}
A_{2}^{\text{Cl}}&=\binom{N_M}2\left(d_+{R_1^+}^2+(d_+-1){R_2^+}\right)+\left(+\Longleftrightarrow-\right)
\\&=\binom{N_M}2\left[\frac d2\left({R_1^+}^2+{R_1^-}^2\right)+\left(\frac d2-1\right)\left(R_2^++R_2^-\right)\right].
\end{aligned}
\end{equation}
and
\begin{equation}\label{eq:A_3_Clifford}
\begin{aligned}
A_{3}^{\text{Cl}}&=6\binom{N_M}3\left(\frac{d_+-1}{d_++2}{R_2^+}^2R_1^++\frac{2(d_++1)}{d_++2}R_3^+-\frac1{d_++2}{R_1^+}^3\right)+\left(+\Longleftrightarrow-\right)
\\&=6\binom{N_M}3\left(\frac{d-2}{d+4}\left({R_2^+}^2R_1^++{R_2^-}^2R_1^-\right)+\frac{2(d+2)}{d+4}\left(R_3^++R_3^-\right)-\frac1{d+4}\left({R_1^+}^3+{R_1^-}^3\right)\right)
\end{aligned}
\end{equation}

Instead of $A_4$ in Eq.~\eqref{eq:A_4}, we obtain
\begin{equation}\label{eq:A_4_Clifford}
\begin{aligned}
A_{4}^{\text{Cl}}=&6\binom{N_M}4\Tr\left(E_{\Delta_{\text{Cl}(d_+)}^4}\left(X_+^{(1,2)\otimes2}\right)\tilde{\rho}_+^{\otimes4}\right)-6\binom{N_M}4\Tr\left(\Phi_2\left(X_+^{(1,2)}\right)\tilde{\rho}_+^{\otimes2}\right)\Tr\left(\Phi_2\left(X_-^{(1,2)}\right)\tilde{\rho}_-^{\otimes2}\right)
\\&+(+\Longleftrightarrow-)
\\=&6\binom{N_M}4\Biggl[
\frac{5{d_+}^5-{d_+}^4-5{d_+}^3-43{d_+}^2+48}{24({d_+}-2)({d_+}-1)({d_+}+1)({d_+}+2)}Q_{1,1,1,1}^+
+\frac{{d_+}^2+{d_+}+2}{2({d_+}+2)}Q_{2,1,1}^+
\\&+\frac{5{d_+}^5-9{d_+}^4-21{d_+}^3+29{d_+}^2+16{d_+}-16}{8({d_+}-2)({d_+}-1)({d_+}+1)({d_+}+2)}Q_{2,2}^+
\\&+\frac{{d_+}^2({d_+}^3+{d_+}^2-{d_+}-5)}{6({d_+}-2)({d_+}-1)({d_+}+1)({d_+}+2)}Q_{3,1}^+
\\&+\frac{{d_+}-1}{2}Q_4^+
-\frac1{{d_+}+2}R_1^+
-\frac2{{d_+}+2}R_2^+
+\frac{{d_+}+1}{{d_+}+2}{R_2^+}^2-R_2^+R_2^-\Biggr]
+(+\Longleftrightarrow-)
\\=&6\binom{N_M}4\Biggl[
\frac{5d^5-2d^4-20d^3-344d^2+1536}{48(d-4)(d-2)(d+2)(d+4)}\left(Q_{1,1,1,1}^++Q_{1,1,1,1}^-\right)
+\frac{d^2+2d+8}{4(d+4)}\left(Q_{2,1,1}^++Q_{2,1,1}^-\right)
\\&+\frac{5d^5-18d^4-84d^3+232d^2+256d-512}{16(d-4)(d-2)(d+2)(d+4)}\left(Q_{2,2}^++Q_{2,2}^-\right)
\\&+\frac{d^2\left(d^3+2d^2-4d-40\right)}{12(d-4)(d-2)(d+2)(d+4)}\left(Q_{3,1}^++Q_{3,1}^-\right)
\\&+\frac{d-2}{4}\left(Q_4^++Q_4^-\right)
-\frac2{d+4}\left(R_1^++R_1^-\right)
-\frac4{d+4}\left(R_2^++R_2^-\right)
+\frac{d+2}{d+4}\left({R_2^+}^2+{R_2^-}^2\right)-2R_2^+R_2^-\Biggr].
\end{aligned}
\end{equation}

We calculate the variance of the estimator $\omega_u$ (see  Eq.~\eqref{eq:direct_sum_estimator}, cf. Eq.~\eqref{eq:variance_final}):
\begin{equation}\label{eq:var_omega_u_Clifford}
\begin{aligned}
\Var^{\text{Cl}}(\omega_u)&=\mathbb{E}[\omega_u^2]-\mathbb{E}[\omega_u]^2
\\&={\binom{N_M}2}^{-2}\sum\limits_{\substack{\{i,j\}\in\binom{[N_M]}2\\\{i',j'\}\in\binom{[N_M]}2}}\mathbb{E}\left[X_2(s_i,s_j)X_2(s_{i'},s_{j'})\right]-\mathbb{E}[\omega_u]^2
\\&={\binom{N_M}2}^{-2}\left(A_2^{\text{Cl}}+A_3^{\text{Cl}}+A_4^{\text{Cl}}\right)-\left(R_2^+-R_2^-\right)^2.
\end{aligned}
\end{equation}

Inserting Eqs.~\eqref{eq:A_2_Clifford},~\eqref{eq:A_3_Clifford}, and ~\eqref{eq:A_4_Clifford}, we have
\begin{equation}\label{eq:variance_final_Clifford}
\begin{aligned}
\Var^{\text{Cl}}(\omega_u)=&{\binom{N_M}2}^{-2}\left(A_2^{\text{Cl}}+A_3^{\text{Cl}}+A_4^{\text{Cl}}\right)-\left(R_2^+-R_2^-\right)^2
\\=&-\frac{2(N_M-2)(N_M-3)}{N_M(N_M-1)(d+4)}
\\&+\frac{d}{N_M(N_M-1)}\left({R_1^+}^2+{R_1^-}^2\right)
\\&-\frac{4(N_M-2)}{N_M(N_M-1)(d+4)}\left({R_1^+}^3+{R_1^-}^3\right)
\\&+\frac{4(N_M-2)(d-2)}{N_M(N_M-1)(d+4)}\left(R_1^+{R_2^+}^2+R_1^-{R_2^-}^2\right)
\\&+\left(\frac{d-2}{N_M(N_M-1)}-\frac{4(N_M-2)(N_M-3)}{N_M(N_M-1)(d+4)}\right)\left(R_2^++R_2^-\right)
\\&-\frac{4N_M-6}{N_M(N_M-1)}\left(R_2^+-R_2^-\right)^2
\\&-\frac{2(N_M-2)(N_M-3)}{N_M(N_M-1)(d+4)}\left({R_2^+}^2+{R_2^-}^2\right)
\\&+\frac{8(N_M-2)(d+2)}{N_M(N_M-1)(d+4)}\left(R_3^++R_3^-\right)
\\&+\frac{(N_M-2)(N_M-3)(5d^5-2d^4-20d^3-344d^2+1536)}{48N_M(N_M-1)(d-4)(d-2)(d+2)(d+4)}\left(Q_{1,1,1,1}^++Q_{1,1,1,1}^-\right)
\\&+\frac{(N_M-2)(N_M-3)(d^2+2d+8)}{4N_M(N_M-1)(d+4)}\left(Q_{2,1,1}^++Q_{2,1,1}^-\right)
\\&+\frac{(N_M-2)(N_M-3)(5d^5-18d^4-84d^3+232d^2+256d-512)}{16N_M(N_M-1)(d-4)(d-2)(d+2)(d+4)}\left(Q_{2,2}^++Q_{2,2}^-\right)
\\&+\frac{(N_M-2)(N_M-3)d^2\left(d^3+2d^2-4d-40\right)}{12N_M(N_M-1)(d-4)(d-2)(d+2)(d+4)}\left(Q_{3,1}^++Q_{3,1}^-\right)
\\&+\frac{(N_M-2)(N_M-3)(d-2)}{4N_M(N_M-1)}\left(Q_4^++Q_4^-\right).
\end{aligned}
\end{equation}

We have $R_k^\pm\leq1$ and $Q_\sigma^\pm\leq2/d$ from Lemma~\ref{lem:upper_bound_Q_sigma}, so that we obtain an upper bound for $\Var^{\text{Cl}}(\omega_u)$
\begin{equation}
\begin{aligned}
\Var^{\text{Cl}}(\omega_u)\leq&\frac{4d}{N_M(N_M-1)}+\frac{24(N_M-2)}{N_M(N_M-1)}
+\frac{4(N_M-2)(N_M-3)}{N_M(N_M-1)}
\\=&\cO\left(\frac{d}{N_M^2} + \frac{1}{N_M} + 1\right).
\end{aligned}
\end{equation}

Furthermore, suppose we draw $N_U$ independent random unitaries and construct an estimator by averaging the individual estimators obtained for each choice of random unitary, i.e.,
\begin{equation}
\omega =\frac{1}{N_U}\sum\limits_{u=1}^{N_U}\omega_u.
\end{equation}
Then we can estimate $\Tr(O \rho^2)$ unbiasedly, with the variance given by
\begin{equation}
\Var^{\text{Cl}}(\omega) = \frac{1}{N_U}\Var^{\text{Cl}}(\omega_u) = \frac{1}{N_U} \cO\left(\frac{d}{N_M^2} + \frac{1}{N_M} + 1\right).
\end{equation}

\end{proof}

Based on Lemmas~\ref{lem:median-of-mean},~\ref{lem:expectation_value_w}, and~\ref{lem:variance_w_Clifford}, we can now prove Lemma~\ref{lem:estimating_dichotomic_Clifford} in the main text.

\begin{proof}[Proof of Lemma~\ref{lem:estimating_dichotomic_Clifford}]

To apply Lemma~\ref{lem:median-of-mean}, it is sufficient to suppress $\Var^{\text{Cl}}(\omega)\leq\varepsilon^2/3<\varepsilon^2/2$. If we choose $N_U=\cO\left(\frac1{\varepsilon^2}\right)$ and $N_M=\cO\left(\sqrt{d}\right)$, then $N_U\times N_M=\cO\left(\frac{\sqrt{d}}{\varepsilon^2}\right)$ copies of states are sufficient to estimate $\Tr(O\rho^2)$ for traceless dichotomic $O$ to $\varepsilon$-precision within constant failure probability.
\end{proof}

%-----------------

\subsection{Decomposition of dichotomic observables}\label{app:decomp_dichotomic_obs}
Now we provide the analysis of the sample complexity stated in Theorem~\ref{thm:good_dicho_estimation}. We assume that $\min\{d_-,d_+\} < \frac{d}{4}$; otherwise, Lemma~\ref{lem:estimating_dichotomic} already yields the desired sample complexity since $d_- = d_+ = \Omega(d)$. We then obtain the following decomposition, which is illustrated in  Fig.~\ref{fig:dichotomic}.

\begin{lemma}[Formal version of Lemma \ref{lem:Dichotomic_decomposition}]
    Suppose $O$ is a dichotomic observable satisfying $\min\{d_+, d_-\} < \frac{d}{4}$. Then, $O = \frac{1}{2}(\pm\mathbb{I}+O_1+O_2+O_3)$, where each $O_i$ is a dichotomic observable with eigen-subspace dimensions $d_-^{(i)}$ and $d_+^{(i)}$ satisfying $\min\{d_-^{(i)},d_+^{(i)}\} \ge \frac{d}{3}$.
\end{lemma}

\begin{proof}
Let $P_{+}$ and $P_{-}$ denote the projectors onto the eigenspaces of $O$ corresponding to the eigenvalues $+1$ and $-1$, respectively, so that $O = P_+-P_-$ and $\mathbb{I} = P_+ + P_-$. 
Without loss of generality, we first assume $d_- < d_+$ and $d_- < \frac{d}{4}$.
Next, decompose the $+1$ eigenspace $\Pi_{+}$ into three orthogonal subspaces, denoted by $\Pi_{+}^{(1)}$, $\Pi_{+}^{(2)}$, and $\Pi_{+}^{(3)}$ (projectors correspondingly denoted as $P_+^{(j)}$), each having dimension at least $\lfloor \frac{d}{3} \rfloor$.  
Define 
\begin{equation}
    O_i = \left(\sum_{j=1}^3 (-1)^{\delta_{i,j}} P_{+}^{(j)}\right) - P_-,
\end{equation}
where the exponents $\delta_{i,j}$ are the delta function, $\delta_{i,j} = 1$ if and only if $i = j$.
A direct calculation then shows that $O = \frac{1}{2}(\mathbb{I} + O_1 + O_2 + O_3)$.
Moreover, for each $O_i$, one can verify that the dimensions of its positive and negative eigenspaces satisfy
\begin{equation}
    d_+^{(i)} \ge d_+ - \left\lceil \frac{d_+}{3}\right\rceil >\frac{2}{3}d_+ - 1 > \frac{d}{2} - 1 \ge \frac{d}{3}
\end{equation}
where the last inequality holds because $0 < d_- < \frac{d}{4}$ implies that $d \ge 8$. Also,
\begin{equation}
    d_-^{(i)} = d_- + \left\lfloor \frac{d_+}{3} \right\rfloor \ge 1 + \left\lfloor \frac{d-1}{3} \right\rfloor \ge \frac{d}{3}.
\end{equation}
This completes the proof of the cases $d_+ > d_-$.

For the case where $d_+ < d_-$, we apply the above procedure to the observable $-O = \frac{1}{2} (\mathbb{I}+O_1+O_2+O_3)$. This yields the decomposition $O = \frac{1}{2}(-\mathbb{I} + O_1'+O_2'+O_3')$,
where $O_i' = -O_i$ for $1\le i \le 3$. Since negating an observable swaps its eigenspaces but preserves their dimensions, the components $O_i'$ also satisfy the condition  $\min\{d_+^{(i)}, d_-^{(i)}\}  \ge \frac{d}{3}$. This completes the proof.
\end{proof}

\begin{figure}[!htbp]
\centering
\includegraphics[width=.6\linewidth]{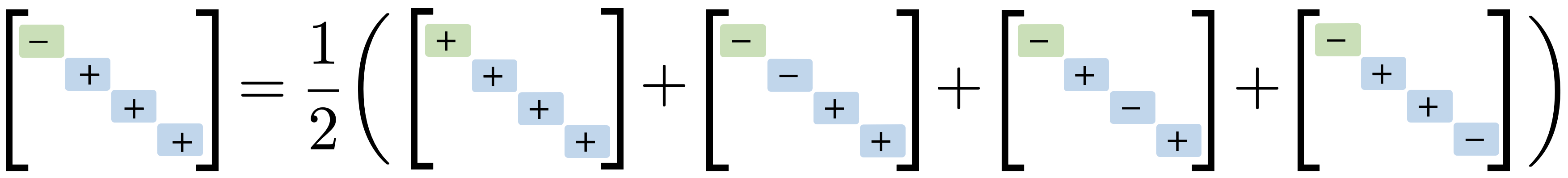}
\caption{Decomposition of a dichotomic observable. On the left, an observable $O$ with $d_- < \frac{d}{4}$ is depicted, where the green and blue blocks represent the eigensubspaces corresponding to $\Pi_-$ and the three partitions of $\Pi_+$, denoted by $\Pi_+^{(1)}$, $\Pi_+^{(2)}$, and $\Pi_+^{(3)}$, respectively. The $\pm$ signs within each block denote the corresponding diagonal entries of $\pm1$. On the right, the decomposition $O=\frac{1}{2}(\mathbb{I}+O_1+O_2+O_3)$
is shown: the first matrix represents the identity $\mathbb{I}$, while the remaining matrices correspond to $O_1$, $O_2$, and $O_3$, each satisfying $d_{\pm}^{(i)}\geq \frac{d}{3}$.}
\label{fig:dichotomic}
\end{figure}

\subsection{Proof of Theorem~\ref{thm:estimating_general}}\label{app:proof_of_general}
\begin{proof}
    Let $k = \lceil \log_2(\varepsilon^{-1}) \rceil + 1$, and use Lemma~\ref{lem:decompose} to decompose $O$ as
    \begin{equation}
        \Tr(O\rho^2) =  \sum\limits_{l=1}^k 2^{-l} \Tr(O_l\rho^2)  \;+\; \Tr(O_{\Delta}\rho^2)
    \end{equation}
    We set $\delta' = \delta/k$ and $\varepsilon'_l =(\tfrac{3}{2})^l \,\tfrac{\varepsilon}{8}$.
    For each \(1 \le l \le k\), we obtain an estimator \(\omega_{l}\) of \(\Tr(O\rho^2_l)\) such that
    \begin{equation}
        \bigl|\omega_{l} - \Tr(O\rho^2_{l})\bigr| \;\le\; \varepsilon'_l 
    \end{equation}
    with probability at least \(1 - \delta'\). By a union bound, with probability at least \(1-\delta\), the above holds for all \(1 \le l \le k\). We then set 
    \begin{equation}
        \omega = \sum\limits_{l=1}^k 2^{-l} \omega_l
    \end{equation}
    as the estimator for $\Tr(O\rho^2)$. With probability at least \(1-\delta\),
    \begin{equation}
    \begin{split}
        \abs{\omega-\Tr(O\rho^2)} &= \abs{\sum\limits_{l=1}^k 2^{-l}\left(\omega_l-\Tr(O_l\rho^2)\right) + \Tr(O_{\Delta}\rho^2)} \\
        &\le \sum\limits_{l=1}^k 2^{-l} \abs{\omega_l-\Tr(O_l\rho^2)} + \abs{\Tr(O_{\Delta}\rho^2)} \\
        & \le \sum\limits_{l=1}^{k} 2^{-l} \varepsilon'_l + \norm{O_{\Delta}}_{\infty} \\
        &\le \sum\limits_{l=1}^k  \left(\frac{3}{4}\right)^l \frac{\varepsilon}{8} + \frac{\varepsilon}{2} < \varepsilon .
    \end{split}
    \end{equation}
    
    Finally, we account for the sample complexity. We obtain $k$ estimators $\{\omega_i\}_{i=1}^l$. For each estimator $\omega_l$, according to Theorem~\ref{thm:good_dicho_estimation}, the sample complexity is $\cO\left(\max\left\{\frac{\sqrt{d}}{\varepsilon'_l}, \frac{1}{(\varepsilon'_l)^2} \right\} \log(1/\delta')\right) = \cO\left([(\frac{2}{3})^{l}\frac{\sqrt{d}}{\varepsilon} + (\frac{2}{3})^{2l} \frac{1}{\varepsilon^2}]\log(1/\delta')\right)$. 
    Summing over \(l = 1, 2, \dots, k\) gives the total sample complexity
    \begin{equation}
    \begin{split}
        &\sum\limits_{l=0}^k\Bigl[ \Bigl(\frac{2}{3}\Bigr)^{l}\Bigl(\frac{\sqrt{d}}{\varepsilon}\Bigr) + \Bigl(\frac{2}{3}\Bigr)^{2l}\Bigl(\frac{1}{\varepsilon^2}\Bigr)\Bigr] \cO\Bigl(\log(1/\delta')\Bigr) \\
        =& \cO\Bigl(\Bigl(\frac{\sqrt{d}}{\varepsilon} + \frac{1}{\varepsilon^2}\Bigr) \log\bigl[\delta^{-1}\log(\varepsilon^{-1})\bigr] \Bigr) \\
        =& \cO\Bigl(\max\Bigl\{\frac{\sqrt{d}}{\varepsilon}, \frac{1}{\varepsilon^{2}}\Bigr\}\log\bigl[\delta^{-1}\log(\varepsilon^{-1})\bigr]\Bigr).
    \end{split}
    \end{equation}
\end{proof}

\section{Lower-bound analysis for sample complexity}\label{app:lower_bound}
In this section, we will give a brief proof of Fact~\ref{fact:large_one_norm}, a direct corollary of Theorem~9 in the Supplementary Material of Ref.~\cite{Huang2022QuantumAdvantage}.
In that work, the authors proved the sample complexity lower bound for estimating $\Tr(Z_1\rho^2)$, where $Z_1$ represents the Pauli-$Z$ observable on the first qubit, by constructing a state discrimination task. We consider specifically two state ensembles:
\begin{enumerate}
\item $\mathcal{E}_1$ containing states of 
\begin{equation}
    \rho=\frac{1}{2}\ketbra{0}{0}\otimes\ketbra{\psi}{\psi}+\frac{1}{2}\ketbra{1}{1}\otimes\frac{\mathbb{I}}{2^{n-1}}.
\end{equation}
\item $\mathcal{E}_2$ containing states of 
\begin{equation}
\rho=\frac{1}{2}\ketbra{1}{1}\otimes\ketbra{\psi}{\psi}+\frac{1}{2}\ketbra{0}{0}\otimes\frac{\mathbb{I}}{2^{n-1}}.
\end{equation}
\end{enumerate}
Here, $\frac{\mathbb{I}}{2^{n-1}}$ and $\ket{\psi}$ represent the $(n-1)$-qubit maximally mixed state and Haar-random state vector.
A Haar-random pure state can be sampled by $U\ket{\psi_0}$, where $U$ is a Haar-random unitary and $\ket{\psi_0}$ is some fixed state vector.
The task is that, given multiple copies of a state $\rho$ sampled from either $\mathcal{E}_1$ or $\mathcal{E}_2$ with equal probabilities, how many copies of $\rho$ should be provided to decide which ensemble it comes from.
Note that, when $n$ is large enough, for states from ensemble $\mathcal{E}_1$ we have $\Tr(Z_1\rho^2)\approx\frac{1}{4}$; while $\Tr(Z_1\rho^2)\approx-\frac{1}{4}$ for states from ensemble $\mathcal{E}_2$.
Therefore, if one can accurately estimate the value of $\Tr(Z_1\rho^2)$, one can distinguish these two state ensembles with a high probability.
Similarly, if the discrimination of these two state ensembles has a fundamental limitation, it also poses difficulties for the estimation of $\Tr(Z_1\rho^2)$.
In Ref.~\cite{Huang2022QuantumAdvantage}, the authors proved that, for protocols utilizing arbitrary adaptive single-copy operations and measurements, $\Omega(\sqrt{d})$ copies of $\rho$ should be provided to decide the state ensemble with success probability $\frac{1}{2}+\Omega(1)$, where $d$ represents the dimension of states in ensembles $\mathcal{E}_1$ and $\mathcal{E}_2$.

It is direct to generalize their conclusion to other observables by defining two new state ensembles and the corresponding state discrimination task:
\begin{enumerate}
\item $\mathcal{E}_1^U$ containing states of 
\begin{equation}
    \rho=U\left(\frac{1}{2}\ketbra{0}{0}\otimes\ketbra{\psi}{\psi}+\frac{1}{2}\ketbra{1}{1}\otimes\frac{\mathbb{I}}{2^{n-1}}\right)U^\dagger.
\end{equation}
\item $\mathcal{E}_2^U$ containing states of 
\begin{equation}
\rho=U\left(\frac{1}{2}\ketbra{1}{1}\otimes\ketbra{\psi}{\psi}+\frac{1}{2}\ketbra{0}{0}\otimes\frac{\mathbb{I}}{2^{n-1}}\right)U^\dagger,
\end{equation}
\end{enumerate}
where $U$ is an arbitrary $d$-dimensional unitary.
Note that the new state ensemble discrimination task is no more difficult than the original discrimination task between $\mathcal{E}_1$ and $\mathcal{E}_2$.
This is because the two original state ensembles can be transformed into $\mathcal{E}_1^U$ and $\mathcal{E}_2^U$ using the unitary transformation $U$, and the protocols we consider can implement arbitrary single-copy operations and measurements.
So we can conclude that distinguishing between $\mathcal{E}_1^U$ and $\mathcal{E}_2^U$ also has a sample complexity lower bound of $\Omega(\sqrt{d})$.
To prove the sample complexity lower bound for estimating $\Tr(P\rho^2)$ with $P$ being a Pauli observable, we can choose $U$ to be the unitary that satisfies $U^\dagger P U=Z_1$.
Then, it can be similarly verified that the accurate estimation of $\Tr(P\rho^2)$ helps to distinguish $\mathcal{E}_1^U$ and $\mathcal{E}_2^U$ with a high probability.
Therefore, we can conclude that the accurate estimation of $\Tr(P\rho^2)$ with arbitrary Pauli observable $P$ and target $d$-dimensional state $\rho$ has a sample complexity lower bound of $\Omega(\sqrt{d})$.

\section{Additional analysis on the Pauli sampling protocol}\label{app:Pauli_sampling}

In this section, we analyze the performance of the Pauli sampling scheme presented in Sec.~\ref{sec:pauli_sampling} and prove Theorem \ref{thm:Pauli_sampling}. We begin by establishing the following lemma.

\begin{lemma}\label{lem:Pauli_sampling}
    The estimator $\omega$ defined in Eq.~\eqref{eq:estimator_Pauli_sampling} satisfies
    \begin{equation}
    \begin{split}
        \Pr[\abs{\omega - \Tr(O \rho^2)} \ge \varepsilon] \le \frac{1}{3},
    \end{split}
    \end{equation}
    and the expected sample complexity is given by
    \begin{equation}
        \mathbb E[\sum_{i=1}^l m_i ] = \cO\left(\frac{K \norm{O}_2^2}{ \sqrt{d} \varepsilon^2}\right).
    \end{equation}
\end{lemma}

\begin{proof}
Set $l = \frac{24 K \norm{O}_2^2}{d\varepsilon^2}$. For a Pauli operator $P_i$, we set $N_U = \cO\left(\frac{\norm{O}_2^4}{d  l \varepsilon^2 \chi_O^2(P_i)}\right)$ and $N_M = \cO(\sqrt{d})$, with sample complexity given by $m_i = N_UN_M = \cO\left(\frac{\norm{O}_2^4}{\sqrt{d} l \varepsilon^2  \chi_O^2(P_i)}\right)$ for this Pauli operator. According to Lemma \ref{lem:variance_w_Clifford}, we have the following:
\begin{equation}
\begin{split}
    \bE[\omega_i | P_i] &= \Tr(P_i\rho^2), \\
    \Var[\omega_i | P_i] &= \frac{d l \varepsilon^2 \chi_O^2(P_i)}{24 \norm{O}_2^4 }
\end{split}
\end{equation}

Now, consider the random variable
\begin{equation}
    \tilde{X}_i = \frac{\norm{O}_2^2}{\Tr(OP_i)} \omega_i = \frac{\norm{O}_2^2 }{\chi_O(P_i) \sqrt{d}} \omega_i,
\end{equation}
for which we have the following properties:
\begin{equation}
\begin{split}
    \bE[\tilde{X}_i | P_i] &= \norm{O}_2^2 \frac{\chi_{\rho^2}(P_i)}{\chi_{O}(P_i)} = X_i, \\ 
    \Var[\tilde{X}_i | P_i] &= \frac{l\varepsilon^2}{24}.
\end{split}
\end{equation}

Therefore, for the random variable $\omega = \frac{1}{l} \sum_{i=1}^l \tilde{X}_i$, the variance is given by:
\begin{equation}
    \Var[\omega | \{P_i\}] = \frac{\varepsilon^2}{24}.
\end{equation}
By Chebyshev's inequality, when fixing $\{P_i\}$, we have
\begin{equation}\label{eq:expectation_condition_P}
    \Pr[\abs{\omega - \bE[\omega | \{P_i\}]} \ge \frac{\varepsilon}{2}]    \le \frac{1}{6}
\end{equation}
where
\begin{equation}
\begin{split}
    \bE[\omega | \{P_i\}] = \frac{1}{l}\sum_{i=1}^l \bE[\tilde{X}_i |P_i] =\frac{1}{l}\sum_{i=1}^l X_i,
\end{split}
\end{equation}
For each $X_i$, we have
\begin{equation}
\begin{split}
    \Var_{P_i}[X_i] &= \bE [X_i^2] - (\bE[X_i])^2 \\
    &\le \bE [X_i^2] \\
    &= \sum_{i=1}^K p_i X_i^2 \\
    &= \sum_{i=1}^K \frac{\chi_O^2(P_i)}{\norm{O}_2^2} \left( \norm{O}_2^2 \frac{\chi_{\rho^2}(P_i)}{\chi_O(P_i)} \right)^2\\
    &= \sum_{i=1}^K \norm{O}_2^2 \frac{\Tr(P_i\rho^2)^2}{d} \\
    &\le \frac{K\norm{O}_2^2}{d}.
\end{split}
\end{equation}
Therefore, 
\begin{equation}
    \Var_{\{P_i\}}(\bE[\omega | \{P_i\}] ) = \frac{1}{l}\Var_{P_i}[X_i] \le \frac{\varepsilon^2}{24}.
\end{equation}
By Chebyshev's inequality, we then have
\begin{equation}\label{eq:expectation_not_condition_P}
     \Pr_{\{P_i\}}[\abs{\bE[\omega | \{P_i\}] - \bE[\omega]} \ge \frac{\varepsilon}{2}]    \le \frac{1}{6}
\end{equation}
Finally, combining Eq.~\eqref{eq:expectation_condition_P} and Eq.~\eqref{eq:expectation_not_condition_P} via a union bound, we obtain
\begin{equation}
    \Pr[\abs{\omega - \bE[\omega]} \ge \varepsilon] \le \frac{1}{3}.
\end{equation}
Substituting $\mathbb{E}[\omega] = \mathbb{E}[X_i] = \sum_i p_i X_i = \Tr(O \rho^2)$ proves the first inequality in Lemma \ref{lem:Pauli_sampling}.

Moreover, the expected sample complexity is given by
\begin{equation}
\begin{split}
    \bE[\sum_{i=1}^l m_i] &= l \sum_{i=1}^K p_i m_i\\
    &= l \sum_{i=1}^K \frac{\chi^2_O(P_i)}{\norm{O}_2^2} \cO\left( \frac{\norm{O}_2^4}{\sqrt{d} l \varepsilon^2  \chi_O^2(P_i)}\right) \\
    &= l \sum_{i=1}^K \cO\left( \frac{\norm{O}_2^2}{\sqrt{d} l \varepsilon^2}\right) \\
    &=\cO\left(  \frac{K \norm{O}_2^2}{\sqrt{d} \varepsilon^2} \right).
\end{split}
\end{equation}
\end{proof}

\begin{proof}[Proof of Theorem~\ref{thm:Pauli_sampling}]
    The claimed performance guarantee is obtained by setting $T = \cO(\log \delta^{-1})$ and applying the median-of-means estimator from Lemma~\ref{lem:median-of-mean}. The expected sample complexity is given by 
    \begin{equation}
        T \cdot \mathbb{E}\left[\sum_{i=1}^l m_i\right] = \cO\left(\frac{K \norm{O}_2^2 \log \delta^{-1}}{\sqrt{d} \varepsilon^2} \right).
    \end{equation}
\end{proof}

\section{Additional analysis on estimating low-rank observables} \label{app:low_rank}
In this section, we analyze the performance of the BRM protocol proposed in Sec.~\ref{sec:low_rank} for estimating non-linear properties $\Tr(O\rho^2)$.
We begin by showing it is unbiased for estimating $\Tr(O\rho^2)$ and then provide a performance analysis of the proposed protocol.
Consider the estimator defined in Eq.~\eqref{eq:RRM_estimator_single_round}
\begin{equation}
    \omega_u 
    = \binom{N_M}{2}^{-1}
      \sum\limits_{1 \le i < j \le N_M}
      \sum\limits_{\sigma \in \{0,1\}^n}
      X_3(s_i, s_j, \sigma)\,
      \bra{\sigma}\,U\,O\,U^\dagger\ket{\sigma},
\end{equation}
where $s_i$ and $s_j$ are the $i$-th and $j$-th measurement outcomes from a set of $N_M$ measurements in the basis defined by $U_u$. The coefficient function $X_3(s_1,s_2,s_3)$ should be chosen such that the corresponding coefficient matrix $X_3 = \sum\limits_{s_1, s_2, s_3} X_3(s_1, s_2, s_3)\, \ketbra{s_1, s_2, s_3}$ satisfies Eq.~\eqref{eq:RRM_third_twirling}, so that
\begin{equation}
    \underset{U}{\mathbb{E}}\bigl[U^{\otimes 3} \,X_3\, U^{\dagger\otimes 3}\bigr]
    = \tfrac{1}{2}\bigl(\mathbb{P}^3_{(132)} + \mathbb{P}^3_{(123)}\bigr),
\end{equation}
i.e., its third-order twirling yields a superposition of the two third-order cyclic operators. 
Let $d = 2^n$ be the dimension of the Hilbert space, one possible choice is shown as Eq.~(18) in the Supplementary Material of Ref.~\cite{zhou2020Single}, which reads as
\begin{equation}
\begin{aligned}
X_3(s_1,s_2,s_3)=X_3({\abs{\{s_1,s_2,s_3\}}}),\text{ with}\\
    X_3(1) = \frac{1 + d^2}{2}, \quad X_3(2) = \frac{1-d}{2}, \quad X_3(3) = 1.
\end{aligned}
\end{equation}
Then we have
\begin{equation}
\begin{aligned}
\mathbb{E}\left[\omega_u\right]&=\underset{U}{\mathbb{E}}\left[\sum\limits_{s_1,s_2,\sigma}\langle s_1|U\rho U^\dagger|s_1\rangle\langle s_2|U\rho U^\dagger|s_2\rangle\langle\sigma|UOU^\dagger|\sigma\rangle X_3(s_1,s_2,\sigma)\right]
\\&=\underset{U}{\mathbb{E}}\left[\Tr\left(U^{\dagger\otimes3}\left(\sum\limits_{s_1,s_2,\sigma}X_3(s_1,s_2,\sigma)|s_1s_2\sigma\rangle\langle s_1s_2\sigma|\right)U^{\otimes3}\left(\rho\otimes\rho\otimes O\right)\right)\right]
\\&=\frac12\left(\Tr\left(\rho O\rho\right)+\Tr\left(\rho^2O\right)\right)
\\&=\Tr\left(O\rho^2\right).
\end{aligned}
\end{equation}

\begin{lemma}[Variance of $\omega_u$ in Eq.~\eqref{eq:RRM_estimator_single_round}]
\label{lem:RRM_variance}
    Given an observable $O$, the variance of the estimator $\omega_u$ defined in Eq.~\eqref{eq:RRM_estimator_single_round} satisfies
    \begin{equation}
        \Var[\omega_u] = \cO\left( \left(1 + \frac{d}{N_M^2}\right) \Tr(O^2) \right).
    \end{equation}
\end{lemma}
\noindent We postpone the proof of Lemma~\ref{lem:RRM_variance} and first analyze the performance of the full protocol. We further draw $N_U$ independent random unitaries and construct the estimator 
\begin{equation}\label{eq:RRM_averaged_estimator}
    \omega =\frac{1}{N_U} \sum\limits_{u=1}^{N_U} \omega_u.
\end{equation}
By standard variance reduction for averaging independent, identically distributed estimators and applying Lemma~\ref{lem:RRM_variance}, 
we obtain the following corollary.

\begin{corollary}[Variance of $\omega$]
\label{col:RRM_variance}
    Given an observable $O$, the variance of the estimator $\omega$ defined in Eq.~\eqref{eq:RRM_averaged_estimator} satisfies
    \begin{equation}
        \Var[\omega] = \cO\left( \left( \frac{1}{N_U} + \frac{d}{N_U N_M^2} \right) \Tr(O^2) \right).
    \end{equation}
\end{corollary}
\noindent This allows us to prove Theorem~\ref{thm:low_rank} in the main text.

\begin{proof}[Proof of Theorem~\ref{thm:low_rank}]
    By Corollary~\ref{col:RRM_variance}, for any constant $0 < c < \tfrac{1}{2}$, we can choose $N_U = \cO(\varepsilon^{-2})$ and $N_M = \cO(\sqrt{d})$ such that 
    \begin{equation}
        \Var[\omega] \le c \varepsilon^2.
    \end{equation}
    Set the failure probability per observable to $\delta' = \delta / M$, and use the median-of-means estimator (Lemma~\ref{lem:median-of-mean}) to estimate $\Tr(O_m\rho^2)$ within additive error $\varepsilon$ with success probability at least $1 - \delta'$. This requires $T = \cO\bigl(\log(1/\delta')\bigr)$ repetitions. Applying the union bound over all $M$ observables, we find that with probability at least $1 - \delta' M = 1 - \delta$, the estimation error satisfies
    \begin{equation}
         \left|\omega - \Tr(O_m\rho^2)\right| \le \varepsilon
    \end{equation}
    for all $1 \le m \le M$. This completes the proof.
\end{proof}

We now prove the validity of Lemma~\ref{lem:RRM_variance}.
\begin{proof}[Proof of Lemma~\ref{lem:RRM_variance}]
We begin by applying the total variance law to arrive at
\begin{equation}
\begin{split}
    \Var [\omega_u] &= \underset{U}{\bE} [\underset{\bs}{\Var}(\omega_u|U)] + \underset{U}{\Var} [\underset{\bs}{\bE}(\omega_u|U)] \\
    &= \underset{U}{\bE}[\underset{\bs}{\bE}(\omega_u^2|U) - \underset{\bs}{\bE}(\omega_u|U)^2] + \underset{U}{\bE} [\underset{\bs}{\bE}(\omega_u|U)]^2 - [\underset{U}{\bE} \underset{\bs}{\bE}(\omega_u|U)]^2 \\
    &=  \underset{U}{\bE}\underset{\bs}{\bE}(\omega_u^2|U) - [\underset{U}{\bE} \underset{\bs}{\bE}(\omega_u|U)]^2 \\
    &\le \underset{U}{\bE}\underset{\bs}{\bE}(\omega_u^2|U).
\end{split}
\end{equation}
where $\bs=(s_1,\dots,s_{N_M}) \in \mathbb{Z}_d^{N_M}$ denotes the sequence of measurement outcomes.
In the analysis below, we use the notation $A \lesssim B$ to denote $A = \cO(B)$. Define $\Delta(i,j;i',j')=\left|\{i,j\}\cup\{i',j'\}\right|\in\{2,3,4\}$ as the total number of distinct indices in index pairs $(i,j)$ and $(i',j')$. We now expand $\omega_u$ and analyze each term to get
\begin{equation}\label{eq:variance}
\begin{split}
    &\underset{U}{\bE} \underset{\bs}{\bE}(\omega_u^2 \mid U) \\
    \lesssim& \frac{1}{N_M^4} \underset{U}{\bE} \underset{\bs}{\bE} \sum\limits_{i < j;\, i' < j'} \sum\limits_{a_3, a_4 \in \bZ_d} X_3(s_i, s_j, a_3) X_3(s_i', s_j', a_4)  \times \bra{a_3} U O U^{\dagger} \ket{a_3} \bra{a_4} U O U^{\dagger} \ket{a_4} \\
    =& \frac{1}{N_M^4} \left( \sum\limits_{\Delta(i,j;i',j')=4} + \sum\limits_{\Delta(i,j;i',j')=3} + \sum\limits_{\Delta(i,j;i',j')=2} \right)  \underset{U}{\bE} \underset{\bs}{\bE} \left[ \Tr\left[(\ketbra{s_i} \otimes \ketbra{s_j} \otimes U O U^{\dagger}) X_3 \right] \cdot \Tr\left[(\ketbra{s_i'} \otimes \ketbra{s_j'} \otimes U O U^{\dagger}) X_3 \right] \right] \\
    \lesssim& \frac{1}{N_M^4} \left\{ 
        \binom{N_M}{4} \underset{U}{\bE} \left[ \Tr\left[(\rho_U \otimes \rho_U \otimes O_U) X_3 \right]^2 \right] \right.+ \binom{N_M}{3} \underset{U}{\bE} \underset{s_1}{\bE} \left[ \Tr\left[(\ketbra{s_1} \otimes \rho_U \otimes O_U) X_3 \right]^2 \right] \\
    \quad& \left. + \binom{N_M}{2} \underset{U}{\bE}\underset{s_1,s_2}{\bE}  \left[ \Tr\left[(\ketbra{s_1} \otimes \ketbra{s_2} \otimes O_U) X_3 \right]^2 \right] \right\} \\
    \lesssim& \Tr\left[\Phi_6(X_3 \otimes X_3)(\rho_1 \otimes \rho_2 \otimes O_3 \otimes \rho_4 \otimes \rho_5 \otimes O_6) \right]  + \frac{1}{N_M} \Tr\left[\Phi_5(X^{(1,2,3),(1,4,5)}_3)(\rho_1 \otimes \rho_2 \otimes O_3 \otimes \rho_4 \otimes O_5) \right] \\
    \quad& + \frac{1}{N_M^2} \Tr\left[\Phi_4(X^{(1,2,3),(1,2,4)}_3)(\rho_1 \otimes \rho_2 \otimes O_3 \otimes O_4) \right] \\
    \eqqcolon& \Gamma_6 + \frac{1}{N_M} \Gamma_5 + \frac{1}{N_M^2} \Gamma_4.
\end{split}
\end{equation}
Here $\rho_U = U \rho U^{\dagger}$, $O_U = U O U^{\dagger}$, and the coefficient matrices are defined as
\begin{equation}
\begin{split}
    X^{(1,2,3),(1,4,5)}_3 &= \sum\limits_{\ba \in \bZ_d^5} X_3(a_1, a_2, a_3)\, X_3(a_1, a_4, a_5)\, \ketbra{\ba}, \\
    X^{(1,2,3),(1,2,4)}_3 &= \sum\limits_{\ba \in \bZ_d^4} X_3(a_1, a_2, a_3)\, X_3(a_1, a_2, a_4)\, \ketbra{\ba}.
\end{split}
\end{equation}
To analyze the terms $\Gamma_i$, we use the Weingarten calculus and retain only the leading-order contributions in the large-$d$ limit ($d \gg 1$). 
The Weingarten coefficient $\text{Wg}(\alpha, d)$ in Eq.~\eqref{eq:weingarten} admits the asymptotic expansion
\begin{equation}
    \text{Wg}(\alpha, d) = d^{k(\alpha) - 2t} \prod_{i=1}^{k(\alpha)} (-1)^{\xi_i - 1} C_{\xi_i - 1} + \cO(d^{k(\alpha) - 2t - 2}) = \cO(d^{k(\alpha) - 2t}),
\end{equation}
where
\begin{itemize}
    \item $k(\alpha)$ is the number of cycles in permutation $\alpha$,
    \item $(\xi_1, \dots, \xi_{k(\alpha)})$ is the cycle type of $\alpha$,
    \item $C_k = \frac{(2k)!}{k!(k+1)!}$ is the $k$-th Catalan number,
\end{itemize}
as given in~\cite[Eq.~(141)]{zhou2020Single}. 
Since $k(\alpha) \le t$ for all $\alpha \in S_t$, it follows that
\begin{equation}
    \text{Wg}(\alpha, d) = \cO(d^{-t}).
\end{equation}
For $\pi \in S_t$ and input $Q = \rho^{\otimes t} O^{\otimes 2}$, where $2 \le t \le 4$, we claim that
\begin{equation} \label{eq:W_sigma_X}
    \Tr[\mathbb{P}^t_\pi Q] \le \Tr(O^2).
\end{equation}
This holds because any contraction involving $\rho$ and $O$ (e.g., terms like $\Tr(\rho^a O \rho^b O)$ for $a,b\in\mathbb{N}_+$) can be bounded by $\norm{O^2}_\infty \le \Tr(O^2)$ using similar analysis as in the proof of Lemma~\ref{lem:O2}. Combining these observations and 
using Eq.~\eqref{eq:t-th_order_twirling}, we have that for constant $t$ and an arbitrary matrix $X$,
\begin{equation} \label{eq:Q_twirling_X_final}
\begin{split}
    \Tr[\Phi_t(X) Q] &=  \sum_{\pi,\sigma \in S_t} \mathrm{Wg}(\pi^{-1}\sigma,d) \Tr[\mathbb{P}^t_{\sigma^{-1}} X] \Tr[\mathbb{P}^t_{\pi} Q].
    \\ &= \cO\bigl(d^{-t} \Tr(O^2)\bigr) \sum\limits_{\pi \in S_t} \Tr[\mathbb{P}^t_\pi X].
\end{split}
\end{equation}

From Eq.~(150) in~\cite{zhou2020Single}, for any permutation $\pi \in S_t$, the following scaling bounds hold:
\begin{equation}
\begin{split}
    \Tr[\mathbb{P}^t_\pi X_3^{\otimes 2}] &= \cO(d^6), \\
    \Tr[\mathbb{P}^t_\pi X^{(1,2,3),(1,4,5)}_3] &= \cO(d^5), \\
    \Tr[\mathbb{P}^t_\pi X^{(1,2,3),(1,2,4)}_3] &= \cO(d^5).
\end{split}
\end{equation}
Substituting these into Eq.~\eqref{eq:Q_twirling_X_final}, we obtain the following bounds for the $\Gamma$ terms in Eq.~\eqref{eq:variance}:
\begin{equation}
\begin{split}
    \Gamma_6 &= \cO\bigl(d^{-6} \Tr(O^2)\bigr) \cdot \cO(d^6)  = \cO(\Tr(O^2)), \\
    \Gamma_5 &= \cO\bigl(d^{-5} \Tr(O^2) \bigr) \cdot \cO(d^5)  = \cO(\Tr(O^2)), \\
    \Gamma_4 &= \cO\bigl(d^{-4}\Tr(O^2)\bigr) \cdot \cO(d^5) = \cO(d \Tr(O^2)).
\end{split}
\end{equation}
Thus, we conclude that
\begin{equation}
\begin{split}
    \Var[\omega_u] 
    &\lesssim \Gamma_6 + \frac{1}{N_M} \Gamma_5 + \frac{1}{N_M^2} \Gamma_4 \\
    &\lesssim \left(1 + \frac{1}{N_M} + \frac{d}{N_M^2} \right) \Tr(O^2) \\
    &\lesssim \left(1 + \frac{d}{N_M^2} \right) \Tr(O^2),
\end{split}
\end{equation}
where the last line uses the fact that $1/N_M \le 1$. This completes the proof of Lemma~\ref{lem:RRM_variance}.
\end{proof}

\end{document}